\documentclass[11pt]{article}

\usepackage[margin=1in]{geometry}

\usepackage{amssymb,amsmath,amsthm,amsfonts}
\usepackage{mathtools}
\usepackage{enumitem}
\usepackage[numbers,comma,sort&compress]{natbib}
\usepackage{authblk}
\usepackage{graphicx}
\usepackage[font=small]{caption}
\usepackage[labelformat=simple]{subcaption}
\usepackage{float}
\usepackage[ruled,vlined,linesnumbered]{algorithm2e}
\usepackage{footnote}

\input{Qcircuit}

\usepackage[colorlinks]{hyperref}
\PassOptionsToPackage{colorlinks}{hyperref}

\makesavenoteenv{tabular}

\mathtoolsset{centercolon}

\DeclarePairedDelimiterX\braket[2]{\langle}{\rangle}{#1 \delimsize\vert #2}

\def\>{\rangle}
\def\<{\langle}

\newcommand{\nn}{\nonumber\\}

\newtheorem{theorem}{Theorem}
\newtheorem{definition}{Definition}
\newtheorem{lemma}{Lemma}
\newtheorem{proposition}{Proposition}
\newtheorem{example}{Example}
\newtheorem{corollary}{Corollary}

\newcommand{\eqn}[1]{(\ref{eqn:#1})}
\newcommand{\thm}[1]{\hyperref[thm:#1]{Theorem~\ref*{thm:#1}}}
\newcommand{\defn}[1]{\hyperref[defn:#1]{Definition~\ref*{defn:#1}}}
\newcommand{\lem}[1]{\hyperref[lem:#1]{Lemma~\ref*{lem:#1}}}
\newcommand{\prop}[1]{\hyperref[prop:#1]{Proposition~\ref*{prop:#1}}}
\newcommand{\fig}[1]{\hyperref[fig:#1]{Figure~\ref*{fig:#1}}}
\newcommand{\tab}[1]{\hyperref[tab:#1]{Table~\ref*{tab:#1}}}
\newcommand{\algo}[1]{\hyperref[algo:#1]{Algorithm~\ref*{algo:#1}}}
\renewcommand{\sec}[1]{\hyperref[sec:#1]{Section~\ref*{sec:#1}}}
\newcommand{\append}[1]{\hyperref[append:#1]{Appendix~\ref*{append:#1}}}

\newcommand{\N}{\mathbb{N}}

\newcommand{\R}{\mathbb{R}}
\newcommand{\C}{\mathbb{C}}

\newcommand{\stine}{sparse Stinespring isometry}
\newcommand{\singly}{identical-coordinate}

\renewcommand{\d}{\mathrm{d}}
\renewcommand{\emptyset}{\varnothing}
\newcommand{\range}[1]{[#1]}

\def\Tr{\operatorname{Tr}}

\def\:{\hbox{\bf:}}

\DeclareMathOperator{\poly}{poly}

\DeclareMathOperator{\spn}{span}
\DeclareMathOperator{\anc}{anc}
\DeclareMathOperator{\diag}{diag}

\newcommand{\linop}[1]{\mathsf{#1}}

\DeclareMathOperator{\num}{num}
\DeclareMathOperator{\Next}{Next}
\DeclareMathOperator{\Prev}{Prev}
\DeclareMathOperator{\NextIdx}{NextIdx}
\DeclareMathOperator{\PrevIdx}{PrevIdx}

\begin{document}

\title{Efficient simulation of sparse Markovian quantum dynamics}

\author{Andrew M.\ Childs and Tongyang Li \\
\small{Department of Computer Science, Institute for Advanced Computer Studies, and \\ Joint Center for Quantum Information and Computer Science, University of Maryland}}

\date{}

\maketitle


\begin{abstract}
Quantum algorithms for simulating Hamiltonian dynamics have been extensively developed, but there has been much less work on quantum algorithms for simulating the dynamics of open quantum systems. We give the first efficient quantum algorithms for simulating Markovian quantum dynamics generated by Lindbladians that are not necessarily local. We introduce two approaches to simulating sparse Lindbladians. First, we show how to simulate Lindbladians that act within small invariant subspaces using a quantum algorithm to implement sparse Stinespring isometries. Second, we develop a method for simulating sparse Lindblad operators by concatenating a sequence of short-time evolutions. We also show limitations on Lindbladian simulation by proving a no--fast-forwarding theorem for simulating sparse Lindbladians in black-box models.
\end{abstract}

\section{Introduction}
\label{sec:intro}

The original motivation for quantum computers came from the observation that such a device would be ideally suited to simulating quantum systems \cite{feynman1982simulating}.  Over the past two decades, there has been substantial work on the development of quantum algorithms for simulating Hamiltonian dynamics.  Lloyd \cite{lloyd1996universal} gave the first explicit quantum algorithm for efficiently simulating local Hamiltonians.  Aharonov and Ta-Shma \cite{aharonov2003adiabatic} introduced the more general notion of sparse Hamiltonians and showed that they can also be simulated efficiently. The complexity of sparse Hamiltonian simulation was subsequently improved using approaches based on product formulas \cite{Chi04,berry2007efficient,childs2010simulating}, discrete-time quantum walks \cite{Chi10,BC12}, and methods for implementing linear combinations of unitaries \cite{childs2012hamiltonian,berry2014exponential,berry2015simulating,berry2015hamiltonian,low2016optimal,BN16}.

Hamiltonian dynamics represent an idealized scenario in which the system is perfectly isolated. More generally, a quantum system coupled to an inaccessible environment can evolve nonunitarily.  Such open quantum systems arise naturally in areas including quantum statistical mechanics and quantum optics, and in the description of realistic quantum information processors that are subject to noise.  However, there has been relatively little work on quantum algorithms for simulating open quantum systems.

In this paper, we focus on quantum algorithms for simulating Markovian quantum dynamics, a well-studied special case that describes a situation in which a system is coupled to a large, memoryless environment.  For an $N$-dimensional system with density matrix $\rho$, such dynamics can be described by a Lindblad equation \cite{lindblad1976generators} of the form
\begin{align}\label{eqn:Lindblad}
\frac{\d{\rho}}{\d{t}} =
\mathcal{L}(\rho)=-i[H,\rho]+\sum_{j}\big(L_{j}\rho L_{j}^{\dagger}-\frac{1}{2}(L_{j}^{\dagger}L_{j}\rho+\rho L_{j}^{\dagger}L_{j})\big),
\end{align}
where $H$ is an $N \times N$ Hermitian matrix and the $L_{j}$ are $N \times N$ matrices called \emph{Lindblad operators}.  The superoperator $\mathcal{L}$ that generates the dynamics is called the \emph{Lindbladian}. We say that a Lindbladian $\mathcal{L}$ can be efficiently simulated if, for any $t,\epsilon>0$, there exists a quantum operation $\mathcal{E}$ consisting of $\poly(n,t,1/\epsilon)$ gates such that $\|\mathcal{E}-e^{\mathcal{L}t}\|_{\diamond}<\epsilon$, where $\|\cdot\|_{\diamond}$ is the diamond norm (see \defn{diamond-norm} below).

Of course, one possible approach to simulating open quantum systems is to explicitly simulate the environment \cite{terhal2000problem,wang2011quantum}.  However, this introduces considerable overhead, and it is challenging to analyze such an algorithm rigorously. To the best of our knowledge, the only case for which explicit, efficient algorithms have been presented is the setting of open quantum systems with local interactions (i.e., with Lindblad operators that act nontrivially on a constant number of qubits). Kliesch et al.\ \cite{kliesch2011dissipative} gave the first efficient algorithm for this case. That work has been extended to some non-Markovian open systems \cite{sweke2016digital}.

More generally, it may be useful to have efficient algorithms for simulating Lindbladians that are not necessarily local. Such a framework might be applied to develop quantum algorithms based on Markovian dynamics, just as sparse Hamiltonian simulation can be used to implement adiabatic optimization \cite{aharonov2003adiabatic}, quantum walk algorithms \cite{CCDFGS03,FGG07}, and the quantum linear systems algorithm \cite{HHL09}.  It might also be useful for simulating realistic systems that are not necessarily described by local Lindbladians, just as sparse Hamiltonian simulation has been useful in the context of quantum chemistry (see for example \cite{babbush2015exponentially} and references therein).  (We give a simple example of such a Lindbladian, a truncated version of a damped harmonic oscillator, in \sec{dho}.)

Just as in Hamiltonian simulation, a simple counting argument shows that Lindbladians must have some special structure to be efficiently simulated. In particular, it would be natural to develop an analog of sparse Hamiltonian simulation for Markovian dynamics. In this paper, we develop two approaches to simulating sparse Lindbladian dynamics.

First, we develop a simulation framework that we call the \stine{} framework. As in sparse Hamiltonian simulation, we divide the Lindbladian into a sum of terms, each of which generates evolution within a low-dimensional subspace. To simulate one such term, we implement it by a sparse Stinespring isometry, which we show how to implement efficiently in \thm{sparse-Stinespring-isometry}. The quantum algorithm first implements the isometry using an ancilla system and then uncomputes this ancilla to obtain the correct isometry. While the first step is reasonably straightforward, the second is more technically challenging, requiring careful application of the orthogonality properties of the Stinespring isometry to uncompute the ancilla without disturbing the effect of the first step.

Second, we develop a method for simulating Lindbladians with sparse Lindblad operators (\thm{k-sparse-generator}). We simulate such Lindbladians by concatenating a sequence of short-time evolutions, where the error in the implementation of each piece is quadratic in the evolution time. By generalizing efficient implementation of sparse unitaries \cite{jordan2009efficient}, we implement these pieces by approximately implementing a sparse map that we show is close to the desired quantum operation.

We apply these two methods to simulate five classes of sparse Lindbladians, which we characterize in terms of a matrix that we call the \emph{overcomplete GKS matrix}, denoted $A$. When the Lindbladian acts on an $N$-dimensional system, $A$ is an $N^{2}\times N^{2}$ matrix.

These simulations are summarized in \tab{main-result}.  We call the first class \emph{identical-coordinate Lindbladians} because the nonzero entries of the overcomplete GKS matrices have two identical coordinates for the row, and similarly for the column. In particular, if all Lindblad operators of a Lindbladian are diagonal, then this Lindbladian is identical-coordinate. In the second class, \emph{sparse-diagonal Lindbladians}, the overcomplete GKS matrix $A$ is diagonal, and the diagonal of $A$ is described by a $d$-sparse matrix. If the matrix $A$ is diagonal but its diagonal is not described by a sparse matrix, then it may be difficult to simulate in general. However, we define a class of \emph{dense-diagonal Lindbladians} for which the entries of the diagonal can be partitioned into $N$ sets of identical entries, such that with appropriate query access to partial sums of the entries, these Lindbladians can be simulated efficiently. The fourth class of efficiently simulatable Lindbladians, which we call \emph{1-ket-sparse Lindbladians}, has nonzero off-diagonal entries, but the sparsity is special as the permutation is independent of the second coordinate of both rows and columns. In the fifth class of Lindbladians, the Lindblad operators are sparse.

More generally, our results show how to efficiently simulate any Lindbladian that can be expressed as a positive linear combination of (efficient unitary transformations of) Lindbladians from the five classes represented in \tab{main-result},  together with local Lindbladians.

\begin{table}
\centering
\begin{tabular}{|c|c|c|c|c|}
\hline
class & nonzero entries & \# nonzeros & rank & result \\ \hline\hline
Identical-coordinate & $A_{(k,k),(l,l)}$ & $\Theta(N^{2})$ & $\Theta(N)$ & \thm{singly-sparse} \\ \hline
Sparse-diagonal & $A_{(k,l),(k,l)}$ & $\Theta(dN)$ & $\Theta(dN)$ & \thm{d-sparse} \\ \hline
Dense-diagonal & $A_{(k,l),(k,l)}$ & $N^{2}$ & $N^{2}$ & \thm{equal-diagonal} \\ \hline
1-ket-sparse & $A_{(k,l),(k,l)},A_{(k,l),(\nu(k),l)}$ & $\Theta(N)$ & $\Theta(N)$ & \thm{1-k-sparse} \\ \hline
Sparse Lindblad operator & $A_{(\nu_{i}(k),k),(\nu_{j}(l),l)}$ & $\Theta(N^{2})$ & $O(1)$ & \thm{k-sparse-generator} \\ \hline
\end{tabular}
\caption{Summary of the main results of this paper. The second column describes the nonzero entries of the matrix $A$ characterizing the overcomplete GKS form presented in \eqn{GKS'}.  Here $\nu$ is a permutation that encodes the location of the nonzero entries.  The third and fourth columns describe the number of nonzero entries and the rank of $A$, respectively, where $N$ is the dimension of the Hilbert space and $d$ is the sparsity.}
\label{tab:main-result}
\end{table}

We describe two applications of our results. First, we show how to efficiently simulate a truncated damped quantum harmonic oscillator, which is not described by a local Lindbladian. Second, we give an efficient implementation of open quantum walks---a mutual generalization of classical Markov chains and quantum walks---on sparse graphs.

Finally, we also show a limitation on Lindbladian simulation by proving a no--fast-forwarding theorem for Markovian quantum dynamics (\thm{no-fast-forward}). For two natural black-box models of sparse Lindbladians, we show that the complexity of a simulation for time $t$ is at least linear in $t$.

The remainder of this paper is organized as follows. In \sec{prelim}, we define the overcomplete GKS form, invariantly $d$-sparse Lindbladians, and $d$-sparse Stinespring isometries, and establish an equivalence between the latter two notions. We also discuss product formulas for superoperators. Then, in \sec{GKS'}, we establish our \stine{} framework, which reduces the Lindbladian simulation problem to the task of finding the Gram vectors of a certain matrix and efficiently implementing a sparse Stinespring isometry defined in terms of these Gram vectors. In \sec{main}, we apply this framework to efficiently simulate four classes of nonlocal Lindbladians. Then in \sec{short-time}, we show how to simulate Lindbladians with sparse Lindblad operators. In \sec{open-quantum-walk}, we describe how to use our results to efficiently simulate open quantum walks. We then turn to limitations on simulation in \sec{no-fast-forwarding}, where we prove a no--fast-forwarding theorem for Lindbladians. Finally, we conclude in \sec{conclusion} and present some open problems about simulating open quantum systems.

\section{Preliminaries}
\label{sec:prelim}

\subsection{Overcomplete GKS form}

Gorini, Kossakowski, and Sudarshan \cite{gorini1976completely} gave an equivalent characterization of Markovian quantum dynamics that is a convenient alternative to the Lindblad equation \eqn{Lindblad}. Given a basis $\{\sigma_{i}\}_{i=1}^{N^{2}-1}$ of traceless operators on $N\times N$ density matrices, $\mathcal{L}$ is the generator of a Markovian open system evolution if and only if it can written as
\begin{align}\label{eqn:GKS}
\mathcal{L}(\rho)=-i[H,\rho]+\sum_{i,j=1}^{N^{2}-1}A_{ij}\big([\sigma_{i}\rho,\sigma_{j}^{\dagger}]+[\sigma_{i},\rho\sigma_{j}^{\dagger}]\big)
\end{align}
for some $N \times N$ Hermitian matrix $H$ and some $(N^2-1)\times(N^2-1)$ positive semidefinite matrix $A$.

It is more convenient for our purposes to choose a basis of matrix elements in the computational basis. Although the resulting form is no longer unique, it provides a natural setting to describe sparse Lindbladians. The resulting \emph{overcomplete GKS form} is as follows.  Here and throughout this paper, let $\range{N} := \{0,1,\ldots,N-1\}$.

\begin{theorem}\label{thm:GKS'}
$\mathcal{L}$ is a Lindbladian if and only if there exists an $N^{2}\times N^{2}$ positive semidefinite matrix $A$ with entries $A_{(k,l),(k',l')}$ for $k,l,k',l' \in [N]$ such that
\begin{align}\label{eqn:GKS'}
\mathcal{L}(\rho)=-i[H,\rho]+\sum_{k,l,k',l' \in [N]} A_{(k,l),(k',l')}\big(2\<l|\rho|l'\>|k\>\<k'|-\delta_{kk'}|l'\>\<l|\rho-\delta_{kk'}\rho|l'\>\<l|\big).
\end{align}
\end{theorem}

\begin{proof}
For each $j$ in \eqn{Lindblad}, write the Lindblad operator $L_{j}$ in the computational basis in the form
\begin{align}\label{eqn:Lindblad-ket-bra}
L_{j}=\sum_{k,l\in\range{N}}a_{j;(k,l)}|k\>\<l|.
\end{align}
Then
\begin{align}
&\sum_{j}\big(L_{j}\rho L_{j}^{\dagger}-\frac{1}{2}(L_{j}^{\dagger}L_{j}\rho+\rho L_{j}^{\dagger}L_{j})\big) \nn
&\quad=\frac{1}{2}\sum_{j}\sum_{k,l,k',l'\in\range{N}}a_{j;(k,l)}a_{j;(k',l')}^{*}\big(2\<l|\rho|l'\>|k\>\<k'|-\delta_{kk'}(|l'\>\<l|\rho+\rho|l'\>\<l|)\big).
\end{align}

Define an $N^{2}\times N^{2}$ matrix $A$ with composite indices by
\begin{align}\label{eqn:GKS'-1}
A_{(k,l),(k',l')}:=\frac{1}{2}\sum_{j}a_{j;(k,l)}a_{j;(k',l')}^{*}.
\end{align}
Clearly, $A$ must be positive semidefinite. On the other hand, if $A$ is positive semidefinite, it can always be written in the above form using the spectral decomposition.
\end{proof}

\subsection{Sparse Stinespring isometries}

We develop a framework for simulating sparse Lindbladians by implementing sparse Stinespring isometries. We begin by defining notions of sparsity for Lindbladians and quantum operations.

\begin{definition}\label{defn:d-sparse-map}
A linear map $\mathcal{T}$ on density matrices is \emph{$d$-sparse} if for any pair $(x,y) \in \range{N}^2$,
\begin{align}\label{eqn:d-sparse-map}
\mathcal{T}(|x\>\<y|)\in \spn\{|x'\>\<y'| : x'\in S_{x},\, y'\in S_{y}\},
\end{align}
where $S_{x} \subseteq \range{N}$ is the set of basis states to which $|x\>$ can transit and $\max_{x}|S_{x}|\leq d$. If for arbitrary $x\neq y$, either $S_{x}\cap S_{y}=\emptyset$ or $S_{x}=S_{y}$, then we say $\mathcal{T}$ is \emph{invariantly $d$-sparse}.
\end{definition}

Invariantly sparse operations are especially simple since their evolution is confined to low-dimensional subspaces.
In particular, observe that for any invariantly $d$-sparse Lindbladian $\mathcal{L}$ and any time $t>0$, $\mathcal{E}=e^{\mathcal{L}t}$ is an invariantly $d$-sparse quantum operation.

In our \stine{} framework, we focus on simulating invariantly $d$-sparse Lindbladians. It is helpful to have a similar notion of invariant sparsity for Stinespring isometries.

\begin{definition}\label{defn:d-sparse-Stine}
An isometry $V$ is \emph{invariantly $d$-sparse} if for any $x \in \range{N}$,
\begin{align}\label{eqn:stinespring-defn}
V|x\>=\sum_{i \in [r_{x}]}|\nu^{i}(x)\>|\phi_{x,i}\>
\end{align}
for some (not necessarily normalized) ancilla states $|\phi_{x,i}\>$, where $\nu$ is a permutation of order at most $d$ and $r_{x}\leq d$ is the order of $x$.
\end{definition}

The assumption that $\nu$ is a permutation of order at most $d$ ensures that the quantum operation obtained by applying $V$ and tracing out the ancilla is invariantly $d$-sparse. More precisely, we have $|S_{x}|=r_{x}$ for each $x$, and we may assign a cyclic order of all $r_{x}$ elements in $S_{x}$ according to the order of their indices. Thus $\nu$ can be viewed as a ``neighbor function'' where the first neighbor of $x$ is $\nu(x)$, the second neighbor of $x$ is $\nu^{2}(x)$ (which is also the first neighbor of $\nu(x)$), etc. Then we have $S_{x}=\{x,\nu(x),\ldots,\nu^{r_{x}-1}(x)\}$, which contains all $r_{x}-1$ neighbors of $x$ together with itself.

\begin{example}
With $N=7$, suppose $S_{0}=S_{1}=S_{4}=S_{5}=\{0,1,4,5\}$ and $S_{2}=S_{3}=S_{6}=\{2,3,6\}$. Then we can take the elements of $S_{0}$ (or $S_{1},S_{4},S_{5}$) in the order $0\rightarrow 1\rightarrow 4\rightarrow 5\rightarrow 0$ and of $S_{2}$ (or $S_{3},S_{6}$) in the order $2\rightarrow 3\rightarrow 6\rightarrow 2$. In other words, we can define the permutation $\nu$ to be $\nu(0)=1,\nu(1)=4,\nu(2)=3,\nu(3)=6,\nu(4)=5,\nu(5)=0,\nu(6)=2$.
\end{example}

Our definitions of invariantly $d$-sparse quantum operations and invariantly $d$-sparse Stinespring isometries are equivalent in the following sense:

\begin{proposition}\label{prop:d-sparse-equivalence}
A quantum operation $\mathcal{E}$ is invariantly $d$-sparse if and only if there exists an invariantly $d$-sparse Stinespring isometry $V$ such that
\begin{align}\label{eqn:stinespring-q-operation}
\mathcal{E}(\rho)=\Tr_{\anc}[V\rho V^{\dagger}]
\end{align}
for all density matrices $\rho$.
\end{proposition}

\begin{proof}
The ``if'' part is trivial since a $d$-sparse Stinespring isometry can only map $|x\>\<y|$ into $\spn\{|x'\>\<y'|:x'\in S_{x},y'\in S_{y}\}$, where $S_{x}=\{x,\nu(x),\ldots,\nu^{r_{x}-1}(x)\}$ for arbitrary $x$. Also, for arbitrary $x\neq y$, either $S_{x}\cap S_{y}=\emptyset$ or $S_{x}=S_{y}$.

It remains to prove the ``only if'' part. Assume that a quantum operation $\mathcal{E}$ is $d$-sparse. Since for arbitrary $x\neq y$ either $S_{x}\cap S_{y}=\emptyset$ or $S_{x}=S_{y}$, we may define a permutation $\nu$ of order at most $d$ for any input such that $S_{x}=\{x,\nu(x),\ldots,\nu^{r_{x}-1}(x)\}$ for arbitrary $x$, where $r_{x}=|S_{x}|\leq d$ is the order of $x$. Without loss of generality, we may assume that $r_{x}=d$ for each $x$. Thus we may write
\begin{align}\label{eqn:1-sparse-quantum-operation}
\mathcal{E}(|x\>\<y|)=\sum_{i,j=0}^{d-1}a_{ij}^{xy}|\nu^{i}(x)\>\<\nu^{j}(y)|
\end{align}
for all $x,y \in \range{N}$, where $a_{ij}^{xy} \in \C$ for all $i,j\in\range{d}$ and $x,y\in\range{N}$. Since $\mathcal{E}$ is trace-preserving,
\begin{align}
\Tr\big[\mathcal{E}(|x\>\<x|)\big]
&=1 \quad \Rightarrow \quad \sum_{j=0}^{d-1}a_{jj}^{xx}=1\quad\forall\,x \in \range{N};  \label{eqn:restriction-1} \\
\Tr\big[\mathcal{E}(|x\>\<\nu^{i}(x)|)\big]
&=0 \quad \Rightarrow \quad \sum_{j=0}^{d-1}a_{i+j,j}^{x\nu^{i}(x)}=0\quad\forall\,x \in \range{N},\, i\in\{1,\ldots,d-1\}. \label{eqn:restriction-2}
\end{align}
By \eqn{1-sparse-quantum-operation}, the Choi matrix of $\mathcal{E}$, defined as $J(\mathcal{E}):=(\mathcal{E}\otimes\mathcal{I})(\sum_{x,y}|x\>\<y|\otimes|x\>\<y|)$, is
\begin{align}\label{eqn:Choi}
J(\mathcal{E})=\sum_{x,y=0}^{N-1}\big(\sum_{i,j=0}^{d-1}a_{ij}^{xy}|\nu^{i}(x)\>\<\nu^{j}(y)|\big)\otimes|x\>\<y|.
\end{align}
The Choi matrix of any quantum operation is positive semidefinite. In other words, for any state $|\varphi\>=\sum_{k,l=0}^{N-1}c_{kl}|k\>|l\>$, we have $\<\varphi|J(\mathcal{E})|\varphi\>\geq 0$. By \eqn{Choi}, this is equivalent to
\begin{align}\label{eqn:Choi-2}
\sum_{x,y=0}^{N-1}\sum_{i,j=0}^{d-1}c_{\nu^{i}(x),x}a_{ij}^{xy}c_{\nu^{j}(y),y}\geq 0.
\end{align}

Denote
\begin{align}\label{eqn:matrix-fy}
M_{\mathcal{E}}:=\sum_{x,y=0}^{N-1}\big(\sum_{i,j=0}^{d-1}a_{ij}^{xy}|i\>\<j|\big)\otimes|x\>\<y|, \quad
|c\>:=\sum_{z=0}^{N-1}\sum_{k=0}^{d-1}c_{\nu^{k}(z),z}|k\>|z\>,
\end{align}
where we call $M_{\mathcal{E}}$ the Gram matrix of $\mathcal{E}$ (see \defn{Gram} below). Equation \eqn{Choi-2} is equivalent to
\begin{align}\label{eqn:Choi-3}
\<c|M_{\mathcal{E}}|c\>\geq 0.
\end{align}
Since \eqn{Choi-2} holds for any $|\varphi\>\in\C^{N\times N}$, equation \eqn{Choi-3} also holds for any $|c\>\in\C^{dN}$. In other words, $M_{\mathcal{E}}$ is positive semidefinite.

Since every positive semidefinite matrix is a Gram matrix, $M_{\mathcal{E}}$ is a Gram matrix. In other words, there exist vectors $\{|\phi_{x,i}\>\}_{x \in \range{N},\, i \in \range{d}}$ such that for all $x,y\in\range{N}$ and $i,j\in\range{N}$,
\begin{align}\label{eqn:inner-product}
\<\phi_{y,j}|\phi_{x,i}\> &=a_{ij}^{xy}.
\end{align}
For each $x$, define $V|x\>:=\sum_{i=0}^{d-1}|\nu^{i}(x)\>|\phi_{x,i}\>$ as in \eqn{stinespring-defn}. Then by \eqn{restriction-1}, \eqn{restriction-2}, and \eqn{inner-product}, $V$ is a $d$-sparse Stinespring isometry:
\begin{align}
\|V|x\>\|^{2}&=\sum_{j=0}^{d-1}\||\phi_{x,i}\>\|^{2}=\sum_{j=0}^{d-1}a_{jj}^{xx}=1\quad\forall\,x\in\range{N}; \\
\<\nu^{i}(x)|V^{\dagger}V|x\>&=\sum_{j=0}^{d-1}\<\phi_{\nu^{i}(x),j}|\phi_{x,i+j}\>=\sum_{j=0}^{d-1}a_{i+j,j}^{x\nu^{i}(x)}=0\quad\forall\,x\in\range{N},\, i\in\{1,\ldots,d-1\}.
\end{align}
Furthermore, for arbitrary $x,y\in\range{N}$,
\begin{align}
\Tr_{\anc}[V|x\>\<y|V^{\dagger}]&=\Tr_{\anc}\big[\sum_{i,j=0}^{d-1}|\nu^{i}(x)\>|\phi_{x,i}\>\<\nu^{j}(y)|\<\phi_{y,j}|\big] \\
&=\sum_{i,j=0}^{d-1}a_{ij}^{xy}|\nu^{i}(x)\>\<\nu^{j}(y)| \\
&=\mathcal{E}(|x\>\<y|).
\end{align}
Consequently,
\begin{align}
\mathcal{E}(\rho)=\Tr_{\anc}[V\rho V^{\dagger}]
\end{align}
for all density matrices $\rho$, which completes the proof.
\end{proof}

We refer to the matrix $M_{\mathcal{E}}$ appearing in the proof of \prop{d-sparse-equivalence} as the Gram matrix of $\mathcal{E}$:
\begin{definition}\label{defn:Gram}
Given an invariantly $d$-sparse Lindbladian $\mathcal{L}$ and a time $t>0$, the Choi matrix of $\mathcal{E}=e^{\mathcal{L}t}$ can be written as
\begin{align}
J(\mathcal{E})=\sum_{x,y=0}^{N-1}\big(\sum_{i,j=0}^{d-1}a_{ij}^{xy}|\nu^{i}(x)\>\<\nu^{j}(y)|\big)\otimes|x\>\<y|.
\end{align}
Then the \emph{Gram matrix of $\mathcal{E}$} is
\begin{align}\label{eqn:M}
M_{\mathcal{E}}:=\sum_{x,y=0}^{N-1}\big(\sum_{i,j=0}^{d-1}a_{ij}^{xy}|i\>\<j|\big)\otimes|x\>\<y|.
\end{align}
\end{definition}

\subsection{Product formulas for Markovian quantum dynamics}

Product formulas are a useful tool for Hamiltonian simulation because they ensure the additivity of efficient Hamiltonian simulation: if $H_{1}$ and $H_{2}$ can be efficiently simulated, then $H_{1}+H_{2}$ can be efficiently simulated. In particular, reference \cite{berry2007efficient} has carefully analyzed the complexity of Hamiltonian simulation using high-order Suzuki product formulas \cite{suzuki1990fractal,suzuki1991general}.

Often the spectral norm is applied to characterize the magnitude of Hamiltonians or the distance between the unitary evolutions they generate.  While the diamond norm has a clearer operational meaning, the spectral norm distance suffices to characterize unitary dynamics since it is within a constant factor of the diamond norm distance \cite[Lemma 7]{berry2015hamiltonian}.  Recall that the diamond norm is defined as follows:

\begin{definition}[\cite{kitaev1997quantum}]\label{defn:diamond-norm}
Let $\mathcal{T}$ be a superoperator in $\linop{L}(\linop{L}(\C^{N}))$, where $\linop{L}(X)$ is the set of linear operators acting on the space $X$. Define the \emph{diamond norm} of $\mathcal{T}$ as
\begin{align}
\|\mathcal{T}\|_{\diamond}:=\max_{\|\rho\|_{1}=1}\|\mathcal{T}\otimes\mathcal{I}_{N}(\rho)\|_{1},
\end{align}
where $\|\cdot\|_{1}$ is the trace norm (i.e., the sum of the absolute values of the eigenvalues), $\mathcal{I}_{N}$ is the identity superoperator in $\linop{L}(\linop{L}(\C^{N}))$, and $\rho$ is an arbitrary density matrix in $\linop{L}(\C^{N})$.
\end{definition}

Error estimates for first- \cite[Theorem 1]{kliesch2011dissipative} and second-order \cite[Lemma 4]{werner2016positive} product formulas have been generalized to Markovian quantum dynamics. In the following, we provide an explicit error bound for a second-order Lindbladian product formula. 


\begin{proposition}\label{prop:trotter-formula-second}
Fix a constant $m$, let $\mathcal{L}_{1},\mathcal{L}_{2},\ldots,\mathcal{L}_{m}$ be Lindbladians, and let $L:=\max_{i\in\range{n}}\|\mathcal{L}_{i}\|_{\diamond}$. Then for any $0<\epsilon<mtL$ and any positive integer $r\geq\frac{4(mtL)^{3/2}}{\sqrt{\epsilon}}$,
\begin{align}\label{eqn:Trotter-2}
\Big\|\exp\big(t\sum_{i=1}^{m}\mathcal{L}_{i}\big)-\Big(\prod_{i=1}^{m}e^{t\mathcal{L}_{i}/2r}\prod_{j=m}^{1}e^{t\mathcal{L}_{j}/2r}\Big)^{r}\Big\|_{\diamond}\leq\epsilon.
\end{align}
\end{proposition}

\begin{proof}
Let $A:=\exp\big(\lambda\sum_{i=1}^{m}\mathcal{L}_{i}\big)$ and $B:=\prod_{i=1}^{m}e^{\lambda\mathcal{L}_{i}/2}\prod_{j=m}^{1}e^{\lambda\mathcal{L}_{j}/2}$ with $\lambda:=t/r$.
Since quantum operations have unit diamond norm, we have $\|e^{t\mathcal{L}_{i}}\|_{\diamond}=1$ for any $i\in[m]$ and $t\geq 0$. In addition, since quantum operations are closed under composition, $\|A\|_{\diamond}=\|B\|_{\diamond}=1$. Therefore,
\begin{align}
\|A^{r}-B^{r}\|_{\diamond}&=\|A^{r}-A^{r-1}B+A^{r-1}B-A^{r-2}B^{2}+\cdots+AB^{r-1}-B^{r}\|_{\diamond} \\
&\leq\|A^{r-1}(A-B)\|_{\diamond}+\|A^{r-2}(A-B)B\|_{\diamond}+\cdots+\|(A-B)B^{r-1}\|_{\diamond} \\
&\leq r\|A-B\|_{\diamond}. \label{eqn:trotter-2-fourth}
\end{align}
This motivates us to focus on bounding $\|A-B\|_{\diamond}$.

We now confirm that $B$ indeed provides a second-order approximation to $A$. Both 
expressions have the identity as the zeroth-order term and $\lambda\sum_{i=1}^{m}\mathcal{L}_{i}$ as the first-order term. The second-order term of $A$ is
\begin{align}\label{eqn:trotter-2-first}
\frac{\lambda^{2}}{2}\left(\sum_{i=1}^{m}\mathcal{L}_{i}\right)^{2}=\frac{\lambda^{2}}{2}\left(\sum_{i=1}^{m}\mathcal{L}_{i}^{2}+\sum_{1\leq i<j\leq m}(\mathcal{L}_{i}\mathcal{L}_{j}+\mathcal{L}_{j}\mathcal{L}_{i})\right),
\end{align}
while the second-order term of $B$ is
\begin{align}\label{eqn:trotter-2-approx}
2\sum_{i=1}^{m}\frac{(\lambda\mathcal{L}_{i}/2)^{2}}{2} + \sum_{i=1}^{m} (\lambda\mathcal{L}_{i}/2)^{2}+4\sum_{i\neq j} (\lambda\mathcal{L}_{i}/2)(\lambda\mathcal{L}_{j}/2),
\end{align}
which is equal to \eqn{trotter-2-first}.
Here the first term of \eqn{trotter-2-approx} comes from the second-order terms in the Taylor expansions of the two factors of $e^{\lambda\mathcal{L}_{i}/2}$ (and the zeroth-order contributions of other factors) for each $i$, the second term comes from multiplying the first-order term of the Taylor expansion of $e^{\lambda\mathcal{L}_{i}/2}$ in the forward product with the same term in the backward product (again for each $i$), and the third term comes from multiplying a first-order term of a Taylor expansion of $e^{\lambda\mathcal{L}_{i}/2}$ with another first-order term of a Taylor expansion of $e^{\lambda\mathcal{L}_{j}/2}$ where $i\neq j$ (noting that each $\{i,j\}$ pair appears $4$ times, since each index can come from either the forward or the backward product).

Since $A-B$ only has terms of order three and higher in $t$, we can write
\begin{align}\label{eqn:trotter-2-second}
A-B=\exp\big(\lambda\sum_{i=1}^{m}\mathcal{L}_{i}\big)-\prod_{i=1}^{m}e^{\lambda\mathcal{L}_{i}/2}\prod_{j=m}^{1}e^{\lambda\mathcal{L}_{j}/2}=\sum_{k=3}^{\infty}\lambda^{k}\sum_{l=1}^{L_{k}}C_{k,l}\prod_{q=1}^{k}\mathcal{L}_{j_{l_q}}
\end{align}
for some constants $C_{k,l}$ and numbers of terms $L_{k}$ that depend on $m$ and $k$, where $\prod_{q=1}^{k}\mathcal{L}_{j_{l_q}}$ represents the $k$th-order terms in the Taylor series expansion. Specifically, for the first term on the left-hand side of \eqn{trotter-2-second}, expanding $(\mathcal{L}_{1}+\cdots+\mathcal{L}_{m})^{k}$ yields $m^{k}$ terms (with a prefactor of $\lambda^{k}$). On the other hand, the expansion of $\prod_{i=1}^{m}e^{\lambda\mathcal{L}_{i}/2}\prod_{j=m}^{1}e^{\lambda\mathcal{L}_{j}/2}$ can be obtained by expanding each individual exponential, and the number of terms with prefactor $\lambda^{k}$ is equivalent to the number of non-negative integer tuples $(\alpha_1,\ldots,\alpha_{2m})$ such that $\alpha_1+\cdots+\alpha_{2m}=k$. There are $\binom{k+2m-1}{2m-1}$ such tuples, and since $k\geq 3$, we have
\begin{align}
\binom{k+2m-1}{2m-1}=\binom{k+2m-1}{k}\leq 2m/1 \cdot (2m+1)/2 \cdot ((2m+2)/3)^{k-2}\leq (2^{k}-1)m^{k}.
\end{align}
Therefore, at most $(2^{k}-1)m^{k}$ terms with coefficient $\lambda^{k}$ contribute to the second term of the left-hand side of \eqn{trotter-2-second}. In total, there are at most $(2m)^{k}$ terms with coefficient $\lambda^{k}$. Recalling that $L:=\max_{i\in\range{n}}\|\mathcal{L}_{i}\|_{\diamond}$, we have
\begin{align}
\left\|\sum_{k=3}^{\infty}\lambda^{k}\sum_{l=1}^{L_{k}}C_{k,l}\prod_{q=1}^{k}\mathcal{L}_{j_{l_q}}\right\|_{\diamond}&\leq\sum_{k=3}^{\infty}(\lambda L)^{k}L_{k}\leq\sum_{k=3}^{\infty}(\lambda L)^{k}(2m)^{k}=\frac{(2m\lambda L)^{3}}{1-2m\lambda L}\leq 16m^{3}\lambda^{3}L^{3}, \label{eqn:trotter-2-third}
\end{align}
assuming $2m\lambda L\leq 1/2$.

Combining \eqn{trotter-2-fourth}, \eqn{trotter-2-second}, and \eqn{trotter-2-third}, and recalling that $\lambda=t/r$, we have
\begin{align}\label{eqn:Trotter-2-fifth}
\Big\|\exp\big(t\sum_{i=1}^{m}\mathcal{L}_{i}\big)-\Big(\prod_{i=1}^{m}e^{t\mathcal{L}_{i}/2r}\prod_{j=m}^{1}e^{t\mathcal{L}_{j}/2r}\Big)^{r}\Big\|_{\diamond}\leq 16m^{3}t^{3}L^{3}/r^{2}.
\end{align}
To obtain the desired conclusion, it suffices to take $16m^{3}t^{3}L^{3}/r^{2}\leq\epsilon$, i.e., $r\geq\frac{4(mtL)^{3/2}}{\sqrt{\epsilon}}$. In addition, note that $2m\lambda L\leq 1/2$ is equivalent to $r\geq 4mtL$, which holds provided $\epsilon\leq mtL$.
\end{proof}

Since the complexity of Hamiltonian simulation can be reduced using high-order product formulas, it is natural to consider an analogous strategy for Lindbladians. However, because Lindbladian dynamics are irreversible, this is only possible if the coefficients in the product formula are positive. Suzuki proved that product formulas of order $3$ and higher must include negative coefficients \cite[Theorem 3]{suzuki1991general}, so we cannot simulate Lindbladians using higher-order product formulas.

\prop{trotter-formula-second} shows that efficient Lindbladian simulation is closed under positive linear combination. In other words, if $\mathcal{L}_{1},\mathcal{L}_{2},\ldots,\mathcal{L}_{m}$ are (polynomially many) Lindbladians that can be efficiently simulated, then $\sum_{i=1}^{m}\mathcal{L}_{i}$ can be efficiently simulated. Note that unlike the case of Hamiltonian simulation, we are restricted to positive linear combinations since Lindbladian dynamics are not in general invertible.

If we can efficiently simulate a Hamiltonian $H$ with corresponding Lindbladian $\mathcal{H}(\rho) = -i[H,\rho]$ and a Lindbladian $\mathcal{L}$ with zero Hamiltonian, then we can efficiently simulate $\mathcal{H+L}$ by \prop{trotter-formula-second}. Thus we assume in the rest of the paper that the Hamiltonian parts of \eqn{Lindblad}, \eqn{GKS}, and \eqn{GKS'} are zero.

\section{Sparse Stinespring isometry framework}
\label{sec:GKS'}

\subsection{A brief introduction}

We now introduce the \stine{} framework, which can efficiently simulate many classes of Lindbladians that are invariantly sparse in the sense of \defn{d-sparse-map}.

As in sparse Hamiltonian simulation, the \stine{} framework expresses the Lindbladian as a sum of terms, each of which generates evolution within a low-dimensional subspace. To simulate one such Lindbladian $\mathcal{L}$ acting for time $t>0$, first we explicitly compute (in superposition) a classical description of the action of the quantum operation $e^{\mathcal{L}t}$ on that low-dimensional subspace. Specifically, we compute the coefficients of the Gram matrix $M$ of $e^{\mathcal{L}t}$ as defined in \defn{Gram}.

Second, we decompose the Gram matrix $M$ into Gram vectors within low-dimensional subspaces. By \prop{d-sparse-equivalence}, the Gram vectors of $M$ are the ancilla states $|\phi_{x,i}\>/\||\phi_{x,i}\>\|$ of a sparse Stinespring isometry $V$ that implements the quantum operation $e^{\mathcal{L}t}$.

Finally, we efficiently implement $V$ by \thm{sparse-Stinespring-isometry}. The implementation has two stages. First we efficiently implement the isometry using an ancilla system; then we efficiently uncompute the ancilla to obtain the correct isometry. As discussed in the introduction, the second step is more technically challenging, as it requires careful application of the orthogonality properties of the Stinespring isometry to uncompute the ancilla without disturbing the effect of the first stage.

The main difficulty in applying the \stine{} framework is to efficiently decompose $M$ into its Gram vectors and to efficiently prepare each Gram vector.  We demonstrate how to do this in several cases (namely, local Lindbladians and the first four classes presented in \tab{main-result}).
\begin{itemize}
\item In \thm{local-Lindbladians} we show that for a local Lindbladian, this can be done because $M$ has low rank. We present the details in \algo{low-rank-Gram} in \append{low-rank-Gram}.
\item In \thm{singly-sparse} we show that for an \singly{} Lindbladian, this can be done since all matrix elements of the density matrix are invariant under $\mathcal{L}$, i.e., $\mathcal{L}(|x\>\<y|)\propto|x\>\<y|$.
\item In \lem{1-sparse-a} we show that for a strongly 1-sparse-diagonal Lindbladian with neighbor function $\nu$, this can be done since $\{|u\>\<u|,|\nu(u)\>\<\nu(u)|\}$ is a two-dimensional invariant subspace for any $u$, and all off-diagonal terms are invariant.  This facilitates the simulation of sparse-diagonal Lindbladians in \thm{d-sparse}.
\item In \thm{equal-diagonal} we show that for a certain type of dense-diagonal Lindbladian defined by a set of coefficients $\{a_{k}\}$, this can be done since $\{|u\>\<u|,\frac{\sum a_{k}|k\>\<k|}{\sum a_{k}}\}$ is a two-dimensional invariant subspace for any $u$.
\item In \thm{1-k-sparse} we show that for a 1-ket-sparse Lindbladian with neighbor function $\nu$, this can be done since $\{|u\>\<u|,|u\>\<\nu(u)|,|\nu(u)\>\<u|,|\nu(u)\>\<\nu(u)|\}$ is a four-dimensional invariant subspace for any $u$, and all other terms are invariant.
\end{itemize}

\subsection{Efficient implementation of $d$-sparse Stinespring isometries}

The following theorem gives a quantum algorithm to implement a sparse Stinespring isometry given the ability to efficiently prepare its Gram vectors:
\begin{theorem}\label{thm:sparse-Stinespring-isometry}
Let $\mathcal{E}$ be an invariantly $d$-sparse quantum operation as in \defn{d-sparse-map}, and denote its $d$-sparse Stinespring isometry by $V$ where $V|x\>=\sum_{i=0}^{r_{x}-1}|\nu^{i}(x)\>|\phi_{x,i}\>$ as in \defn{d-sparse-Stine}. Suppose we are given a black box that outputs $S_{x}$ on input $x$. Furthermore, suppose that $|\phi_{x,i}\>/\||\phi_{x,i}\>\|$ can be prepared by $O(c_{\phi})$ 2-qubit gates for arbitrary $x\in\range{N},\,i\in\range{r_{x}}$. Then $\mathcal{E}$ can be implemented by $O(d^{2}(\log N+c_{\phi}))$ 2-qubit gates and $O(1)$ queries to the black box.
\end{theorem}

Without loss of generality, we can assume that $r_{x}=d$ for any $x$ since the proof is identical if $r_{x}<d$.

It suffices to show that for each $x$, we can implement a unitary gate $U_{S_{x}}$ acting as
\begin{align}\label{eqn:U}
U_{S_{x}}|\nu^{k}(x)\>|0\>&=\sum_{j=0}^{d-1}|\nu^{j+k}(x)\>|\phi_{\nu^{k}(x),j}\> \quad\forall\,k\in\range{d}
\end{align}
using $O(d^{2}(\log N+c_{\phi}))$ 2-qubit gates. Then we can implement the following procedure:
\begin{align}
|x\>|0\>&\xmapsto{\text{compute}}|x\>|0\>|S_{x}\>|U_{S_{x}}\> \\
&\xmapsto{\text{by \eqn{U}}}U_{S_{x}}|x\>|0\>|S_{x}\>|U_{S_{x}}\> \\
&\xmapsto{\text{uncompute}}U_{S_{x}}|x\>|0\>=\sum_{j=0}^{d-1}|\nu^{i}(x)\>|\phi_{x,i}\>,
\end{align}
where $|U_{S_{x}}\>$ represents a classical description of $U_{S_{x}}$ and $|S_{x}\>$ is computed/uncomputed by the black box.

We implement $U_{S_{x}}$ by a two-stage quantum circuit. First we prepare a unitary $U_{\rightarrow}$ such that for each $k\in\range{d}$,
\begin{align}\label{eqn:phase-1}
U_{\rightarrow}|0\>|\nu^{k}(x)\>|0\>=\sum_{i=0}^{d-1}|i\>|\nu^{i+k}(x)\>|\phi_{\nu^{k}(x),i}\>.
\end{align}
Then we uncompute the ancilla: we implement a unitary $U_{\leftarrow}$ such that for each $k\in\range{d}$,
\begin{align}\label{eqn:phase-2}
U_{\leftarrow}\sum_{i=0}^{d-1}|i\>|\nu^{i+k}(x)\>|\phi_{\nu^{k}(x),i}\>=\sum_{i=0}^{d-1}|0\>|\nu^{i+k}(x)\>|\phi_{\nu^{k}(x),i}\>.
\end{align}
After both steps, we have
\begin{align}\label{eqn:U'}
U_{\leftarrow}U_{\rightarrow}|0\>|\nu^{k}(x)\>|0\>=|0\>\sum_{j=0}^{d-1}|\nu^{j+k}(x)\>|\phi_{\nu^{k}(x),j}\> \quad\forall\,k\in\range{d},
\end{align}
which gives the unitary $U_{S_{x}}$ that we want to implement.

\subsection{Proof of \thm{sparse-Stinespring-isometry}}

\subsubsection{Notations}

By assumption, for arbitrary $k,i\in\range{d}$ there exists a unitary gate $\Phi_{\nu^{k}(x),i}$ satisfying
\begin{align}\label{eqn:U_x,i}
\Phi_{\nu^{k}(x),i}|0\>&=\frac{|\phi_{\nu^{k}(x),i}\>}{\||\phi_{\nu^{k}(x),i}\>\|},
\end{align}
where $\Phi_{\nu^{k}(x),i}$ can be implemented by $O(c_{\phi})$ 2-qubit gates.

In addition, for each $x\in\range{d}$ we define a unitary gate $W_{x}$ acting as
\begin{align}
W_{x}|0\>&=\sum_{i=0}^{d-1}\||\phi_{x,i}\>\| \, |i\>.
\end{align}
This gate can be implemented using $O(d)$ 2-qubit gates \cite{shende2005synthesis}.

For any $x\in\range{N}$ and any unitary gate $U$, let $\wedge_{x}(U) := |x\>\<x| \otimes U + (I-|x\>\<x|) \otimes I$ be the gate that performs $U$ on the target register when the control register is in the state $|x\>$. See \fig{building-block-1} for an example.

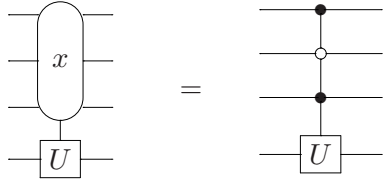
\begin{figure}[H]
\centering
\begin{subfigure}{0.5\textheight}
\centering
\Qcircuit @C=1em @R=.7em {
   & \multimeasure{2}{x} & \qw \\
   & \ghost{x} & \qw \\
   & \ghost{x} \qwx[1] & \qw \\
   & \gate{U} & \qw
}
\end{subfigure}
\qquad=\qquad
\begin{subfigure}{0.5\textheight}
\centering
\Qcircuit @C=1.4em @R=1.2em {
   & \ctrl{1} & \qw \\
   & \ctrlo{1} & \qw \\
   & \ctrl{1} & \qw \\
   & \gate{U} & \qw
}
\end{subfigure}
\caption{The gate $\wedge_{x}(U)$ with $x=101$.}
\label{fig:building-block-1}
\end{figure}

More generally, we consider unitary gates that are also controlled on a given quantum state.  For any $x \in [N]$, quantum state $|\phi\>$, and unitary gate $U$, let $\wedge_{x,|\phi\>}(U) := |x\>\<x| \otimes |\phi\>\<\phi| \otimes U + (I - |x\>\<x| \otimes |\phi\>\<\phi|) \otimes I$ be the gate that performs $U$ on the target register only if the first control register is $|x\>$ and the second control register is in the state $|\phi\>$.  An implementation of such a gate is shown in \fig{statecontrol}, where $\Phi$ is a unitary operation satisfying $\Phi|0\> = |\phi\>$.  We analogously define $\wedge_{|\phi\>}(U) := |\phi\>\<\phi| \otimes I + (I - |\phi\>\<\phi|) \otimes I$ (i.e., the application of $U$ is only controlled on the quantum state $|\phi\>$).

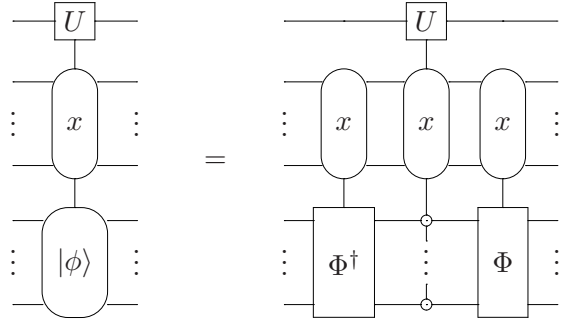
\begin{figure}[H]
\centering
\begin{subfigure}{0.5\textheight}
\centering
\Qcircuit @C=1em @R=1em {
   & \gate{U} \qwx[1] & \qw \\
   & \multimeasure{2}{x} & \qw \\
\raisebox{6pt}{\vdots} & & \raisebox{6pt}{\vdots} \\
   & \ghost{x} \qwx[1] & \qw \\
   & \multimeasure{2}{|\phi\>} & \qw \\
\raisebox{6pt}{\vdots} & & \raisebox{6pt}{\vdots} \\
   & \ghost{|\phi\>} & \qw
}
\end{subfigure}
\qquad=\qquad
\begin{subfigure}{0.5\textheight}
\centering
\Qcircuit @C=1em @R=1em {
   & \qw & \gate{U} \qwx[1] & \qw & \qw \\
   & \multimeasure{2}{x} & \multimeasure{2}{x} & \multimeasure{2}{x} & \qw \\
\raisebox{6pt}{\vdots} & & & & \raisebox{6pt}{\vdots} \\
   & \ghost{x} \qwx[1] & \ghost{x} \qwx[1] & \ghost{x} \qwx[1] & \qw \\
   & \multigate{2}{\Phi^{\dagger}} & \ctrlo{0} {\ar @{-}+<0em,-0.7em>} &  \multigate{2}{\Phi} & \qw \\
\raisebox{6pt}{\vdots} & & \raisebox{6pt}{\vdots} & & \raisebox{6pt}{\vdots} \\
   & \ghost{\Phi^{\dagger}} & \ctrlo{0} {\ar @{-}+<0em,0.7em>} & \ghost{\Phi} & \qw
}
\end{subfigure}
\caption{The gate $\wedge_{x,|\phi\>}(U)$.}
\label{fig:statecontrol}
\end{figure}

Let $P\colon \range{d} \to \range{d}$ be the permutation
\begin{align}
P|j\>:=|j-1\bmod d\> \quad\forall\, j\in\range{d}.
\end{align}
Furthermore, for arbitrary $k\in\range{d}$, define
\begin{align}
|\Phi_{x,k}\>:=\sum_{i=0}^{d-1}|\nu^{i+k}(x)\>|\phi_{\nu^{k}(x),i}\>.
\end{align}
Since $|\phi_{\nu^{k}(x),i}\>/\||\phi_{\nu^{k}(x),i}\>\|$ can be prepared by $O(c_{\phi})$ 2-qubit gates for each $i$ and $k$, the state $|\Phi_{x,k}\>$ can be prepared by $O(dc_{\phi})$ 2-qubit gates for each $k$.

Using the controlled gates defined above, for each $k\in\{1,2,\ldots,d-1\}$ we define
\begin{align}
\tilde{U}_{x,k}&:=\wedge_{\nu^{k}(x),|\phi_{\nu^{k}(x),0}\>}(P^{-k}) \, \wedge_{\nu^{k+1}(x),|\phi_{\nu^{k}(x),1}\>}(P^{-(k+1)})\cdots\wedge_{\nu^{k-1}(x),|\phi_{\nu^{k}(x),d-1}\>}(P^{-(k-1)}) \\
&\quad \wedge_{|\Phi_{x,k}\>}(P^{-k}) \, \wedge_{\nu^{k-1}(x),|\phi_{\nu^{k}(x),d-1}\>}(P^{k-1}) \cdots \wedge_{\nu^{k+1}(x),|\phi_{\nu^{k}(x),1}\>}(P^{k+1}) \, \wedge_{\nu^{k}(x),|\phi_{\nu^{k}(x),0}\>}(P^{k}), \nonumber
\end{align}
as shown in \fig{U_x,k}.
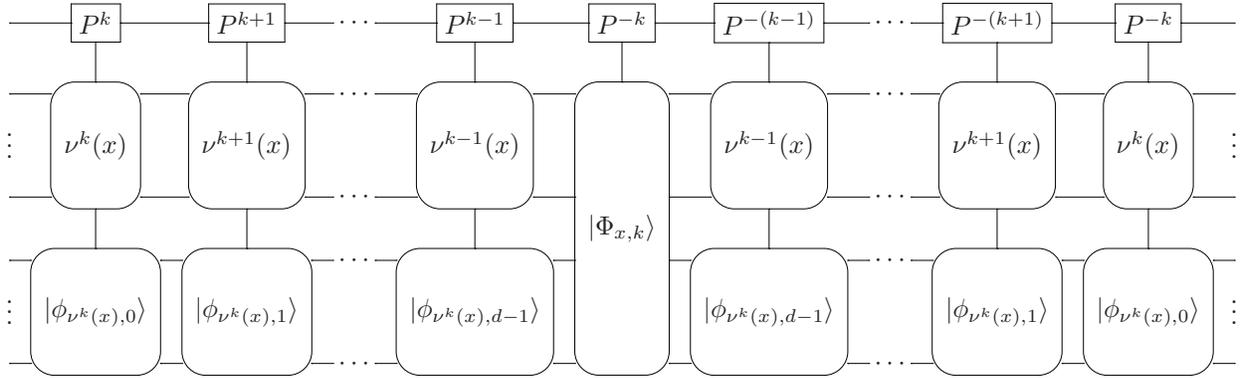
\begin{figure}[H]
{\small
  \setlength{\abovedisplayskip}{6pt}
  \setlength{\belowdisplayskip}{\abovedisplayskip}
  \setlength{\abovedisplayshortskip}{0pt}
  \setlength{\belowdisplayshortskip}{3pt}
\begin{align*}
\Qcircuit @C=0.8em @R=1.5em {
& \gate{P^{k}} \qwx[1] & \gate{P^{k+1}} \qwx[1] & \qw & \cdots & & \gate{P^{k-1}} \qwx[1] & \gate{P^{-k}} \qwx[1] & \gate{P^{-(k-1)}} \qwx[1] & \qw & \cdots & & \gate{P^{-(k+1)}} \qwx[1] & \gate{P^{-k}} \qwx[1] & \qw \\
& \multimeasure{2}{\nu^{k}(x)} & \multimeasure{2}{\nu^{k+1}(x)} & \qw &  \cdots & & \multimeasure{2}{\nu^{k-1}(x)} & \multimeasure{5}{|\Phi_{x,k}\>} & \multimeasure{2}{\nu^{k-1}(x)} & \qw & \cdots & & \multimeasure{2}{\nu^{k+1}(x)} & \multimeasure{2}{\nu^{k}(x)} & \qw \\
\raisebox{6pt}{\vdots} & & & & & & & & & & & & & & \raisebox{6pt}{\vdots} \\
& \ghost{\nu^{k}(x)} \qwx[1] & \ghost{\nu^{k+1}(x)} \qwx[1] & \qw & \cdots & & \ghost{\nu^{k-1}(x)} \qwx[1] & \ghost{|\Phi_{x,k}\>} & \ghost{\nu^{k-1}(x)} \qwx[1] & \qw & \cdots & & \ghost{\nu^{k+1}(x)} \qwx[1] & \ghost{\nu^{k}(x)} \qwx[1] & \qw \\
& \multimeasure{2}{|\phi_{\nu^{k}(x),0}\>} & \multimeasure{2}{|\phi_{\nu^{k}(x),1}\>} & \qw & \cdots & & \multimeasure{2}{|\phi_{\nu^{k}(x),d-1}\>} & \ghost{|\Phi_{x,k}\>} & \multimeasure{2}{|\phi_{\nu^{k}(x),d-1}\>} & \qw &  \cdots & & \multimeasure{2}{|\phi_{\nu^{k}(x),1}\>} & \multimeasure{2}{|\phi_{\nu^{k}(x),0}\>} & \qw \\
\raisebox{6pt}{\vdots} & & & & & & & & & & & & & & \raisebox{6pt}{\vdots} \\
& \ghost{|\phi_{\nu^{k}(x),0}\>} & \ghost{|\phi_{\nu^{k}(x),1}\>} & \qw & \cdots & & \ghost{|\phi_{\nu^{k}(x),d-1}\>} & \ghost{|\Phi_{x,k}\>} & \ghost{|\phi_{\nu^{k}(x),d-1}\>} & \qw & \cdots & & \ghost{|\phi_{\nu^{k}(x),1}\>} & \ghost{|\phi_{\nu^{k}(x),0}\>} & \qw \\
}
\end{align*}
\caption{The quantum circuit for $\tilde{U}_{x,k}$.}
\label{fig:U_x,k}
}%
\end{figure}

\subsubsection{Proof}

\begin{proof}[Proof of \thm{sparse-Stinespring-isometry}]
Define $U_{\rightarrow}$ and $U_{\leftarrow}$ to be the quantum circuits shown in \fig{phase-1} and \fig{phase-2}, respectively.
\begin{figure}[htbp]
{\small
  \setlength{\abovedisplayskip}{6pt}
  \setlength{\belowdisplayskip}{\abovedisplayskip}
  \setlength{\abovedisplayshortskip}{0pt}
  \setlength{\belowdisplayshortskip}{3pt}
\begin{align*}
\Qcircuit @C=0.8em @R=1.5em {
\lstick{|0\>} & \gate{W_{x}} \qwx[1] & \gate{W_{\nu(x)}} \qwx[1] & \qw & \cdots & & \gate{W_{\nu^{d-1}(x)}} \qwx[1] & \measure{0} \qwx[1] & \measure{1} \qwx[1] & \qw & \cdots & & \measure{d-1} \qwx[1] & \measure{1} \qwx[1] & \measure{2} \qwx[1] & \qw & \cdots & & \measure{d-1} \qwx[1] & \qw \\
& \multimeasure{2}{x} & \multimeasure{2}{\nu(x)} & \qw & \cdots & & \multimeasure{2}{\nu^{d-1}(x)} & \multimeasure{2}{x} & \multimeasure{2}{x} & \qw & \cdots & & \multimeasure{2}{\nu^{d-1}(x)} & \multigate{2}{\nu} & \multigate{2}{\nu^{2}} & \qw & \cdots & & \multigate{2}{\nu^{d-1}} & \qw \\
\raisebox{6pt}{\vdots} & & & & & & & & & & & & & & & & & & & \raisebox{6pt}{\vdots} \\
& \ghost{x} & \ghost{\nu(x)} & \qw & \cdots & &\ghost{\nu^{d-1}(x)} & \ghost{x} \qwx[1] & \ghost{x} \qwx[1] & \qw & \cdots & & \ghost{\nu^{d-1}(x)} \qwx[1] & \ghost{\nu} & \ghost{\nu^{2}} & \qw & \cdots & & \ghost{\nu^{d-1}} & \qw \\
\lstick{|0\>} & \qw & \qw & \qw & \cdots & & \qw & \multigate{2}{\Phi_{x,0}} & \multigate{2}{\Phi_{x,1}} & \qw & \cdots & & \multigate{2}{\Phi_{\nu^{d-1}(x),d-1}} & \qw & \qw & \qw & \cdots & & \qw & \qw \\
\raisebox{6pt}{\vdots} & & & & & & & & & & & & & & & & & & & \raisebox{6pt}{\vdots} \\
\lstick{|0\>} & \qw & \qw & \qw & \cdots & & \qw & \ghost{\Phi_{x,0}} & \ghost{\Phi_{x,1}} & \qw  & \cdots & & \ghost{\Phi_{\nu^{d-1}(x),d-1}} & \qw & \qw & \qw & \cdots & & \qw & \qw \\
}
\end{align*}
\caption{The quantum circuit for $U_{\rightarrow}$.}
\label{fig:phase-1}
}%
\end{figure}
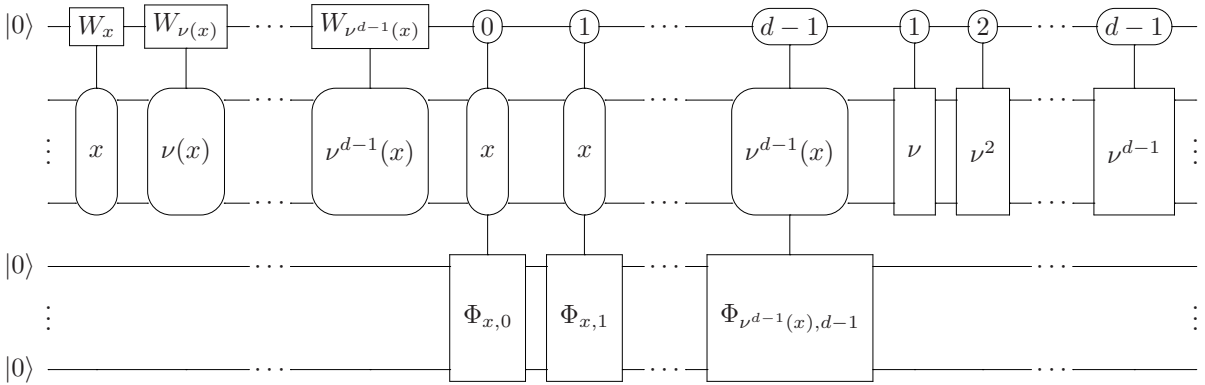

\begin{figure}[htbp]
{
  \setlength{\abovedisplayskip}{6pt}
  \setlength{\belowdisplayskip}{\abovedisplayskip}
  \setlength{\abovedisplayshortskip}{0pt}
  \setlength{\belowdisplayshortskip}{3pt}
\begin{align*}
\Qcircuit @C=0.9em @R=1.5em {
& \multigate{6}{\tilde{U}_{x,1}} & \multigate{6}{\tilde{U}_{x,2}} & \qw & \cdots & & \multigate{6}{\tilde{U}_{x,d-1}} & \gate{P} \qwx[1] & \gate{P^{2}} \qwx[1] & \qw & \cdots & & \gate{P^{d-1}} \qwx[1] & \qw \\
& \ghost{\tilde{U}_{x,1}} & \ghost{\tilde{U}_{x,2}} & \qw & \cdots & & \ghost{\tilde{U}_{x,d-1}} & \multimeasure{2}{x} & \multimeasure{2}{\nu(x)} & \qw & \cdots & & \multimeasure{2}{\nu^{d-1}(x)} & \qw \\
\raisebox{6pt}{\vdots} & & & & & & & & & & & & & \raisebox{6pt}{\vdots} \\
& \ghost{\tilde{U}_{x,1}} & \ghost{\tilde{U}_{x,2}} & \qw & \cdots & & \ghost{\tilde{U}_{x,d-1}} & \ghost{x} & \ghost{\nu(x)} & \qw & \cdots & &\ghost{\nu^{d-1}(x)} & \qw \\
& \ghost{\tilde{U}_{x,1}} & \ghost{\tilde{U}_{x,2}} & \qw & \cdots & & \ghost{\tilde{U}_{x,d-1}} & \qw & \qw & \qw & \cdots & & \qw & \qw \\
\raisebox{6pt}{\vdots} & & & & & & & & & & & & & \raisebox{6pt}{\vdots} \\
& \ghost{\tilde{U}_{x,1}} & \ghost{\tilde{U}_{x,2}}& \qw & \cdots & & \ghost{\tilde{U}_{x,d-1}} & \qw & \qw & \qw & \cdots & & \qw & \qw \\
}
\end{align*}
\caption{The quantum circuit for $U_{\leftarrow}$.}
\label{fig:phase-2}
}%
\end{figure}
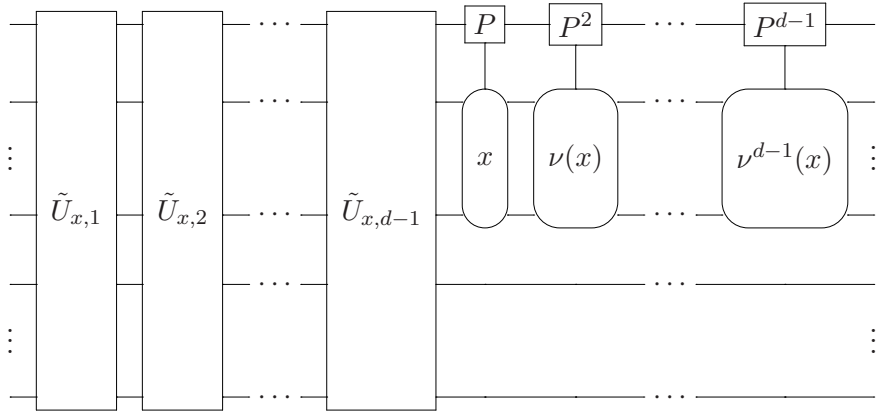

We show below that $U_{\leftarrow}U_{\rightarrow}$ satisfies \eqn{U'}. Using the explicit implementations of $U_{\rightarrow}$ and $U_{\leftarrow}$ shown above, the desired unitary $U_{S_{x}}$ in \eqn{U} can be implemented efficiently.

First consider $U_{\rightarrow}$. It has the effect claimed in \eqn{phase-1} because for any $k\in\range{d}$, it acts as
\begin{align}
|0\>|\nu^{k}(x)\>|0\>&\xmapsto{\wedge_{\nu^{k}(x)}(W_{\nu^{k}(x)})} \sum_{i=0}^{d-1}\||\phi_{\nu^{k}(x),i}\>\||i\>|\nu^{k}(x)\>|0\> \\
&\xmapsto{\Phi_{\nu^{k}(x),0},\dots,\Phi_{\nu^{k}(x),d-1}} \sum_{i=0}^{d-1}|i\>|\nu^{k}(x)\>|\phi_{\nu^{k}(x),i}\> \\
&\xmapsto{\nu,\dots,\nu^{d-1}} \sum_{i=0}^{d-1}|i\>|\nu^{i+k}(x)\>|\phi_{\nu^{k}(x),i}\>.
\end{align}

Now we show that $U_{\leftarrow}$ has the effect claimed in \eqn{phase-2}. First observe that, by \eqn{stinespring-defn}, for each $k\in\{1,2,\ldots,d-1\}$ we have
\begin{align}\label{eqn:trace-preserving}
\Tr[\mathcal{E}(|\nu^{k}(x)\>\<x|)]=0\quad\Rightarrow\quad\sum_{j=0}^{d-1}\<\phi_{x,j+k}|\phi_{\nu^{k}(x),j}\>=0.
\end{align}

On the one hand, applying $\tilde{U}_{x,1}$ to $\sum_{i=0}^{d-1}|i\>|\nu^{i+1}(x)\>|\phi_{\nu(x),i}\>$ gives
\begingroup
\allowdisplaybreaks
\begin{align}
\sum_{i=0}^{d-1}|i\>|\nu^{i+1}(x)\>|\phi_{\nu(x),i}\>&\xmapsto{\wedge_{\nu(x),|\phi_{\nu(x),0}\>}(P^{1})}  |-1\>|\nu(x)\>|\phi_{\nu(x),0}\>+\sum_{i=1}^{d-1}|i\>|\nu^{i+1}(x)\>|\phi_{\nu(x),i}\> \\
&\ \vdots \nonumber \\
&\xmapsto{\wedge_{\nu^{0}(x),|\phi_{\nu(x),d-1}\>}(P^{0})}  |-1\>\sum_{i=0}^{d-1}|\nu^{i+1}(x)\>|\phi_{\nu(x),i}\>= |-1\>|\Phi_{x,1}\> \\
&\xmapsto{\wedge_{|\Phi_{x,1}\>}(P^{-1})}  |0\>\sum_{i=0}^{d-1}|\nu^{i+1}(x)\>|\phi_{\nu(x),i}\> \\
&\xmapsto{\wedge_{\nu^{0}(x),|\phi_{\nu(x),d-1}\>}(P^{-0})}  \sum_{i=0}^{d-2}|0\>|\nu^{i+1}(x)\>|\phi_{\nu(x),i}\>+|d\>|\nu^{d}(x)\>|\phi_{\nu(x),d-1}\> \\
&\ \vdots \nonumber \\
&\xmapsto{\wedge_{\nu(x),|\phi_{\nu(x),0}\>}(P^{-1})}  \sum_{i=0}^{d-2}|i+1\>|\nu^{i+1}(x)\>|\phi_{\nu(x),i}\>.
\end{align}%
\endgroup
\normalsize
In other words,
\begin{align}
\tilde{U}_{x,1}\sum_{i=0}^{d-1}|i\>|\nu^{i+1}(x)\>|\phi_{\nu(x),i}\>=\sum_{i=0}^{d-1}|i+1\>|\nu^{i+1}(x)\>|\phi_{\nu(x),i}\>.
\end{align}

On the other hand, for any $k\neq 1$, the circuit before the $\wedge_{|\Phi_{x,1}\>}(P^{-1})$ gate in $\tilde{U}_{x,1}$ maps $\sum_{i=0}^{d-1}|i\>|\nu^{i+k}(x)\>|\phi_{\nu^{k}(x),i}\>$ to
\begin{align}
\sum_{i=0}^{d-1}|i\>|\nu^{i+k}(x)\>|\phi_{\nu^{k}(x),i}\>&= \sum_{i=0}^{d-1}|i\>|\nu^{i+k}(x)\>\<\phi_{\nu(x),i+k-1}|\phi_{\nu^{k}(x),i}\>|\phi_{\nu(x),i+k-1}\> \nonumber \\
&\quad + \sum_{i=0}^{d-1}|i\>|\nu^{i+k}(x)\>\big(|\phi_{\nu^{k}(x),i}\>-\<\phi_{\nu(x),i+k-1}|\phi_{\nu^{k}(x),i}\>|\phi_{\nu(x),i+k-1}\>\big) \\
&\mapsto |-1\>\sum_{i=0}^{d-1}|\nu^{i+k}(x)\>\<\phi_{\nu(x),i+k-1}|\phi_{\nu^{k}(x),i}\>|\phi_{\nu(x),i+k-1}\> \nonumber \\
&\quad + \sum_{i=0}^{d-1}|i\>|\nu^{i+k}(x)\>\big(|\phi_{\nu^{k}(x),i}\>-\<\phi_{\nu(x),i+k-1}|\phi_{\nu^{k}(x),i}\>|\phi_{\nu(x),i+k-1}\>\big). \label{eqn:intermediate}
\end{align}
By \eqn{trace-preserving},
\begin{align}\label{eqn:orthogonal-1}
\<\Phi_{x,1}|\big(\sum_{i=0}^{d-1}|\nu^{i+k}(x)\>\<\phi_{\nu(x),i+k-1}|\phi_{\nu^{k}(x),i}\>|\phi_{\nu(x),i+k-1}\>\big)=\sum_{i=0}^{d-1}\<\phi_{\nu(x),i+k-1}|\phi_{\nu^{k}(x),i}\>=0.
\end{align}
Furthermore,
\begin{align}
& \<\Phi_{x,1}|\big(|\nu^{i+k}(x)\>(|\phi_{\nu^{k}(x),i}\>-\<\phi_{\nu(x),i+k-1}|\phi_{\nu^{k}(x),i}\>|\phi_{\nu(x),i+k-1}\>)\big) \nonumber \\
&\quad= \<\phi_{\nu(x),i+k-1}|\phi_{\nu^{k}(x),i}\>-\<\phi_{\nu(x),i+k-1}|\phi_{\nu^{k}(x),i}\>=0. \label{eqn:orthogonal-2}
\end{align}
Therefore, by \eqn{orthogonal-1} and \eqn{orthogonal-2} we know that the $\wedge_{|\Phi_{x,1}\>}(P^{-1})$ gate in $\tilde{U}_{x,1}$ does not change the state in \eqn{intermediate}. Finally, the part of $\tilde{U}_{x,1}$ after the $\wedge_{|\Phi_{x,1}\>}(P^{-1})$ gate in $\tilde{U}_{x,1}$ maps \eqn{intermediate} back to $\sum_{i=0}^{d-1}|i\>|\nu^{i+k}(x)\>|\phi_{\nu^{k}(x),i}\>$.

Consequently, $\tilde{U}_{x,1}$ maps $\sum_{i=0}^{d-1}|i\>|\nu^{i+1}(x)\>|\phi_{\nu(x),i}\>$ to $\sum_{i=0}^{d-1}|i+1\>|\nu^{i+1}(x)\>|\phi_{\nu(x),i}\>$, while all other $d-1$ possible inputs $\sum_{i=0}^{d-1}|i\>|\nu^{i+k}(x)\>|\phi_{\nu^{k}(x),i}\>$ (for $k\neq 1$) remain fixed. The same happens for all $\tilde{U}_{x,k}$ with $k\in\{1,\ldots,d-1\}$. Thus the overall effect of $U_{\leftarrow}$ is
\begin{align}
\sum_{i=0}^{d-1}|i\>|\nu^{i+k}(x)\>|\phi_{\nu^{k}(x),i}\>&\xmapsto{\tilde{U}_{x,1},\dots,\tilde{U}_{x,d-1}} \sum_{i=0}^{d-1}|i+k\>|\nu^{i+k}(x)\>|\phi_{\nu^{k}(x),i}\> \\
&\xmapsto{\wedge_{x}(P),\dots,\wedge_{\nu^{d-1}(x)}(P^{d-1})} \sum_{i=0}^{d-1}|0\>|\nu^{i+k}(x)\>|\phi_{\nu^{k}(x),i}\>, \label{eqn:phase-2'}
\end{align}
which is the claimed output, \eqn{phase-2}. Therefore, by \eqn{phase-1} and \eqn{phase-2}, we have
\begin{align}
U_{\leftarrow}U_{\rightarrow}|0\>|\nu^{k}(x)\>|0\>=|0\>\sum_{j=0}^{d-1}|\nu^{j+k}(x)\>|\phi_{\nu^{k}(x),j}\> \quad\forall\,k\in\range{d},
\end{align}
as claimed in \eqn{U'}.

By Lemma 7.11 of \cite{barenco1995elementary}, gates of the form $\wedge_{\nu^{i+k}(x),|\phi_{\nu^{k}(x),i}\>}(P^{i+k})$ can be implemented using $O(\log N+c_{\phi})$ 2-qubit gates for each $i\in[d-1]$ and $k\in[d-1]$, and gates of the form $\wedge_{|\Phi_{x,k}\>}(P^{-k})$ can be implemented using $O(\log N+dc_{\phi})$ 2-qubit gates for each $k\in[d-1]$. As a result, $O(d(\log N+c_{\phi}))+O(\log N+dc_{\phi})=O(d(\log N+c_{\phi}))$ 2-qubit gates suffice to implement $\tilde{U}_{x,k}$ for each $k\in[d-1]$, and $O(d^{2}(\log N+c_{\phi}))$ 2-qubit gates suffice to implement $U_{\leftarrow}$. In addition, $O(d^{2}(\log dN+c_{\phi}))+O(d(\log N+d))=O(d^{2}(\log N+c_{\phi}))$ 2-qubit gates suffice to implement $U_{\rightarrow}$, where we use the fact that $d\leq N$ and $\log dN\leq 2\log N$. Therefore, $U_{S_{x}}$ can be implemented using $O(d^{2}(\log N+c_{\phi}))$ 2-qubit gates.
\end{proof}

\section{Applications of the \stine{} framework}
\label{sec:main}

We now apply the \stine{} framework to give simulations of specific classes of Lindbladians. We show how these methods can efficiently simulate local Lindbladians, subsuming previous work \cite{kliesch2011dissipative,sweke2014simulation}. We also give efficient simulations for four classes of nonlocal Lindbladians.

\subsection{Local Lindbladians}
\label{sec:local}

First we show that our sparse Stinespring isometry framework can efficiently simulate local Lindbladians.

Consider a Lindbladian $\mathcal{L}$ that only acts nontrivially on $c$ qubits of a $(\log N)$-qubit system. Assuming without loss of generality that it acts on the first $c$ qubits, such a Lindbladian has the form
\begin{align}\label{eqn:locality}
\mathcal{L}=\mathcal{L}_{c}\otimes\mathcal{I}_{-c},
\end{align}
where $\mathcal{L}_{c}$ is a Lindbladian on the first $c$ qubits and $\mathcal{I}_{-c}$ is the identity on the remaining $(\log N)-c$ qubits. For arbitrary integer $x\in[N]$ and positive integer $c$, let $x_{c}$ denote the first $c$ bits in the binary representation of $x$, and let $x_{-c}$ denote the remaining $(\log N)-c$ bits. We can efficiently simulate such Lindbladians:

\begin{lemma}\label{lem:local-Lindbladians}
A Lindbladian acting nontrivially on at most $c$ qubits can be simulated for any time $t>0$ using $O(2^{2c}(\log N+2^{6c}))$ 2-qubit gates.
\end{lemma}

\begin{proof}
For a Lindbladian of the form \eqn{locality}, for arbitrary $x_{c},y_{c}\in [2^{c}]$ and $x_{-c},y_{-c}\in [2^{(\log N)-c}]$, we have
\begin{align}
\mathcal{L}\big(|x_{c},x_{-c}\>\<y_{c},y_{-c}|\big)\in\linop{L}(\C^{2^{c}})\otimes |x_{-c}\>\<y_{-c}|.
\end{align}
Thus, $\mathcal{L}$ is invariantly $2^{c}$-sparse. Furthermore, the Gram matrix of $e^{\mathcal{L}t}$ has rank $2^{2c}$ at any time $t>0$, and its first $2^{2c}$ rows and columns constitute a principal submatrix. By \algo{low-rank-Gram} in \append{low-rank-Gram}, for any $x\in\range{N}$ we can prepare the Gram vector $v_{x}$ as a quantum state in $O(2^{6c})$ time. Finally, by \thm{sparse-Stinespring-isometry}, $\mathcal{L}$ can be simulated for any time $t>0$ using $O\big(2^{2c}(\log N+2^{6c})\big)$ 2-qubit gates.
\end{proof}

A $c$-local Lindbladian is sum of Lindbladians, each of which acts nontrivially on at most $c$ qubits. Since there are ${\log N\choose c}$ different choices of the $c$ qubits, the number of Lindbladians in the sum is at most ${\log N\choose c}$. Thus by \lem{local-Lindbladians} and \prop{trotter-formula-second} with $m={\log N\choose c}\leq(\log N)^{c}$ and $L=\|\mathcal{L}\|_{\diamond}$, we have:

\begin{theorem}\label{thm:local-Lindbladians}
A $c$-local Lindbladian $\mathcal{L}$ can be efficiently simulated for any time $t>0$ within error $\epsilon$ using
\begin{align}
O\Big(\frac{\tau^{1.5}}{\sqrt{\epsilon}}\cdot 2^{2c}(\log N+2^{6c})(\log N)^{2.5c}\Big)
\end{align}
2-qubit gates, where $\tau:=\|\mathcal{L}\|_{\diamond}t$.
\end{theorem}

\subsection{Identical-coordinate Lindbladians}

Suppose the matrix $A$ in \eqn{GKS'} satisfies
\begin{alignat}{2}
A_{(i,i),(i,i)}&=a_{i} &\quad& \forall\,i\in\range{N} \label{eqn:identical-coordinate-defn-1} \\
A_{(i,i),(j,j)}&=c_{i}+c_{j} && \forall\,i,j\in\range{N}, i\neq j \label{eqn:identical-coordinate-defn-2} \\
A_{(k,l),(k',l')}&=0 && \text{otherwise}, \label{eqn:identical-coordinate-defn-3}
\end{alignat}
where $\{a_{i}\}_{i=0}^{N-1}$, $\{c_{i}\}_{i=0}^{N-1}$ are real numbers that satisfy
\begin{alignat}{2}
a_{i}&\geq 0 &\quad& \forall\,i\in\range{N} \label{eqn:identical-coordinate-defn-4} \\
|c_{i}+c_{j}|&\leq \sqrt{a_{i}a_{j}} && \forall\,i,j\in\range{N}. \label{eqn:identical-coordinate-defn-5}
\end{alignat}
Then we call the corresponding Lindbladian $\emph{\singly}$, since nonzero elements can only appear if the two row coordinates are identical, and similarly for their columns. By \eqn{GKS'-1}, this property holds if all Lindblad operators are diagonal. Equation \eqn{identical-coordinate-instance-2} is an instance of \singly{} Lindbladians.

Identical-coordinate Lindbladians can be efficiently simulated as follows:

\begin{theorem}\label{thm:singly-sparse}
Given a black box that takes $x\in\range{N}$ as input and outputs $a_{x}$ and $c_{x}$, the \singly{} Lindbladian defined by \eqn{identical-coordinate-defn-1}--\eqn{identical-coordinate-defn-5} can be simulated for any time $t>0$ using $O(\log N)$ 2-qubit gates and $O(1)$ queries to the black box, with no error.
\end{theorem}

\begin{proof}
By \eqn{GKS'}, when $u\neq v$,
\begin{align}\label{eqn:u-neq-v}
\mathcal{L}(|u\>\<v|)=-\big(a_{u}+a_{v}-2(c_{u}+c_{v})\big)|u\>\<v|;
\end{align}
when $u=v$,
\begin{align}\label{eqn:u-equal-v}
\mathcal{L}(|u\>\<u|)=0.
\end{align}
Therefore, coherences $|u\>\<v|$ with $u\neq v$ decay exponentially in time \eqn{u-neq-v}:
\begin{align}\label{eqn:u-neq-v-exp}
e^{\mathcal{L}t}(|u\>\<v|)=e^{-(a_{u}+a_{v}-2(c_{u}+c_{v}))t}|u\>\<v|.
\end{align}
When $u=v$, the state $|u\>\<u|$ is fixed by \eqn{u-equal-v}:
\begin{align}\label{eqn:u-equal-v-exp}
e^{\mathcal{L}t}(|u\>\<u|)=|u\>\<u|.
\end{align}

To implement the Stinespring isometry at time $t$, it suffices to find ancilla states $\{|\phi_{x}\>\}_{x=0}^{N-1}$ such that for arbitrary $x\neq y$,
\begin{align}
\<\phi_{x}|\phi_{x}\>&=1; \label{eqn:A-restriction-1-singly} \\
\<\phi_{x}|\phi_{y}\>&=e^{-(a_{x}+a_{y}-2(c_{x}+c_{y}))t}. \label{eqn:A-restriction-2-singly}
\end{align}

For convenience, denote
\begin{align}
b_{x,t}:=e^{-(a_{x}-2c_{x})t}.
\end{align}
Since $a_{x}\geq 2c_{x}$ by assumption, we have $0\leq b_{x,t}\leq 1$ for arbitrary $x$ and $t>0$.

Consider the following construction of the ancilla states for each $x\in\range{N}$:
\begin{align}
|\phi_{x}\>=b_{x,t}|0\ldots0\>_{1,\ldots,n}|0\>_{n+1}+\sqrt{1-b_{x,t}^{2}}|x\>_{1,\ldots,n}|1\>_{n+1}. \label{eqn:defn-x-singly}
\end{align}
Clearly \eqn{A-restriction-1-singly} and \eqn{A-restriction-2-singly} hold. Furthermore, since $|\phi_{x}\>$ lies in a known two-dimensional subspace, it can be prepared using $O(\log N)$ 2-qubit gates for arbitrary $x$ with $O(1)$ queries to the black box. Therefore, by \thm{sparse-Stinespring-isometry}, we can simulate the \singly{} Lindbladian using $O(\log N)$ 2-qubit gates.
\end{proof}

\subsection{Diagonal Lindbladians}

Next, we consider Lindbladians with diagonal matrix $A$ in \eqn{GKS'}, i.e., only terms of the form $A_{(k,l),(k,l)}$ can be nonzero. Letting
\begin{align}
a_{k,l}:=A_{(k,l),(k,l)},
\end{align}
we see that in this case equation \eqn{GKS'} can be simplified to the following form:
\begin{align}\label{eqn:GKS'-diagonal}
\mathcal{L}(\rho)=\sum_{k,l=0}^{N-1}a_{k,l}\big(2\<l|\rho|l\>|k\>\<k|-|l\>\<l|\rho-\rho|l\>\<l|\big),
\end{align}
where $a_{k,l}\geq 0$ for all $k,l\in\range{N}$. We call such a Lindbladian \emph{diagonal}.

By \eqn{GKS'-diagonal},
\begin{alignat}{2}
\mathcal{L}(|u\>\<u|)&=2\sum_{k=0}^{N-1}a_{k,u}\big(|k\>\<k|-|u\>\<u|\big) && \forall\,u \label{eqn:GKS'-diagonal-1} \\
\mathcal{L}(|u\>\<v|)&=-\sum_{k=0}^{N-1}(a_{k,u}+a_{k,v})|u\>\<v| &\quad& \forall\,u\neq v. \label{eqn:GKS'-diagonal-2}
\end{alignat}
Intuitively, \eqn{GKS'-diagonal-1} shows that diagonal terms follow a classical continuous-time Markov process, while \eqn{GKS'-diagonal-2} shows that off-diagonal terms simply decay.

\subsubsection{Sparse-diagonal Lindbladians}
\label{sec:sparse-diagonal}

If the matrix $a$ with entries $a_{k,l}$ for $k,l\in\range{N}$ is $d$-sparse (i.e., each row and column has at most $d$ nonzero entries), where $d\in\N$ is a fixed constant, we say that $\mathcal{L}$ is \emph{$d$-sparse-diagonal}. Such Lindbladians can be efficiently simulated:

\begin{theorem}\label{thm:d-sparse}
If the matrix $a$ is $d$-sparse, then given a black box for $a$ that takes a row index or column index as input and outputs all nonzero entries in that row/column together with their positions, the Lindbladian defined by \eqn{GKS'-diagonal} can be simulated for any time $t>0$ within error $\epsilon$ using
\begin{align}
O\Big(\frac{d^{5}\tau^{1.5}}{\sqrt{\epsilon}}\log N\Big)
\end{align}
2-qubit gates and $O\big(\frac{d^{5}\tau^{1.5}}{\sqrt{\epsilon}}\big)$ queries to the black box, where $\tau:=\|\mathcal{L}\|_{\diamond}t$.
\end{theorem}

First consider the special case where the matrix $a$ is \emph{strongly 1-sparse}, i.e., there exists an involution $\nu$ such that for all $x\in\range{N}$, $a_{x,y}\neq 0$ only if $y\in\{x,\nu(x)\}$. Letting $G=(V,E)$ be the underlying undirected graph of $a$, where rows/columns represent vertices and $(v_{x},v_{y})\in E$ if and only if $a_{x,y}\neq 0$, then $G$ consists only of isolated vertices (if $\nu(x)=x$, $v_{x}$ is an isolated vertex) and isolated edges (if $\nu(x)\neq x$, $(x,\nu(x))\in E$ and $x,\nu(x)$ have no other neighbors). We call $\nu$ the neighbor function as in \defn{d-sparse-Stine}. We first show \thm{d-sparse} in this special case:

\begin{lemma}\label{lem:1-sparse-a}
Suppose the matrix $a$ is strongly 1-sparse and we are given a black box for the matrix $a$ that takes input $x\in\range{N}$ and outputs $\nu(x)$, $a_{x,x}$, and $a_{\nu(x),x}$. Then the Lindbladian defined by \eqn{GKS'-diagonal} can be simulated for any time $t>0$ by $O(\log N)$ 2-qubit gates and $O(1)$ queries to the black box, with no error.
\end{lemma}

To handle the general case, we show that any $d$-sparse-diagonal Lindbladian can be decomposed into at most $3d^{2}$ strongly 1-sparse-diagonal Lindbladians:
\begin{lemma}\label{lem:a-decomposition}
Every $d$-sparse-diagonal Lindbladian can be written as the sum of at most $3d^{2}$ strongly 1-sparse-diagonal Lindbladians with constant overhead in queries using the black box in \thm{d-sparse}.
\end{lemma}
\begin{proof}[Proof of \thm{d-sparse}]
The theorem follows directly from \lem{1-sparse-a}, \lem{a-decomposition}, and \prop{trotter-formula-second} with $m=3d^{2}$ and $L=\|\mathcal{L}\|_{\diamond}$.
\end{proof}

The proof of \lem{a-decomposition} is in \append{sparse-diagonal}. We prove \lem{1-sparse-a} here.

\begin{proof}[Proof of \lem{1-sparse-a}]
By \thm{sparse-Stinespring-isometry}, it suffices to show that for any time $t>0$, there exists a 1-sparse Stinespring isometry $V$ whose ancilla states can be prepared using $O(\log N)$ 2-qubit gates, such that
\begin{align}
e^{\mathcal{L}t}(\rho)=\Tr_{\anc}[V\rho V^{\dagger}].
\end{align}

Without loss of generality, assume $\nu(x)\neq x$ for all $x\in\range{N}$. Since $a$ is strongly 1-sparse, by \eqn{GKS'-diagonal} we have
\begin{alignat}{2}
\mathcal{L}(|x\>\<x|)&=2a_{\nu(x),x}|\nu(x)\>\<\nu(x)|-2a_{\nu(x),x}|x\>\<x| &\quad& \forall\,x; \label{eqn:diagonal-A1} \\
\mathcal{L}(|x\>\<y|)&=-(a_{\nu(x),x}+a_{\nu(y),y}+a_{x,x}+a_{y,y})|x\>\<y| && \forall\,x\neq y. \label{eqn:diagonal-A2}
\end{alignat}
Because
\begin{align}
\exp\Big[\left(
      \begin{array}{cc}
        -\alpha & \alpha \\
        \beta & -\beta
      \end{array} \right)t\Big]=\left(
      \begin{array}{cc}
        \frac{\beta}{\alpha+\beta}+\frac{\alpha}{\alpha+\beta}e^{-(\alpha+\beta)t} & \frac{\alpha}{\alpha+\beta}-\frac{\alpha}{\alpha+\beta}e^{-(\alpha+\beta)t} \\
        \frac{\beta}{\alpha+\beta}-\frac{\beta}{\alpha+\beta}e^{-(\alpha+\beta)t} & \frac{\alpha}{\alpha+\beta}+\frac{\beta}{\alpha+\beta}e^{-(\alpha+\beta)t}
      \end{array} \right),
\end{align}
\eqn{diagonal-A1} implies
\begin{align}
e^{\mathcal{L}t}(|x\>\<x|)&=\Big(\frac{a_{x,\nu(x)}}{a_{x,\nu(x)}+a_{\nu(x),x}}+\frac{a_{\nu(x),x}}{a_{x,\nu(x)}+a_{\nu(x),x}}e^{-2(a_{x,\nu(x)}+a_{\nu(x),x})t}\Big)|x\>\<x| \nonumber \\
&\quad+\frac{a_{\nu(x),x}}{a_{x,\nu(x)}+a_{\nu(x),x}}\big(1-e^{-2(a_{x,\nu(x)}+a_{\nu(x),x})t}\big)|\nu(x)\>\<\nu(x)|.
\end{align}
By \eqn{diagonal-A2}, we have
\begin{align}
e^{\mathcal{L}t}(|x\>\<y|)=e^{-(a_{\nu(x),x}+a_{\nu(y),y}+a_{x,x}+a_{y,y})t}|x\>\<y|\quad\forall\,x\neq y.
\end{align}

By \defn{d-sparse-Stine}, it suffices to find ancilla states $\{|\phi_{x,1}\>,|\phi_{x,2}\>\}_{x=0}^{N-1}$ such that for arbitrary $x\neq y$,
\begin{align}
	\<\phi_{x,1}|\phi_{x,1}\>&=\frac{a_{x,\nu(x)}}{a_{x,\nu(x)}+a_{\nu(x),x}}+\frac{a_{\nu(x),x}}{a_{x,\nu(x)}+a_{\nu(x),x}}e^{-2(a_{x,\nu(x)}+a_{\nu(x),x})t}; \label{eqn:A-restriction-1-diag} \\
	\<\phi_{x,2}|\phi_{x,2}\>&=\frac{a_{\nu(x),x}}{a_{x,\nu(x)}+a_{\nu(x),x}}\big(1-e^{-2(a_{x,\nu(x)}+a_{\nu(x),x})t}\big); \label{eqn:A-restriction-2-diag} \\
	\<\phi_{x,1}|\phi_{x,2}\>&=0; \label{eqn:A-restriction-3-diag} \\
	\<\phi_{y,1}|\phi_{x,1}\>&=e^{-(a_{\nu(x),x}+a_{\nu(y),y}+a_{x,x}+a_{y,y})t}; \label{eqn:A-restriction-4-diag} \\
	\<\phi_{y,1}|\phi_{x,2}\>&=0; \label{eqn:A-restriction-5-diag} \\
	\<\phi_{y,2}|\phi_{x,2}\>&=0. \label{eqn:A-restriction-6-diag}
\end{align}

Consider the following construction of the ancilla states for each $x\in\range{N}$:
\begin{align}
|\phi_{x,1}\>&=b_{x}|0\ldots0\>_{1,\ldots,n}|1\>_{n+1}|0\>_{n+2}+c_{x}|x\>_{1,\ldots,n}|1\>_{n+1}|1\>_{n+2} \label{eqn:defn-x1-diag} \\
|\phi_{x,2}\>&=a_{x}|x\>_{1,\ldots,n}|0\>_{n+1}|0\>_{n+2}, \label{eqn:defn-x2-diag}
\end{align}
where $a_{x}$, $b_{x}$, and $c_{x}$ are defined as
\begin{align}
a_{x}&:=\sqrt{\frac{a_{\nu(x),x}}{a_{x,\nu(x)}+a_{\nu(x),x}}\big(1-e^{-2(a_{x,\nu(x)}+a_{\nu(x),x})t}\big)} \label{eqn:sparse-diagonal-a} \\
b_{x}&:=e^{-(a_{\nu(x),x}+a_{x,x})t} \label{eqn:sparse-diagonal-b} \\
c_{x}&:=\sqrt{\frac{a_{x,\nu(x)}}{a_{x,\nu(x)}+a_{\nu(x),x}}+\frac{a_{\nu(x),x}}{a_{x,\nu(x)}+a_{\nu(x),x}}e^{-2(a_{x,\nu(x)}+a_{\nu(x),x})t}-e^{-2(a_{\nu(x),x}+a_{x,x})t}} \label{eqn:sparse-diagonal-c}
\end{align}
for all $x\in\range{N}$. The entry inside the square root in $c_{x}$ is positive because from
\begin{align}
&\frac{a_{x,\nu(x)}}{a_{x,\nu(x)}+a_{\nu(x),x}}e^{2a_{\nu(x),x}t}+\frac{a_{\nu(x),x}}{a_{x,\nu(x)}+a_{\nu(x),x}}e^{-2a_{x,\nu(x)}t} \nonumber \\
&\quad\geq\frac{a_{x,\nu(x)}}{a_{x,\nu(x)}+a_{\nu(x),x}}(1+2a_{\nu(x),x}t)+\frac{a_{\nu(x),x}}{a_{x,\nu(x)}+a_{\nu(x),x}}(1-2a_{x,\nu(x)}t) \\
&\quad= 1 \\
&\quad\geq e^{-2a_{x,x}}t,
\end{align}
we have
\begin{align}
\frac{a_{x,\nu(x)}}{a_{x,\nu(x)}+a_{\nu(x),x}}+\frac{a_{\nu(x),x}}{a_{x,\nu(x)}+a_{\nu(x),x}}e^{-2(a_{x,\nu(x)}+a_{\nu(x),x})t}\geq e^{-2(a_{\nu(x),x}+a_{x,x})t}.
\end{align}

The definitions of $|\phi_{x,1}\>$ and $|\phi_{x,2}\>$ as in \eqn{defn-x1-diag} and \eqn{defn-x2-diag} directly give \eqn{A-restriction-2-diag}--\eqn{A-restriction-6-diag}. As for \eqn{A-restriction-1-diag}, we have
\begin{align}
\<\phi_{x,1}|\phi_{x,1}\>&=b_{x}^{2}+c_{x}^{2} \\
&=\frac{a_{x,\nu(x)}}{a_{x,\nu(x)}+a_{\nu(x),x}}+\frac{a_{\nu(x),x}}{a_{x,\nu(x)}+a_{\nu(x),x}}e^{-2(a_{x,\nu(x)}+a_{\nu(x),x})t}.
\end{align}
In addition, $\||\phi_{x,1}\>\|^{2}+\||\phi_{x,2}\>\|^{2}=1$ and $\<\phi_{\nu(x),2}|\phi_{x,1}\>+\<\phi_{\nu(x),1}|\phi_{x,2}\>=0$ are satisfied for all $x$. Therefore, \eqn{defn-x1-diag} and \eqn{defn-x2-diag} give a construction of the 1-sparse Stinespring isometry as claimed.

The state $|\phi_{x,2}\>/\||\phi_{x,2}\>\|$ can be directly prepared using $O(\log N)$ 2-qubit gates. The state $|\phi_{x,1}\>/\||\phi_{x,1}\>\|$ can be prepared by $O(\log N)$ 2-qubit gates and $O(1)$ queries using the following procedure:
\begin{align}
&|0\ldots0\>_{1,\ldots,n}|0\>_{n+1}|0\>_{n+2} \nonumber \\
&\quad\mapsto \frac{1}{\sqrt{b_{x}^{2}+c_{x}^{2}}}\big(b_{x}|0\ldots0\>_{1,\ldots,n}|1\>_{n+1}|0\>_{n+2}+c_{x}|0\ldots0\>_{1,\ldots,n}|1\>_{n+1}|1\>_{n+2}\big) \\
&\quad\mapsto \frac{1}{\sqrt{b_{x}^{2}+c_{x}^{2}}}\big(b_{x}|0\ldots0\>_{1,\ldots,n}|1\>_{n+1}|0\>_{n+2}+c_{x}|x\>_{1,\ldots,n}|1\>_{n+1}|1\>_{n+2}\big). \label{eqn:diagonal-state-preparation}
\end{align}
This completes the proof.
\end{proof}

\subsubsection{Dense-diagonal Lindbladians}
\label{sec:dense-diagonal}

If the matrix $a$ is dense, a counting argument shows it is hard to simulate $\mathcal{L}$ in general. However, we can efficiently simulate $\mathcal{L}$ in the following special case:

\begin{theorem}\label{thm:equal-diagonal}
Suppose the entries of $a$ are independent of the column index, i.e., $a_{k,l} =: a_{k}$ is independent of $l$ for all $k$. Then, given a black box that can compute $\sum_{k=k_{1}}^{k_{2}}a_{k}$ for arbitrary $k_{1},k_{2}\in\range{N}, k_{1}\leq k_{2}$, the Lindbladian $\mathcal{L}$ defined by \eqn{GKS'-diagonal} can be simulated for any time $t>0$ using $\poly(\log N)$ 2-qubit gates, with no error.
\end{theorem}

Observe that, unlike the other simulations we present, \thm{equal-diagonal} efficiently simulates a set of Lindbladians with full-rank matrices $A$ in \eqn{GKS'}.  In other words, it can efficiently simulate a set of Lindbladians with maximum number of Lindblad operators. \thm{equal-diagonal} also applies to the special case $A=I_{N^{2}}$.

Intuitively, if we consider the underlying bipartite graph $G=((V_{R},V_{C}),E)$ of $a$ where $V_{R}$ represents the part for rows and $V_{C}$ represents the part for columns, then $G$ contains $N$ unweighted stars centered in $V_{R}$ and we can simultaneously simulate all $N$ stars with arbitrary weights on them.

\begin{proof}[Proof of \thm{equal-diagonal}]
For convenience, define
\begin{align}
S:=\sum_{k=0}^{N-1}a_{k}.
\end{align}

By \eqn{GKS'-diagonal}, we have
\begin{alignat}{2}
\mathcal{L}(|u\>\<v|)&= -2S|u\>\<v| &\quad& \forall\,u\neq v; \label{eqn:diagonal-u-neq-v} \\
\mathcal{L}(|u\>\<u|)&= 2\sum_{k=0}^{N-1}a_{k}|k\>\<k|-2S|u\>\<u| && \forall\,u. \label{eqn:diagonal-u-eqn-v}
\end{alignat}
Furthermore, we have
\begin{align}
\mathcal{L}\big(\sum_{k=0}^{N-1}a_{k}|k\>\<k|\big)=\sum_{u=0}^{N-1}a_{u}\big(2\sum_{k=0}^{N-1}a_{k}|k\>\<k|-2S|u\>\<u|\big)=0. \label{eqn:diagonal-u-eqn-v-2}
\end{align}
Collecting \eqn{diagonal-u-eqn-v} and \eqn{diagonal-u-eqn-v-2} into a single vector equation, we have
\begin{align}
\mathcal{L}
\begin{pmatrix}
|u\>\<u| \\
\frac{\sum_{k=0}^{N-1}a_{k}|k\>\<k|}{S}
\end{pmatrix}=\left(
\begin{array}{cc}
-2S & 2S \\
0 & 0
\end{array} \right)
\begin{pmatrix}
|u\>\<u| \\
\frac{\sum_{k=0}^{N-1}a_{k}|k\>\<k|}{S}
\end{pmatrix}.
\end{align}

Because
\begin{align}
\exp\biggl[\left(
\begin{array}{cc}
-2S & 2S \\
0 & 0
\end{array} \right)t\biggr]=\left(
\begin{array}{cc}
e^{-2St} & 1-e^{-2St} \\
0 & 1
\end{array} \right),
\end{align}
we have
\begin{align}\label{eqn:diagonal-condition-1}
e^{\mathcal{L}t}(|u\>\<u|)=e^{-2St}|u\>\<u|+(1-e^{-2St})\frac{\sum_{k=0}^{N-1}a_{k}|k\>\<k|}{S}.
\end{align}
Also, by \eqn{diagonal-u-neq-v}, we have
\begin{align}\label{eqn:diagonal-condition-2}
e^{\mathcal{L}t}(|u\>\<v|)=e^{-2St}|u\>\<v|.
\end{align}
It suffices to efficiently implement an isometry that satisfies \eqn{diagonal-condition-1} and \eqn{diagonal-condition-2}.

Consider the isometry $V\colon\C^{N}\rightarrow \C^{N}\otimes \C^{N}\otimes \C^{N+1}$ defined by
\begin{align}\label{eqn:diagonal-isometry}
V|m\>=e^{-St}|m\>|0\>|0\>+\sum_{k=0}^{N-1}\sqrt{\frac{(1-e^{-2St})a_{k}}{S}}|k\>|k\>|m+1\>\quad\forall\,m\in\range{N},
\end{align}
where the second and third subsystems of the output of $V$ are ancilla states. Considering the effect on the third subsystem, we have
\begin{align}
\Tr_{\anc}[V|m\>\<n|V^{\dagger}]=e^{-2St}|m\>\<n| \quad\forall\,m\neq n
\end{align}
and
\begin{align}
\Tr_{\anc}[V|m\>\<m|V^{\dagger}]&=e^{-2St}|m\>\<m|+(1-e^{-2St})\frac{\sum_{k=0}^{N-1}a_{k}|k\>\<k|}{S} \quad\forall\,m,
\end{align}
which match \eqn{diagonal-condition-1} and \eqn{diagonal-condition-2}.

The isometry $V$ can be implemented by the quantum circuit in \fig{diagonal-star}.
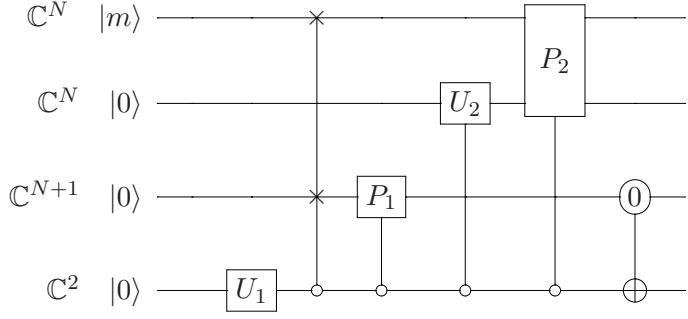
\begin{figure}[htbp]
\centering
\begin{align*}
\Qcircuit @C=1.2em @R=1.8em {
\lstick{\C^{N}\quad|m\>} & \qw & \qw & \qswap \qwx[2] & \qw & \qw & \multigate{1}{P_{2}} & \qw & \qw \\
\lstick{\C^{N}\quad|0\>} & \qw & \qw & \qw & \qw & \gate{U_{2}} & \ghost{P_{2}} & \qw & \qw \\
\lstick{\C^{N+1}\quad|0\>} &  \qw & \qw & \qswap & \gate{P_{1}} & \qw & \qw &\measure{0} \qwx[1] & \qw \\
\lstick{\C^{2}\quad|0\>} & \qw & \gate{U_{1}} & \ctrlo{-1} & \ctrlo{-1} & \ctrlo{-2} & \ctrlo{-2} & \targ & \qw \\
}
\end{align*}
\caption{The quantum circuit for simulating dense-diagonal Lindbladians as in \thm{equal-diagonal}.}
\label{fig:diagonal-star}
\end{figure}
Here, $U_{1}$ is a unitary acting on an ancilla qubit such that
\begin{align}
U_{1}|0\>=\sqrt{1-e^{-2St}}|0\>+e^{-St}|1\>;
\end{align}
the swap gate swaps the state in the first and the third subsystems (the state $|N\>$ cannot appear in the third system, so this swap is well-defined); $P_{1}$ is a permutation acting on the third subsystem such that
\begin{align}
P_{1}|N\>=|0\>,\quad P_{1}|k\>=|k+1\>\quad\forall\,k\in\range{N};
\end{align}
$U_{2}$ is a unitary gate acting on the second subsystem such that
\begin{align}
U_{2}|0\>=\sum_{k=0}^{N-1}\sqrt{\frac{a_{k}}{S}}|k\>;
\end{align}
and $P_{2}$ is a permutation acting on the first and second subsystems such that
\begin{alignat}{2}
P_{0}|0\>|k\>&=|k\>|k\> &\quad& \forall\,k\in\{1,\ldots,N-1\}, \\
P_{0}|k\>|k\>&=|0\>|k\> && \forall\,k\in\{1,\ldots,N-1\}, \\
P_{0}|i\>|j\>&=|i\>|j\> && \text{otherwise.}
\end{alignat}
The gates $U_{1}$, $P_{1}$, $P_{2}$, $\textsf{SWAP}$, and $\wedge_{0}(\sigma_{x})$ (defined in \fig{building-block-1}) can be implemented with $\poly(\log N)$ 2-qubit gates using standard techniques. The gate $U_{2}$ can be implemented with $\poly(\log N)$ 2-qubit gates using the black box and the method of \cite{grover2002creating}.

Finally, we verify that this quantum circuit gives the isometry $V$ as claimed:
\begin{align}
|m\>|0\>|0\>|0\>&\xmapsto{U_{1}}|m\>|0\>|0\>(\sqrt{1-e^{-2St}}|0\>+e^{-St}|1\>) \\
&=\sqrt{1-e^{-2St}}|m\>|0\>|0\>|0\>+e^{-St}|m\>|0\>|0\>|1\> \\
&\xmapsto{\textsf{SWAP}}\sqrt{1-e^{-2St}}|0\>|0\>|m\>|0\>+e^{-St}|m\>|0\>|0\>|1\> \\
&\xmapsto{P_{1}}\sqrt{1-e^{-2St}}|0\>|0\>|m+1\>|0\>+e^{-St}|m\>|0\>|0\>|1\> \\
&\xmapsto{U_{2}}\sqrt{1-e^{-2St}}|0\>\big(\sum_{k=0}^{N-1}\sqrt{\frac{a_{k}}{S}}|k\>\big)|m+1\>|0\>+e^{-St}|m\>|0\>|0\>|1\> \\
&=\sum_{k=0}^{N-1}\sqrt{\frac{(1-e^{-2St})a_{k}}{S}}|0\>|k\>|m+1\>|0\>+e^{-St}|m\>|0\>|0\>|1\> \\
&\xmapsto{P_{2}}\sum_{k=0}^{N-1}\sqrt{\frac{(1-e^{-2St})a_{k}}{S}}|k\>|k\>|m+1\>|0\>+e^{-St}|m\>|0\>|0\>|1\> \\
&\xmapsto{\wedge_{0}(\sigma_{x})}\Big(\sum_{k=0}^{N-1}\sqrt{\frac{(1-e^{-2St})a_{k}}{S}}|k\>|k\>|m+1\>+e^{-St}|m\>|0\>|0\>\Big)|0\>.
\end{align}
This completes the proof.
\end{proof}

\subsection{1-ket-sparse diagonal Lindbladians}

In our final application of the sparse Stinespring isometry framework, we consider Lindbladians for which the only nonzero elements of the matrix $A$ in \eqn{GKS'} are
\begin{align}
a_{k,l}:=A_{(k,l),(k,l)} \quad \text{and} \quad b_{k,l}:=A_{(k,l),(\nu(k),l)},
\end{align}
where $\nu$ is an involution of $\range{N}$. As in \sec{sparse-diagonal}, we call $\nu$ the neighbor function. Furthermore, we assume that $a$ and $b$ are also 1-sparse with respect to the same neighbor function $\nu$ and that their coefficients satisfy the following relationships:
\begin{align}
a_{\{x,\nu(x)\}}&:=a_{x,\nu(x)}=a_{\nu(x),x}, \label{eqn:1-k-sparse-defn-1} \\
a'_{\{x,\nu(x)\}}&:=a_{x,x}=a_{\nu(x),\nu(x)}, \label{eqn:1-k-sparse-defn-2} \\
b_{\{x,\nu(x)\}}&:=b_{x,x}=b_{x,\nu(x)}=b_{\nu(x),x}=b_{\nu(x),\nu(x)}\in\R. \label{eqn:1-k-sparse-defn-3}
\end{align}
We call such a Lindbladian \emph{1-ket-sparse} because $\nu$ can only act on the first coordinates of both rows and columns, which correspond to the ket vectors in \eqn{Lindblad-ket-bra}. Applying \eqn{1-k-sparse-defn-1}, \eqn{1-k-sparse-defn-2}, and \eqn{1-k-sparse-defn-3} to \eqn{GKS'}, we see that 1-ket-sparse Lindbladians have the form
\begin{align}\label{eqn:GKS'-1-k-sim}
\mathcal{L}(\rho)&=\sum_{k=0}^{N-1}a_{\{k,\nu(k)\}}\big(2\<\nu(k)|\rho|\nu(k)\>|k\>\<k|-|\nu(k)\>\<\nu(k)|\rho-\rho|\nu(k)\>\<\nu(k)|\big) \nonumber \\
&\quad+\sum_{k=0}^{N-1}a'_{\{k,\nu(k)\}}\big(2\<k|\rho|k\>|k\>\<k|-|k\>\<k|\rho-\rho|k\>\<k|\big) \nonumber \\
&\quad+2\sum_{k=0}^{N-1}b_{\{k,\nu(k)\}}\big(\<k|\rho|k\>+\<\nu(k)|\rho|\nu(k)\>\big)|k\>\<\nu(k)|.
\end{align}

We show that 1-ket-sparse Lindbladians can be efficiently simulated:
\begin{theorem}\label{thm:1-k-sparse}
Given a black box that takes $x\in\range{N}$ as input and outputs $\nu(x)$, $a_{\{x,\nu(x)\}}$, $a'_{\{x,\nu(x)\}}$ and $b_{\{x,\nu(x)\}}$, the corresponding 1-ket-sparse Lindbladian defined by \eqn{GKS'-1-k-sim} can be simulated for any time $t>0$ using $O(\log N)$ 2-qubit gates and $O(1)$ queries to the black box, with no error.
\end{theorem}

While the class of 1-ket-sparse Lindbladians is somewhat artificial, this result shows that we can efficiently simulate some Lindbladians that include off-diagonal elements, together with a structured diagonal.

\begin{proof}[Proof of \thm{1-k-sparse}]
By \thm{sparse-Stinespring-isometry}, it suffices to show that for any time $t>0$, there exists an 1-sparse Stinespring isometry $V$ whose ancilla states can be prepared using $O(\log N)$ 2-qubit gates such that
\begin{align}
e^{\mathcal{L}t}(\rho)=\Tr_{\anc}[V\rho V^{\dagger}].
\end{align}

Since the matrix $A$ in \eqn{GKS'} is a positive semidefinite matrix,
\begin{align}\label{eqn:semi-definite-small}
b_{\{u,\nu(u)\}}^{2}\leq 4a_{\{u,\nu(u)\}}a'_{\{u,\nu(u)\}}\quad\forall\,u\in\range{N}.
\end{align}
By \eqn{GKS'-1-k-sim},
\begin{alignat}{2}
\mathcal{L}(|u\>\<v|)&=-(a_{u,u}+a_{v,v}+a_{\nu(u),u}+a_{\nu(v),v})|u\>\<v|&\quad& \forall\,u\neq v; \label{eqn:off-diagonal} \\
\mathcal{L}(|u\>\<u|)&=2a_{\{u,\nu(u)\}}\big(|\nu(u)\>\<\nu(u)|-|u\>\<u|)+b_{\{u,\nu(u)\}}(|\nu(u)\>\<u|+|u\>\<\nu(u)|\big) && \forall\,u. \label{eqn:on-diagonal}
\end{alignat}
Therefore, when $u\notin\{v,\nu(v)\}$, we have
\begin{align}
e^{\mathcal{L}t}(|u\>\<v|)=e^{-(a_{u,u}+a_{v,v}+a_{\nu(u),u}+a_{\nu(v),v})t}|u\>\<v|. \label{eqn:off-diagonal-exp}
\end{align}

For convenience, fix $u$ and denote
\begin{align}
a:=a_{\{u,\nu(u)\}},\quad a':=a'_{\{u,\nu(u)\}},\quad b:=b_{\{u,\nu(u)\}}.
\end{align}
Let $\vec{V}_{\{u,\nu(u)\}}:=\big(|u\>\<u|,|u\>\<\nu(u)|,|\nu(u)\>\<u|,|\nu(u)\>\<\nu(u)|\big)^{T}$. In the subspace spanned by $\vec{V}_{\{u,\nu(u)\}}$, we have
\begin{align}\label{eqn:on-diagonal-subspace}
\mathcal{L}\big(\vec{V}_{\{u,\nu(u)\}}\big)=
\left(\begin{array}{cccc}
-2a & b & b & 2a \\
0 & -2(a+a') & 0 & 0 \\
0 & 0 & -2(a+a') & 0 \\
2a & b & b & -2a
\end{array} \right)
\vec{V}_{\{u,\nu(u)\}}.
\end{align}
Solving the differential equation $\frac{\d\rho}{\d t}=\mathcal{L}(\rho)$, we get
\begin{align}
e^{\mathcal{L}t}\big(\vec{V}_{\{u,\nu(u)\}}\big)=
\left(\begin{array}{cccc}
\frac{1}{2}(1+e^{-4at}) & \frac{b(1-e^{-2(a+a')t})}{2(a+a')} & \frac{b(1-e^{-2(a+a')t})}{2(a+a')} & \frac{1}{2}(1-e^{-4at}) \\
0 & e^{-2(a+a')t} & 0 & 0 \\
0 & 0 & e^{-2(a+a')t} & 0 \\
\frac{1}{2}(1-e^{-4at}) & \frac{b(1-e^{-2(a+a')t})}{2(a+a')} & \frac{b(1-e^{-2(a+a')t})}{2(a+a')} & \frac{1}{2}(1+e^{-4at})
\end{array} \right) \vec{V}_{\{u,\nu(u)\}}.
\end{align}
Therefore, for all $u\in\range{N}$, we have
\begin{align}
e^{\mathcal{L}t}(|u\>\<u|)&=\frac{1}{2}(1+e^{-4at})|u\>\<u|+\frac{1}{2}(1-e^{-4at})|\nu(u)\>\<\nu(u)| \nonumber \\
&\quad+\frac{b(1-e^{-2(a+a')t})}{2(a+a')}(|u\>\<\nu(u)|+|\nu(u)\>\<u|). \label{eqn:on-diagonal-exp}
\end{align}

By \eqn{off-diagonal-exp}, \eqn{on-diagonal-exp}, and \defn{d-sparse-Stine}, it suffices to find ancilla states $\{|\phi_{u,1}\>,|\phi_{u,2}\>\}_{x=0}^{N-1}$ such that for arbitrary $u\neq v$,
\begin{align}
	\<\phi_{u,1}|\phi_{u,1}\>&=\frac{1}{2}(1+e^{-4at}); \label{eqn:A-restriction-1} \\
	\<\phi_{u,2}|\phi_{u,2}\>&=\frac{1}{2}(1-e^{-4at}); \label{eqn:A-restriction-2} \\
	\<\phi_{u,1}|\phi_{u,2}\>&=\frac{b(1-e^{-2(a+a')t})}{2(a+a')}; \label{eqn:A-restriction-3} \\
	\<\phi_{v,1}|\phi_{u,1}\>&=e^{-(a_{\{\nu(u),u\}}+a_{\{\nu(v),v\}}+a_{\{u,u\}}+a_{\{v,v\}})t}; \label{eqn:A-restriction-4} \\
	\<\phi_{v,1}|\phi_{u,2}\>&=0; \label{eqn:A-restriction-5} \\
	\<\phi_{v,2}|\phi_{u,2}\>&=0. \label{eqn:A-restriction-6}
\end{align}
Consider the following construction of the ancilla states for all $u\in\range{N}$:
\begin{align}
|\phi_{u,1}\>&=b_{u}|0\ldots0\>_{1,\dots,n}|1\>_{n+1}|0\>_{n+2}+c_{u}|u\>_{1,\ldots,n}|1\>_{n+1}|1\>_{n+2}; \label{eqn:defn-x1} \\
|\phi_{u,2}\>&=\frac{d_{u}}{c_{u}}|u\>_{1,\ldots,n}|1\>_{n+1}|1\>_{n+2}+\sqrt{a_{u}^{2}-\frac{d_{u}^{2}}{c_{u}^{2}}}|u\>_{1,\ldots,n}|0\>_{n+1}|0\>_{n+2}, \label{eqn:defn-x2}
\end{align}
with
\begin{align}
a_{u}&:=\sqrt{\frac{1}{2}(1-e^{-4at})} \\
b_{u}&:=e^{-(a+a')t} \\
c_{u}&:=\sqrt{\frac{1}{2}(1+e^{-4at}-2e^{-2(a+a')t})} \\
d_{u}&:=\frac{b(1-e^{-2(a+a')t})}{2(a+a')}.
\end{align}

For any $a,a',t>0$, one can show that
\begin{align}\label{eqn:ac-d}
(1-e^{-4at})(1+e^{-4at}-2e^{-2(a+a')t})\geq\frac{4aa'}{(a+a')^{2}}(1-e^{-2(a+a')t})^{2}.
\end{align}
By \eqn{semi-definite-small} and \eqn{ac-d}, we have
\begin{align}\label{eqn:ac-geq-d}
a_{u}^{2}c_{u}^{2}\geq d_{u}^{2}.
\end{align}
If $c_{u}=0$, then by \eqn{ac-geq-d} we must have $d_{u}=0$. In this case, we simply define $\frac{d_{u}^{2}}{c_{u}^{2}}:=0$ in \eqn{defn-x2}.

The definitions of $|\phi_{u,1}\>$ and $|\phi_{u,2}\>$ as in \eqn{defn-x1} and \eqn{defn-x2} directly give \eqn{A-restriction-3}--\eqn{A-restriction-6}. As for \eqn{A-restriction-1} and \eqn{A-restriction-2}, we have
\begin{align}
\<\phi_{u,1}|\phi_{u,1}\>&=b_{u}^{2}+c_{u}^{2}=e^{-2(a+a')t}+\frac{1}{2}(1+e^{-4at}-2e^{-2(a+a')t})=\frac{1}{2}(1+e^{-4at}); \\
\<\phi_{u,2}|\phi_{u,2}\>&=\frac{d_{u}^{2}}{c_{u}^{2}}+(a_{u}^{2}-\frac{d_{u}^{2}}{c_{u}^{2}})=a_{u}^{2}=\frac{1}{2}(1-e^{-4at}).
\end{align}
In addition, for all $u\in\range{N}$, $\||\phi_{u,1}\>\|^{2}+\||\phi_{u,2}\>\|^{2}=1$ by \eqn{A-restriction-1} and \eqn{A-restriction-2}, and $\<\phi_{\nu(u),2}|\phi_{u,1}\>+\<\phi_{\nu(u),1}|\phi_{u,2}\>=0$ by \eqn{A-restriction-5}. Therefore, \eqn{defn-x1} and \eqn{defn-x2} give a construction of the 1-sparse Stinespring isometry as claimed.  Similarly to \eqn{diagonal-state-preparation}, all these ancilla states can be prepared using $O(\log N)$ 2-qubit gates.
\end{proof}

\section{Sparse Lindblad operators}
\label{sec:short-time}

In this section, we consider simulating a Lindbladian generated by a single Lindblad operator $L$:
\begin{align}\label{eqn:single-generator}
\mathcal{L}(\rho)=L\rho L^{\dagger}-\frac{1}{2}(L^{\dagger}L\rho+\rho L^{\dagger}L).
\end{align}
Of course, any Lindbladian can be decomposed into a sum of such terms.

We assume that the matrix $L$ in \eqn{single-generator} is both row $k$-sparse and column $k$-sparse, and $\|L\|_{\max}\leq 1$ (i.e., the largest entry of $L$ in absolute value is bounded by 1), which can always be obtained by rescaling. Under these assumptions, we can efficiently implement a quantum operation $\mathcal{E}_{\epsilon}$ that approximates $\mathcal{L}$ for a short time $\epsilon$ up to first order.

\begin{lemma}\label{lem:k-short-time}
Suppose we are given a black box for a $k$-sparse Lindblad operator $L$ that takes a row index or column index as input and outputs all nonzero entries in that row/column and their positions. Then for any $\epsilon>0$, we can efficiently implement a quantum operation $\mathcal{E}_{\epsilon}$ such that
\begin{align}\label{eqn:k-main-bound-generator}
\|(1+\epsilon\mathcal{L})-\mathcal{E}_{\epsilon}\|_{\diamond}=O(k^{4}\epsilon^{2}).
\end{align}
If $k=1$, we can implement $\mathcal{E}_{\epsilon}$ using $O(\log N)$ 2-qubit gates and $O(1)$ queries, with no error. More generally, for $k\geq 2$, we can implement $\mathcal{E}_{\epsilon}$ using $O\big(k^{2}[\log N+\log^{5/2}(k/\epsilon)]\frac{\log(k/\epsilon)}{\log\log(k/\epsilon)}\big)$ 2-qubit gates and $O(k^{2}\frac{\log(k/\epsilon)}{\log\log(k/\epsilon)})$ queries.
\end{lemma}

By concatenating a sequence of short-time evolutions, \lem{k-short-time} implies an efficient simulation of $\mathcal{L}$ for arbitrary times, as follows.

\begin{theorem}\label{thm:k-sparse-generator}
Suppose we are given a black box for $L$ that takes a row index or column index as input and outputs all nonzero entries in that row/column and their positions. If $L$ is 1-sparse, for any time $t>0$ we can simulate $\mathcal{L}$ within error $\epsilon$ using $O(\frac{t^{2}\log N}{\epsilon})$ 2-qubit gates and $O(\frac{t^{2}}{\epsilon})$ queries. If $L$ is $k$-sparse with $k\geq 2$, for any time $t>0$ we can simulate $\mathcal{L}$ within error $\epsilon$ using $O\big(\frac{k^{6}t^{2}}{\epsilon}[\log N+\log^{5/2}(k^{5}t/\epsilon)]\frac{\log(k^{5}t/\epsilon)}{\log\log(k^{5}t/\epsilon)}\big)$ 2-qubit gates and $O\big(\frac{k^{6}t^{2}}{\epsilon}\frac{\log(k^{5}t/\epsilon)}{\log\log(k^{5}t/\epsilon)}\big)$ queries.
\end{theorem}

\begin{proof}
By \lem{k-short-time} we have
\begin{align}\label{eqn:k-sparse-generator-1}
\|(1+\varepsilon\mathcal{L})^{t/\varepsilon}-\mathcal{E}_{\varepsilon}^{t/\varepsilon}\|_{\diamond}=O(k^{4}t\varepsilon).
\end{align}

By L'H{\^o}pital's rule,
\begin{align}
\lim_{\varepsilon\rightarrow 0}\frac{e^{\mathcal{L}t}-(1+\varepsilon\mathcal{L})^{t/\varepsilon}}{\varepsilon}=\frac{\text{d}}{\text{d}\varepsilon}\big(e^{\mathcal{L}t}-(1+\varepsilon\mathcal{L})^{t/\varepsilon}\big)\Big|_{\varepsilon=0}=
\frac{1}{2}t\mathcal{L}^{2}e^{\mathcal{L}t}.
\end{align}
Therefore,
\begin{align}\label{eqn:k-sparse-generator-2}
\|e^{\mathcal{L}t}-(1+\varepsilon\mathcal{L})^{t/\varepsilon}\|_{\diamond}=\frac{1}{2}\varepsilon t\|\mathcal{L}^{2}e^{\mathcal{L}t}\|_{\diamond}+O(\varepsilon^{2}).
\end{align}

Because $L$ is $k$-sparse and $\|L\|_{\max}=O(1)$, one can show that $\|\mathcal{L}\|_{\diamond}=O(k^{2})$; since $e^{\mathcal{L}t}$ is a quantum operation, $\|e^{\mathcal{L}t}\|_{\diamond}=1$. Thus $\|\mathcal{L}^{2}e^{\mathcal{L}t}\|_{\diamond}\leq \|\mathcal{L}\|_{\diamond}^{2}=O(k^{4})$. Therefore, combining \eqn{k-sparse-generator-1} and \eqn{k-sparse-generator-2}, we have
\begin{align}\label{eqn:k-sparse-generator-3}
\|e^{\mathcal{L}t}-\mathcal{E}_{\varepsilon}^{t/\varepsilon}\|_{\diamond}=O(k^{4}t\varepsilon).
\end{align}
Therefore, if we take $\varepsilon=\epsilon/(tk^{4})$, we find
\begin{align}\label{eqn:k-sparse-generator-4}
\|e^{\mathcal{L}t}-\mathcal{E}_{\epsilon/(tk^{4})}^{t^{2}k^{4}/\epsilon}\|_{\diamond}=O(\epsilon).
\end{align}
If $k=1$, since by \lem{k-short-time} we can implement $\mathcal{E}_{\epsilon/t}$ using $O(\log N)$ 2-qubit gates and $O(1)$ queries, we can repeat it $t^{2}/\epsilon$ times to simulate $e^{\mathcal{L}t}$ within error $\epsilon$ using using $O(\frac{t^{2}\log N}{\epsilon})$ 2-qubit gates and $O(\frac{t^{2}}{\epsilon})$ queries. If $k\geq 2$, since by \lem{k-short-time} we can implement $\mathcal{E}_{\epsilon/(tk^{4})}$ using $O\big(k^{2}[\log N+\log^{5/2}(k^{5}t/\epsilon)]\frac{\log(k^{5}t/\epsilon)}{\log\log(k^{5}t/\epsilon)}\big)$ 2-qubit gates and $O(k^{2}\frac{\log(k^{5}t/\epsilon)}{\log\log(k^{5}t/\epsilon)})$ queries, we can repeat it $t^{2}k^{4}/\epsilon$ times to simulate $e^{\mathcal{L}t}$ within error $\epsilon$ using $O\big(\frac{k^{6}t^{2}}{\epsilon}[\log N+\log^{5/2}(k^{5}t/\epsilon)]\frac{\log(k^{5}t/\epsilon)}{\log\log(k^{5}t/\epsilon)}\big)$ 2-qubit gates and $O\big(\frac{k^{6}t^{2}}{\epsilon}\frac{\log(k^{5}t/\epsilon)}{\log\log(k^{5}t/\epsilon)}\big)$ queries.
\end{proof}

When $L$ is 1-sparse, it is straightforward to implement the quantum operation $\mathcal{E}_{\epsilon}$ of \lem{k-short-time}, as explained in \sec{1-sparse-generator}. More generally, we show how to implement $\mathcal{E}_{\epsilon}$ for any $k$-sparse $L$ in \sec{k-sparse-generator}.

A higher-order version of this construction might lead to improved performance. We leave this as an open problem.

\subsection{1-sparse Lindblad operator}
\label{sec:1-sparse-generator}

\begin{proof}[Proof of \lem{k-short-time}]
If $L$ is 1-sparse, there exists a permutation $\nu\colon\range{N}\rightarrow\range{N}$ such that
\begin{align}
  \forall\,j\in\range{N} ~
  \exists\,c_j \in \C \quad \text{such that} \quad
L|j\>=c_{j}|\nu(j)\>.
\end{align}

For $0<\epsilon\leq 1/\|L\|_{\max}^{2}$, define the isometry
\begin{align}\label{eqn:1-defn-V}
V_{\epsilon}|m\>:=\sqrt{1-\epsilon |c_{m}|^{2}}|m\>|0\>+\sqrt{\epsilon}c_{m}|\nu(m)\>|1\>\quad\forall\,m\in\range{N}.
\end{align}
We show that the quantum operation $\mathcal{E}_{\epsilon}(\rho):=\Tr_{\anc}[V_{\epsilon}\rho V_{\epsilon}^{\dagger}]$ satisfies \lem{k-short-time}. To do this, by subadditivity of the trace norm it suffices to show that for an arbitrary pure state $|\Phi\>=\sum_{m,m'=0}^{N-1}b_{mm'}|m\>|m'\>$ where $\sum_{m,m'=0}^{N-1}|b_{mm'}|^{2}=1$,
\begin{align}\label{eqn:1-main-bound-2}
\big\|((1+\epsilon\mathcal{L})\otimes\mathcal{I}_{N})(|\Phi\>\<\Phi|)-\Tr_{\anc}[(V_{\epsilon}\otimes I_{N})|\Phi\>\<\Phi|(V_{\epsilon}^{\dagger}\otimes I_{N})]\big\|_{1}=O(\epsilon^{2}).
\end{align}

We have $L^{\dagger}|\nu(j)\>=c_{j}^{*}|j\>\ \forall\, j\in\range{N}$. Thus for any $m,n\in\range{N}$,
\begin{align}\label{eqn:1-short}
(1+\epsilon\mathcal{L})(|m\>\<n|)=|m\>\<n|+\epsilon\big(c_{m}c_{n}^{*}|\nu(m)\>\<\nu(n)|-\frac{1}{2}(|c_{m}|^{2}+|c_{n}|^{2})|m\>\<n|\big).
\end{align}
For convenience, denote $l_{m}:=1-\sqrt{1-\epsilon |c_{m}|^{2}}$ for each $m$. By \eqn{1-defn-V} and \eqn{1-short}, we have
\begin{align}
& (1+\epsilon\mathcal{L})(|m\>\<n|)-\Tr_{\anc}[V_{\epsilon}|m\>\<n|V_{\epsilon}^{\dagger}] \nonumber \\
&\quad= \big(1-\frac{\epsilon}{2}(|c_{m}|^{2}+|c_{n}|^{2})-\sqrt{(1-\epsilon |c_{m}|^{2})(1-\epsilon |c_{n}|^{2})}\big)|m\>\<n| \\
&\quad= \frac{1}{2}(l_{m}-l_{n})^{2}|m\>\<n|.
\end{align}
Therefore,
\begin{align}
&\ \big\|((1+\epsilon\mathcal{L})\otimes\mathcal{I}_{N})(|\Phi\>\<\Phi|)-\Tr_{\anc}[(V_{\epsilon}\otimes I_{N})|\Phi\>\<\Phi|(V_{\epsilon}^{\dagger}\otimes I_{N})]\big\|_{1}  \nonumber \\
&\quad=\ \Big\|\sum_{m,m',n,n'}b_{mm'}b_{nn'}^{*}\big((1+\epsilon\mathcal{L})(|m\>\<n|)-\Tr_{\anc}[V_{\epsilon}|m\>\<n|V_{\epsilon}^{\dagger}]\big)\otimes |m'\>\<n'|\Big\|_{1} \\
&\quad= \frac{1}{2}\Big\|\sum_{n,n'}\big(\sum_{m,m'}b_{mm'}b_{nn'}^{*}(l_{m}-l_{n})^{2}|m\>\<n|\big)\otimes |m'\>\<n'|\Big\|_{1} \\
&\quad\leq \Big\|\sum_{n,n'}\sum_{m,m'}b_{mm'}b_{nn'}^{*}l_{m}^{2}|m\>\<n|\otimes |m'\>\<n'|\Big\|_{1}+\Big\|\sum_{n,n'}\sum_{m,m'}b_{mm'}b_{nn'}^{*}l_{m}l_{n}|m\>\<n|\otimes |m'\>\<n'|\Big\|_{1} \\
&\quad= \sqrt{\sum_{m,m'}|b_{mm'}|^{2}l_{m}^{4}}\sqrt{\sum_{n,n'}|b_{nn'}|^{2}}+\sum_{m,m'}|b_{mm'}|^{2}l_{m}^{2} \\
&\quad= O(\epsilon^{2}),
\end{align}
where the last equality follows from $l_{m}=O(\epsilon)$ and $\sum_{m,m'=0}^{N-1}|b_{mm'}|^{2}=1$.

The isometry $V_{\epsilon}$ can be implemented by the following procedure:
\begin{align}
|m\>&\mapsto |m\>(\sqrt{1-\epsilon |c_{m}|^{2}}|0\>+\sqrt{\epsilon}c_{m}|1\>) \\
&\xmapsto{\text{control-}\nu} \sqrt{1-\epsilon |c_{m}|^{2}}|m\>|0\>+\sqrt{\epsilon}c_{m}|\nu(m)\>|1\>.
\end{align}
Here in the first step we compute $c_{m}$, prepare the ancilla qubit, and then uncompute $c_{m}$. To compute $\nu$ in place for the second step, we compute $|m\>\mapsto|m\>|\nu(m)\>$, and then perform the inverse of the map $|\nu(m)\>\mapsto|\nu(m)\>|\nu(\nu(m))\>=|\nu(m)\>|m\>$ to erase the $|m\>$ register. This procedure uses $O(\log N)$ 2-qubit gates and $O(1)$ queries, with no error.
\end{proof}

\subsection{$k$-sparse Lindblad operator}
\label{sec:k-sparse-generator}

If $L$ is $k$-sparse, there exist permutations
\begin{align}
\nu_{1},\nu_{2},\ldots,\nu_{k}\colon\range{N}\rightarrow\range{N}
\end{align}
such that
\begin{align}
L|j\>=\sum_{i=1}^{k}c_{j,i}|\nu_{i}(j)\>\quad\forall\,j\in\range{N},
\end{align}
where
\begin{align}
\nu_{i}(x)\neq\nu_{j}(x)\quad \forall\,x\in\range{N},\forall\,i\neq j.
\end{align}

For $0<\epsilon\leq 1/\|L\|_{\max}^{2}$, define $V_{\epsilon}\colon\C^{N}\rightarrow \C^{N}\otimes\C^{2}$ such that
\begin{align}\label{eqn:k-defn-V}
V_{\epsilon}|m\>:=\sum_{i=1}^{k}\sqrt{\epsilon}c_{m,i}|\nu_{i}(m)\>|1\>+|m\>|0\>-\epsilon\sum_{1\leq i\neq j\leq k}\frac{1}{2}c_{m,i}\cdot c_{\nu_{j}^{-1}(\nu_{i}(m)),j}^{*}|\nu_{j}^{-1}(\nu_{i}(m))\>|0\>
\end{align}
and $\mathcal{E}'_{\epsilon}(\rho):=\Tr_{\anc}[V_{\epsilon}\rho V_{\epsilon}^{\dagger}]$.

The construction of the quantum operation $\mathcal{E}_{\epsilon}$ in \lem{k-short-time} is based on two observations:
\begin{lemma}\label{lem:E-first-order-approximation}
$\mathcal{E}'_{\epsilon}$ approximates $e^{\mathcal{L}t}$ up to first order:
\begin{align}\label{eqn:k-isometry-bound}
\|(1+\epsilon\mathcal{L})-\mathcal{E}'_{\epsilon}\|_{\diamond}=O(k^{4}\epsilon^{2}).
\end{align}
\end{lemma}

\begin{lemma}\label{lem:square-error-spectral-norm}
$V_{\epsilon}^{\dagger}V_{\epsilon}$ approximates the identity up to first order in the spectral norm:
\begin{align}
\|V_{\epsilon}^{\dagger}V_{\epsilon}-I\|=O(k^{4}\epsilon^{2}).
\end{align}
\end{lemma}

The proofs of \lem{E-first-order-approximation} and \lem{square-error-spectral-norm} are shown in \sec{k-isometry-part-1} and \sec{k-isometry-part-2}, respectively.

Furthermore, $\mathcal{E}'_{\epsilon}$ can be approximately implemented by a sparse unitary transformation, which follows from a more general proposition:

\begin{proposition}\label{prop:sparse-isometry-implementation}
Suppose $d$ is a fixed positive integer, and let $V\colon\C^{N}\rightarrow \C^{N}\otimes\C^{d}$ be a matrix that is both row and column sparse. Furthermore, suppose $V$ is $\varepsilon$-close to an isometry, in the sense that
\begin{align}
\|V^{\dagger}V-I\| \le \varepsilon.
\end{align}
Define the map $\mathcal{E}'(\rho):=\Tr_{\anc}[V\rho V^{\dagger}]$. Then we can efficiently implement a quantum operation $\mathcal{E}$ such that
\begin{align}\label{eqn:k-unitary-bound}
\|\mathcal{E}-\mathcal{E}'\|_{\diamond}=O(\varepsilon).
\end{align}
\end{proposition}

\prop{sparse-isometry-implementation} shows that we can efficiently implement sparse approximate isometries, generalizing the implementation of sparse unitaries \cite{jordan2009efficient}.

\begin{proof}[Proof of \prop{sparse-isometry-implementation}]
Define
\begin{align}\label{eqn:k-defn-H}
H:=\left(
      \begin{array}{cc}
        0 & V \\
        V^{\dagger} & 0
      \end{array} \right)
\end{align}
to be a Hamiltonian in $\C^{(d+1)N}$. Consider the quantum operation
\begin{align}
\mathcal{E}(\rho):=\Tr_{\anc}[e^{-iH\frac{\pi}{2}}(|d\>\<d|\otimes\rho) e^{iH\frac{\pi}{2}}],
\end{align}
where $\mathcal{E}$ uses ancilla of dimension $d+1$ and $|d\>$ is the $(d+1)$st state of the computational basis of the ancilla system. $\mathcal{E}$ can be efficiently implemented because $H$ is a sparse Hamiltonian, which can be efficiently simulated. We will show that $\mathcal{E}$ satisfies \eqn{k-unitary-bound}.

By subadditivity of the trace norm, it suffices to prove \eqn{k-unitary-bound} for pure states, i.e., for any state $|\Phi\>=\sum_{m,m'=0}^{N-1}b_{mm'}|m\>|m'\>$ where $\sum_{m,m'=0}^{N-1}|b_{mm'}|^{2}=1$,
\begin{align}
\big\|\Tr_{\anc}\big[(e^{-iH\frac{\pi}{2}}\otimes I_{N})(|d\>\<d|\otimes|\Phi\>\<\Phi|)(e^{iH\frac{\pi}{2}}\otimes I_{N})-(V\otimes I_{N})|\Phi\>\<\Phi|(V^{\dagger}\otimes I_{N})\big]\big\|_{1}=O(\varepsilon).
\end{align}
Furthermore, since $(e^{-iH\frac{\pi}{2}}\otimes I_{N})(|d\>\<d|\otimes|\Phi\>\<\Phi|)(e^{iH\frac{\pi}{2}}\otimes I_{N})$ and $(V\otimes I_{N})|\Phi\>\<\Phi|(V^{\dagger}\otimes I_{N})$ are both rank-1 matrices, $(e^{-iH\frac{\pi}{2}}\otimes I_{N})(|d\>\<d|\otimes|\Phi\>\<\Phi|)(e^{iH\frac{\pi}{2}}\otimes I_{N})-(V\otimes I_{N})|\Phi\>\<\Phi|(V^{\dagger}\otimes I_{N})$ can be at most rank 2. Therefore, its trace norm must be upper-bounded by twice its spectral norm, and thus it suffices to prove
\begin{align}\label{eqn:k-unitary-spectral-bound-with-trace}
\big\|\Tr_{\anc}\big[(e^{-iH\frac{\pi}{2}}\otimes I_{N})(|d\>\<d|\otimes|\Phi\>\<\Phi|)(e^{iH\frac{\pi}{2}}\otimes I_{N})-(V\otimes I_{N})|\Phi\>\<\Phi|(V^{\dagger}\otimes I_{N})\big]\big\|=O(\varepsilon).
\end{align}
In turn, this shows that it suffices to prove
\begin{align}\label{eqn:k-unitary-spectral-bound}
\big\|(e^{-iH\frac{\pi}{2}}\otimes I_{N})(|d\>\<d|\otimes|\Phi\>\<\Phi|)(e^{iH\frac{\pi}{2}}\otimes I_{N})-(\tilde{V}\otimes I_{N})|\Phi\>\<\Phi|(\tilde{V}^{\dagger}\otimes I_{N})\big\|=O(\varepsilon),
\end{align}
where $\tilde{V}\colon\C^{N}\rightarrow \C^{N}\otimes\C^{d+1}$ such that $\tilde{V}=\bigl(\begin{smallmatrix}
        V \\
        0
\end{smallmatrix}\bigr)$, i.e., the first $dN$ rows of $\tilde{V}$ are identical to $V$, and the last $N$ rows of $\tilde{V}$ are all 0.

We claim that for any $m,n\in[N]$,
\begin{align}\label{eqn:k-unitary-spectral-bound-lemma}
\big\|e^{-iH\frac{\pi}{2}}(|d\>\<d|\otimes|m\>\<n|) e^{iH\frac{\pi}{2}}-\tilde{V}|m\>\<n|\tilde{V}^{\dagger}\big\|=O(\varepsilon).
\end{align}
For convenience, denote $P:=VV^{\dagger}$. Then
\begin{align}
\|P^{2}-P\|=\|VV^{\dagger}VV^{\dagger}-VV^{\dagger}\|=\|V(I+O(\varepsilon))V^{\dagger}-VV^{\dagger}\|=O(\varepsilon).
\end{align}
Furthermore,
\begin{align}
\|H^{3}-H\|=\biggl\|\left(
      \begin{array}{cc}
        0 & VV^{\dagger}V \\
        V^{\dagger}VV^{\dagger} & 0
      \end{array} \right)-\left(
      \begin{array}{cc}
        0 & V \\
        V^{\dagger} & 0
      \end{array} \right)\biggr\|=\biggl\|\left(
      \begin{array}{cc}
        0 & VO(\varepsilon) \\
        O(\varepsilon)V^{\dagger} & 0
      \end{array} \right)\biggr\|=O(\varepsilon).
\end{align}
By the same method, we can show that for any positive integer $n$,
\begin{align}\label{eqn:Hamiltonian-relationship}
\|H^{n+2}\|=\|H^{n}\|+O(\varepsilon)
\end{align}
where the big-$O$ constant is independent of $n$.

By \eqn{Hamiltonian-relationship}, for a fixed time $t>0$ we have
\begin{align}
&\|\cos(Ht)-(1-H^{2}+H^{2}\cos t)\| \nonumber \\
&\quad=\biggl\|1-\frac{H^{2}t^{2}}{2}+\sum_{n=2}^{\infty}(-1)^{n}\frac{H^{2n}t^{2n}}{(2n)!}-(1-\frac{H^{2}t^{2}}{2}+\sum_{n=2}^{\infty}(-1)^{n}\frac{H^{2}t^{2n}}{(2n)!})\biggr\| \\
&\quad=\biggl\|\sum_{n=2}^{\infty}(-1)^{n}\frac{(H^{2n}-H^{2})t^{2n}}{(2n)!}\biggr\| \\
&\quad\leq\sum_{n=2}^{\infty}\frac{t^{2n}\|H^{2n}-H^{2}\|}{(2n)!} \\
&\quad=\sum_{n=2}^{\infty}\frac{t^{2n}nO(\varepsilon)}{(2n)!} \\
&\quad=O(\varepsilon).
\end{align}
A similar calculation shows that
\begin{align}
\|\sin(Ht)-H\sin t\|=O(\varepsilon).
\end{align}
Therefore,
\begin{align}\label{eqn:Hamiltonian-error-analysis}
\|e^{-iHt}-(1-H^{2}+H^{2}\cos t)+iH\sin t\|=O(\varepsilon).
\end{align}
Taking $t=\frac{\pi}{2}$ in \eqn{Hamiltonian-error-analysis}, we have
\begin{align}\label{eqn:Hamiltonian-error-analysis-pi/2}
\|e^{-iH\frac{\pi}{2}}-(1-iH-H^{2})\|=O(\varepsilon).
\end{align}
As a block matrix, we have
\begin{align}\label{eqn:Hamiltonian-error-analysis-matrix}
1-iH-H^{2}=\left(
      \begin{array}{cc}
        I-VV^{\dagger} & -iV \\
       -iV^{\dagger} & I-V^{\dagger}V
      \end{array} \right).
\end{align}
Plugging \eqn{Hamiltonian-error-analysis-pi/2} and \eqn{Hamiltonian-error-analysis-matrix} into \eqn{k-unitary-spectral-bound-lemma}, and by the definition of $\tilde{V}$, we have
\begin{align}
&\big\|e^{-iH\frac{\pi}{2}}(|d\>\<d|\otimes|m\>\<n|) e^{iH\frac{\pi}{2}}-\tilde{V}|m\>\<n|\tilde{V}^{\dagger}\big\| \nonumber \\
&\quad=\biggl\|e^{-iH\frac{\pi}{2}}(|d\>\<d|\otimes|m\>\<n|) e^{iH\frac{\pi}{2}}-\left(
      \begin{array}{cc}
        V|m\>\<n|V^{\dagger} & 0 \\
        0 & 0
      \end{array} \right)\biggr\| \\
&\quad=\biggl\|\left(
      \begin{array}{c}
        -iV|m\> \\
        (I-V^{\dagger}V)|m\>
      \end{array} \right)\left(
      \begin{array}{cc}
        i\<n|V^{\dagger} & \<n|(I-V^{\dagger}V)
      \end{array} \right)-\left(
      \begin{array}{cc}
        V|m\>\<n|V^{\dagger} & 0 \\
        0 & 0
      \end{array} \right)\biggr\|+O(\varepsilon) \\
&\quad=\biggl\|\left(
      \begin{array}{cc}
        0 & -iV|m\>\<n|(I-V^{\dagger}V) \\
        i(I-V^{\dagger}V)|m\>\<n|V^{\dagger} & (I-V^{\dagger}V)|m\>\<n|(I-V^{\dagger}V)
      \end{array} \right)\biggr\|+O(\varepsilon) \\
&\quad=O(\varepsilon),
\end{align}
where the last equality comes from the assumption that $\|I-V^{\dagger}V\|=O(\varepsilon)$.

Now, for any $m\in[N]$, let
\begin{align}
|\Psi_{m}\>&:=e^{-iH\frac{\pi}{2}}|d\>|m\>, \label{eqn:k-unitary-spectral-bound-states-1} \\
|\Phi_{m}\>&:=\tilde{V}|m\>. \label{eqn:k-unitary-spectral-bound-states-2}
\end{align}
Equation \eqn{k-unitary-spectral-bound-lemma} ensures that for any normalized pure states $|\gamma_{1}\>,|\gamma_{2}\> \in \C^{d+1}\otimes\C^{N}$,
\begin{align}\label{eqn:normalized-difference}
\big|\<\gamma_{1}|\Psi_{m}\>\<\Psi_{n}|\gamma_{2}\>-\<\gamma_{1}|\Phi_{m}\>\<\Phi_{n}|\gamma_{2}\>\big|=O(\varepsilon)\quad\forall\,m,n\in\range{N}.
\end{align}

To prove \eqn{k-unitary-spectral-bound} it suffices to show that for any pure state $|\Gamma\>\in \C^{d+1}\otimes\C^{N}\otimes\C^{N}$,
\begin{align}\label{eqn:k-unitary-spectral-bound-2}
\big|\<\Gamma|\big((e^{-iH\frac{\pi}{2}}\otimes I_{N})(|d\>\<d|\otimes|\Phi\>\<\Phi|)(e^{iH\frac{\pi}{2}}\otimes I_{N})-(\tilde{V}\otimes I_{N})|\Phi\>\<\Phi|(\tilde{V}^{\dagger}\otimes I_{N})\big)|\Gamma\>\big|=O(\varepsilon).
\end{align}
By the Schmidt decomposition, we can write $|\Gamma\>=\sum_{m'=0}^{N-1}c_{m'}|\gamma_{m'}\>|m'\>$ where $|\gamma_{m'}\>$ are orthonormal pure states in $\C^{d+1}\otimes\C^{N}$ for $m'\in\range{N}$, $c_{m'}\geq 0$ for any $m'\in\range{N}$, and $\sum_{m'=0}^{N-1}c_{m'}^{2}=1$. Then, \eqn{k-unitary-spectral-bound-2} holds because
\begin{align}
&\big|\<\Gamma|\big((e^{-iH\frac{\pi}{2}}\otimes I_{N})(|d\>\<d|\otimes|\Phi\>\<\Phi|)(e^{iH\frac{\pi}{2}}\otimes I_{N})-(\tilde{V}\otimes I_{N})|\Phi\>\<\Phi|(\tilde{V}^{\dagger}\otimes I_{N})\big)|\Gamma\>\big| \nonumber \\
&\quad= \Big|\sum_{m,m'}\sum_{n,n'}b_{mm'}b_{nn'}^{*}c_{m'}c_{n'}\big(\<\gamma_{m'}|\Psi_{m}\>\<\Psi_{n}|\gamma_{n'}\>-\<\gamma_{m'}|\Phi_{m}\>\<\Phi_{n}|\gamma_{n'}\>\big)\Big| \label{eqn:spectral-bound-normalization-1} \\
&\quad\le \sum_{m,m'}\sum_{n,n'}|b_{mm'}b_{nn'}^{*}c_{m'}c_{n'}|\cdot\big|\<\gamma_{m'}|\Psi_{m}\>\<\Psi_{n}|\gamma_{n'}\>-\<\gamma_{m'}|\Phi_{m}\>\<\Phi_{n}|\gamma_{n'}\>\big| \\
&\quad\le \sqrt{\sum_{m,m'}\sum_{n,n'}|b_{mm'}b_{nn'}^{*}c_{m'}c_{n'}|^{2}}\cdot O(\varepsilon) \label{eqn:spectral-bound-normalization-2} \\
&\quad= \sqrt{\sum_{m,m'}|b_{mm'}|^{2}c_{m'}^{2}}\sqrt{\sum_{n,n'}|b_{nn'}|^{2}c_{n'}^{2}}\cdot O(\varepsilon) \\
&\quad\le  \sqrt{\sum_{m,m'}|b_{mm'}|^{2}}\sqrt{\sum_{n,n'}|b_{nn'}|^{2}}\cdot O(\varepsilon) \\
&\quad= O(\varepsilon) \label{eqn:spectral-bound-normalization-3},
\end{align}
where \eqn{spectral-bound-normalization-1} comes from \eqn{k-unitary-spectral-bound-states-1} and \eqn{k-unitary-spectral-bound-states-2}, \eqn{spectral-bound-normalization-2} comes from \eqn{normalized-difference} and orthogonality of the states $|\gamma_{m'}\>$, and \eqn{spectral-bound-normalization-3} comes from the assumption that $\sum_{m}\sum_{m'}|b_{mm'}|^{2}=1$.
\end{proof}

Finally, we use \lem{E-first-order-approximation}, \lem{square-error-spectral-norm}, and \prop{sparse-isometry-implementation} to prove \lem{k-short-time}.

\begin{proof}[Proof of \lem{k-short-time}]
Plugging \lem{square-error-spectral-norm} into \prop{sparse-isometry-implementation} by taking $V=V_{\epsilon}$, $\varepsilon=k^{4}\epsilon^{2}$, and $d=2$, we know that
\begin{align}
\|\mathcal{E}_{\epsilon}-\mathcal{E}'_{\epsilon}\|_{\diamond}=O(k^{4}\epsilon^{2}),
\end{align}
where $\mathcal{E}_{\epsilon}(\rho):=\Tr_{\anc}[e^{-iH\frac{\pi}{2}}(|2\>\<2|\otimes\rho) e^{iH\frac{\pi}{2}}]$. Together with \lem{E-first-order-approximation}, we have shown that the quantum operation $\mathcal{E}_{\epsilon}$ satisfies \lem{k-short-time}. Furthermore, $H$ is a $(k^{2}+1)$-sparse Hamiltonian, so by Theorem 1 of \cite{berry2015hamiltonian}, $\mathcal{E}_{\epsilon}$ can be implemented using $O\big(k^{2}[\log N+\log^{5/2}(k/\epsilon)]\frac{\log(k/\epsilon)}{\log\log(k/\epsilon)}\big)$ 2-qubit gates and $O(k^{2}\frac{\log(k/\epsilon)}{\log\log(k/\epsilon)})$ queries to the black box.
\end{proof}

\subsection{Truncated damped quantum harmonic oscillators}
\label{sec:dho}

In this section we briefly describe an application of the above simulation, namely to simulating a damped quantum harmonic oscillator truncated to the first $N$ levels in Fock space. This system is described by the creation operator $a^{\dagger}$ and its Hermitian conjugate, the annihilation operator $a$, satisfying
\begin{align}
a^{\dagger}|j\>=\sqrt{j+1}|j+1\>\qquad a|j\>=\sqrt{j}|j-1\>,
\end{align}
where $|n\>=0$ if $n>N-1$. Then the Lindbladians
\begin{align}
\mathcal{L}_{\mathrm{up}}(\rho)&:=\frac{1}{N}\big(a^{\dagger}\rho a-\frac{1}{2}(aa^{\dagger}\rho+\rho aa^{\dagger})\big) \\
\mathcal{L}_{\mathrm{down}}(\rho)&:=\frac{1}{N}\big(a\rho a^{\dagger}-\frac{1}{2}(a^{\dagger}a\rho+\rho a^{\dagger}a)\big)
\end{align}
are not local, but they can be efficiently simulated by \thm{k-sparse-generator}:
\begin{corollary}
For any time $t>0$, the Lindbladians $\mathcal{L}_{\mathrm{up}}$ and $\mathcal{L}_{\mathrm{down}}$ can both be simulated within error $\epsilon$ using $O(\frac{t^{2}\log N}{\epsilon})$ 2-qubit gates.
\end{corollary}

The Lindbladian $\mathcal{L}_{\mathrm{down}}$ represents damping at zero temperature, whereas $\mathcal{L}_{\mathrm{up}}$ represents damping at infinite temperature. More generally, Lindbladians of the form $\lambda\mathcal{L}_{\mathrm{up}}+(1-\lambda)\mathcal{L}_{\mathrm{down}}$ for $0<\lambda<1$ represent generalized damping at finite temperature, and can also be simulated efficiently.

\section{Open quantum walks}
\label{sec:open-quantum-walk}
This section briefly discusses possible applications of our simulation methods to implementations of non-unitary quantum walks.

\subsection{Quantum stochastic walks}

\label{sec:stochastic-walk}
Reference \cite{whitfield2010quantum} introduced the notion of quantum stochastic walks, a mutual generalization of classical random walks and quantum walks. Specifically, given a graph $G=(V,E)$ and a continuous-time classical Markov chain on $G$ with transition rate matrix $M$, consider Lindblad operators $L_{k,l}=\sqrt{M_{k,l}}|k\>\<l|$ for each $k\neq l$. This gives a diagonal Lindbladian $\mathcal{L}$ of the form of equation \eqn{GKS'-diagonal}, where for each $m$ we have
\begin{align}\label{eqn:quantum-stochastic-walk-Lindblad}
\mathcal{L}(|m\>\<m|)=\sum_{k\neq m}M_{k,m}(|k\>\<k|-|m\>\<m|).
\end{align}
In other words, $\mathcal{L}$ simulates the classical Markov chain on diagonal states. On the other hand, a continuous-time quantum walk is characterized by a Hamiltonian $\mathcal{H}(\rho):=-i[H,\rho]$, which can appear as part of a Lindbladian as in \eqn{Lindblad} and \eqn{GKS}. By taking a positive linear combination of $\mathcal{H}$ and $\mathcal{L}$, we obtain a quantum stochastic walk, generalizing classical random walks and quantum walks.

If the underlying graph $G$ of the Lindbladian $\mathcal{L}$ in \eqn{quantum-stochastic-walk-Lindblad} is sparse, then $\mathcal{L}$ is a sparse-diagonal Lindbladian and it can be efficiently simulated by \thm{d-sparse}. This shows that quantum stochastic walks on sparse graphs can be carried out in practice. To the best of our knowledge, this was not known previously.

As a concrete example, consider the unweighted random walk on $G$ defined by
\begin{align}
M_{k,l}=\left\{
\begin{array}{ll}
\frac{1}{\deg(l)} & \text{if } (k,l)\in E \\
0 & \text{if } (k,l)\notin E.
\end{array} \right.
\end{align}
We can use quantum stochastic walks to prepare the stationary state of such a process:
\begin{theorem}[\cite{liu2016continuous}]\label{thm:quantum-stochastic-walk-relaxing}
If $G$ is connected, then for any Hamiltonian $\mathcal{H}$, the Lindbladian $\mathcal{H}+\mathcal{L}$ is relaxing, i.e., there exists a unique stationary state $\rho_{\infty}$ such that $e^{\mathcal{L}t}(\rho_{\infty})=\rho_{\infty}$ for any time $t$ and $\lim_{t\rightarrow\infty}e^{\mathcal{L}t}(\rho)=\rho_{\infty}$ for any initial state $\rho$. Furthermore, if $G$ is both connected and regular, then $\rho_{\infty}=\frac{1}{|V|}I$.
\end{theorem}

If $G$ is connected but not regular, then the state $\rho_{\infty}$ in \thm{quantum-stochastic-walk-relaxing} can be complicated even though $G$ is not. For instance, if $V=\{v_{1},v_{2},v_{3}\}$, $E=\{(v_{1},v_{2}),(v_{2},v_{3})\}$, and the Hamiltonian $H$ is the Laplacian of $G$, then \cite{liu2016continuous} showed that
\begin{align}
\rho_{\infty}=\left(
\begin{array}{ccc}
\frac{2}{7} & -\frac{1}{28}+\frac{1}{28}i & \frac{1}{14} \\
-\frac{1}{28}-\frac{1}{28}i & \frac{3}{7} & -\frac{1}{28}-\frac{1}{28}i \\
\frac{1}{14} & -\frac{1}{28}+\frac{1}{28}i & \frac{2}{7}
\end{array} \right).
\end{align}
This type of behavior may be of interest for quantum state engineering or dissipative quantum computation \cite{kraus2008preparation,verstraete2009quantum}.

\subsection{Decoherence in quantum walks on the hypercube}

\label{sec:decoherence-walk}
There is an extensive literature on the effects of decoherence on quantum walks. In particular, references \cite{alagic2005decoherence,strauch2009reexamination} discussed models of decoherence in quantum walks on the hypercube, the graph with vertex set $V=\{0,1\}^{n}$ and edge set $E=\{(x,y)\in V^{2}:\Delta(x,y)=1\}$, where $\Delta(x,y)$ denotes the Hamming distance between the strings $x$ and $y$. The adjacency matrix of the hypercube is
\begin{align}
A=\sum_{j=1}^{n}\sigma_{x}^{(j)}
\end{align}
where $\sigma_{x}^{(j)}$ denotes the Pauli $X$ operator acting on the $j$th qubit.

For simplicity, we consider the quantum walk with the Hamiltonian given by $A$ instead of the Laplacian, because the hypercube is a regular graph and the walks generated by the adjacency matrix and the Laplacian only differ by a global phase.

Reference \cite{alagic2005decoherence} considered the decoherence model with Lindblad operators
\begin{align}\label{eqn:decoherence-Alagic}
L_{j,\alpha} := I \otimes \cdots \otimes I \otimes \Pi_{\alpha} \otimes I \otimes \cdots \otimes I,
\end{align}
where $\alpha\in\{0,1\}$, $j\in\{1,2,\ldots,n\}$, and $\Pi_{\alpha}:=|\alpha\>\<\alpha|$ acts on the $j$th qubit. The Lindbladian is
\begin{align}\label{eqn:identical-coordinate-instance-1}
\mathcal{L}(\rho):=-(1-p)i[A,\rho]+p\sum_{j,\alpha}\big(L_{j,\alpha}\rho L_{j,\alpha}^{\dagger}-\frac{1}{2}L_{j,\alpha}^{\dagger}L_{j,\alpha}\rho-\frac{1}{2}\rho L_{j,\alpha}^{\dagger}L_{j,\alpha}\big),
\end{align}
where $0<p<4/n$ is the probability of decoherence per unit time. This $\mathcal{L}$ can be efficiently simulated using \thm{local-Lindbladians} and the fact that efficient simulation is closed under positive linear combinations.

Reference \cite{strauch2009reexamination} considered another decoherence model with Lindblad operators
\begin{align}\label{eqn:decoherence-Strauch}
L_{x}:=\Pi_{x_{1}}\otimes\cdots\otimes \Pi_{x_{n}},
\end{align}
where $x\in\{0,1\}^{n}$. The Lindbladian is
\begin{align}\label{eqn:identical-coordinate-instance-2}
\mathcal{L}(\rho):=-i[A,\rho]+\lambda\sum_{x}\big(L_{x}\rho L_{x}^{\dagger}-\frac{1}{2}L_{x}^{\dagger}L_{x}\rho-\frac{1}{2}\rho L_{x}^{\dagger}L_{x}\big),
\end{align}
where $\lambda\ll 1$. The matrix $A$ in \eqn{GKS'} of this $\mathcal{L}$ satisfies $A_{(x,x),(x,x)}=\frac{1}{2}\lambda$ for all $x \in \{0,1\}^n$ and $A_{(k,l),(k',l')}=0$ otherwise. By taking $a_{x}=\frac{1}{2}\lambda$ and $c_{x}=0$ for all $x$ in \thm{singly-sparse}, we see that this is an identical-coordinate Lindbladian, so it can be efficiently simulated.

References \cite{alagic2005decoherence,strauch2009reexamination} studied the hitting and mixing times of these walks, showing that they can sometimes achieve the fast hitting time of a quantum walk while also retaining fast convergence to a uniform distribution as in a classical random walk. Our work shows that this process can be implemented efficiently, and we hope this may lead to practical applications of open quantum walks.

\section{Linear limit on simulation time}
\label{sec:no-fast-forwarding}

Reference \cite{berry2007efficient} established a so-called no--fast-forwarding theorem for Hamiltonian simulation, showing that simulating a Hamiltonian $H$ for time $t$ with constant precision requires $\Omega(t)$  queries to a black box description of $H$. This result shows that one cannot fast-forward the dynamics of a generic Hamiltonian. Here we extend this result to Lindbladians in two different query models.

\begin{theorem}[No--fast-forwarding theorem for Lindbladians]\label{thm:no-fast-forward}
For any positive integer $N\geq 7$, there exists a Lindbladian $\mathcal{L}$ generated by a single Lindblad operator $L$ such that $\|\mathcal{L}\|_{\diamond}=1$ and simulating the evolution of $\mathcal{L}$ for time $t=4N$ within precision 1/8 requires at least
\begin{enumerate}[itemsep=0pt,topsep=4pt,label=(\roman*)]
\item $N/4=t/16$ queries to a black box that takes both a row and column index of the matrix $A$ of $\mathcal{L}$ defined in \eqn{GKS'} and outputs the corresponding entry in $A$, or
\item $N/2=t/8$ queries to a black box that takes a row or column index of $L$ as input and outputs the locations and values of all nonzero elements in that row or column.
\end{enumerate}
\end{theorem}

To prove \thm{no-fast-forward}, we first establish two lemmas:

\begin{lemma}\label{lem:exp_tail_analysis}
For all positive integers $N\geq 7$,
\begin{align}\label{eqn:exp_tail_analysis}
\sum_{m=0}^{N-1}\frac{(2N)^{m}e^{-2N}}{m!}\leq\frac{1}{64}.
\end{align}
\end{lemma}
\begin{proof}
When $N=7,8,9,10,11,12$, \eqn{exp_tail_analysis} holds. It suffices to prove \eqn{exp_tail_analysis} when $N\geq 13$.

A theorem of Ramanujan \cite[Question 294]{hardy1927collected} states that for any positive integer $M$,
\begin{align}\label{eqn:Ramanujan}
\frac{1}{2}e^{M}=\sum_{m=0}^{M-1}\frac{M^{m}}{m!}+\theta(M)\cdot\frac{M^{M}}{M!},
\end{align}
where $\frac{1}{3}\leq\theta(M)\leq\frac{1}{2}\ \forall\, M$. Therefore,
\begin{align}\label{eqn:Ramanujan-2}
\frac{1}{2}e^{2N}\geq\sum_{m=0}^{2N-1}\frac{(2N)^{m}}{m!}.
\end{align}

If
\begin{align}\label{eqn:exp_tail_analysis-2}
\sum_{m=N}^{2N-1}\frac{(2N)^{m}}{m!}\geq 31\cdot \sum_{m=0}^{N-1}\frac{(2N)^{m}}{m!},
\end{align}
then \eqn{Ramanujan-2} gives
\begin{align}
\frac{1}{2}e^{2N}\geq\sum_{m=0}^{N-1}\frac{(2N)^{m}}{m!}+31\cdot \sum_{m=0}^{N-1}\frac{(2N)^{m}}{m!}\quad\Rightarrow\quad\sum_{m=0}^{N-1}\frac{(2N)^{m}e^{-2N}}{m!}\leq\frac{1}{64}.
\end{align}
Hence it suffices to prove \eqn{exp_tail_analysis-2}, which holds provided
\begin{align}
\frac{(2N)^{m+N}}{(m+N)!}\geq 31\cdot \frac{(2N)^{m}}{m!}\quad\forall\,m\in\range{N}.
\end{align}
Moreover, since
\begin{align}
\frac{\frac{(2N)^{m+1+N}}{(m+1+N)!}/\frac{(2N)^{m+1}}{(m+1)!}}{\frac{(2N)^{m+N}}{(m+N)!}/\frac{(2N)^{m}}{m!}}=\frac{m+1}{m+1+N}<1\quad\forall\,m\in\range{N},
\end{align}
it suffices to prove
\begin{align}\label{eqn:exp_tail_analysis-3}
\frac{(2N)^{2N}}{(2N)!}\geq 31\cdot\frac{(2N)^{N}}{N!}.
\end{align}

Define $g(N):=\frac{(2N)^{2N}}{(2N)!}/\frac{(2N)^{N}}{N!}$. Since $g(13)\approx 38.3102>31$ and
\begin{align}
\frac{g(N+1)}{g(N)}=\frac{N+1}{2N+1}\cdot(1+\frac{1}{N})^{N}>\frac{1}{2}\cdot 2=1\quad \forall\,N,
\end{align}
we have $g(N)>31$ for $N\geq 13$. Therefore, \eqn{exp_tail_analysis-3} holds and \eqn{exp_tail_analysis} follows.
\end{proof}

\begin{lemma}\label{lem:Jordan_exp}
Consider the Lindblad operator $L$ acting in $\C^{N+1}$ such that
\begin{align}
L|0\>=0,\quad L|n\>=|n-1\>\quad\forall\,n\in\{1,\ldots,N\}.
\end{align}
Then for any time $t>0$, the Lindbladian $\mathcal{L}(\rho):=L\rho L^{\dagger}-\frac{1}{2}(L^{\dagger}L\rho+\rho L^{\dagger}L)$ satisfies
\begin{align}\label{eqn:Jordan_exp}
e^{\mathcal{L}t}(|N\>\<N|)=\sum_{m=0}^{N-1}\frac{t^{m}e^{-t}}{m!}|N-m\>\<N-m|+\Big(1-\sum_{m=0}^{N-1}\frac{t^{m}e^{-t}}{m!}\Big)|0\>\<0|.
\end{align}
\end{lemma}
\begin{proof}
We have
\begin{align}
\mathcal{L}(|n\>\<n|)&=|n-1\>\<n-1|-|n\>\<n|\quad \forall\,n\in\{1,2,\ldots,N\}; \label{eqn:Jordan_exp-1} \\
\mathcal{L}(|0\>\<0|)&=0. \label{eqn:Jordan_exp-2}
\end{align}

Denote $C:=(|N\>\<N|,|N-1\>\<N-1|,\ldots,|0\>\<0|)^{\dagger}$, and let $A$ be the $(N+1)\times(N+1)$ matrix with nonzero entries $A_{n,n}=-1$ and $A_{n,n+1}=1\ \forall\, n\in\range{N}$. In other words,
\begin{align}
A=\left(
      \begin{array}{cccc}
        -1 & 1 & & \\
        & \ddots & \ddots & \\
        & & -1 & 1 \\
        & & & 0
      \end{array} \right).
\end{align}
Then $\mathcal{L}(C)=AC$, so
\begin{align}
e^{\mathcal{L}t}(C)=e^{tA}C.
\end{align}
Since every row of $A$ has zero sum, $\textbf{e}=(1,1,\ldots,1)^{\dagger}$ satisfies $A\textbf{e}=0$. Therefore, $e^{tA}\textbf{e}=\textbf{e}$, so
\begin{align}\label{eqn:Jordan_exp-3}
(e^{tA})_{0,N}=1-\sum_{m=0}^{N-1}(e^{tA})_{0,m}.
\end{align}
By induction, for arbitrary $n\in\N$ we have
\begin{align}
(A^{n})_{0,m}={n\choose m}(-1)^{n+m}\quad\forall\,m\in\range{N}.
\end{align}
Therefore,
\begin{align}\label{eqn:Jordan_exp-4}
(e^{tA})_{0,m}=\sum_{n=0}^{\infty}\frac{(A^{n})_{0,m}t^{n}}{n!}=\sum_{n=0}^{\infty}\frac{{n\choose m}(-1)^{n+m}t^{n}}{n!}=\frac{t^{m}}{m!}\sum_{n=m}^{\infty}\frac{(-1)^{n-m}t^{n-m}}{(n-m)!}=\frac{t^{m}e^{-t}}{m!}.
\end{align}
Combining \eqn{Jordan_exp-3} and \eqn{Jordan_exp-4}, we get
\begin{align}\label{eqn:Jordan_exp-5}
(e^{tA})_{0,N}=1-\sum_{m=0}^{N-1}\frac{t^{m}e^{-t}}{m!}.
\end{align}
Finally, \eqn{Jordan_exp} follows from \eqn{Jordan_exp-4} and \eqn{Jordan_exp-5}.
\end{proof}

\begin{proof}[Proof of \thm{no-fast-forward}]
We now construct a Lindbladian whose dynamics compute the parity of a string $s\in\{0,1\}^{N}$, analogous to the Hamiltonian constructed in \cite{berry2007efficient}. Specifically, let
\begin{align}\label{eqn:L_s-definition}
\mathcal{L}_{s}(\rho)=\frac{1}{2}\big(L_{s}\rho L_{s}^{\dagger}-\frac{1}{2}(L_{s}^{\dagger}L_{s}\rho+\rho L_{s}^{\dagger}L_{s})\big)
\end{align}
be a Lindbladian acting on the $2(N+1)$-dimensional Hilbert space $\mathcal{H}\otimes\C^{2}$, where $s=s_{1}s_{2}\cdots s_{N}$ is a binary string and
\begin{align}
L_{s}|0,j\>=0,\quad L_{s}|n,j\>=|n-1,j\oplus s_{n}\>\quad\forall\,n\in\{1,2,\ldots,N\},\, j\in\{0,1\}.
\end{align}

First we prove $\|\mathcal{L}_{s}\|_{\diamond}=1$. Another commonly used norm for superoperators is the \emph{1-to-1 norm} defined as
\begin{align}\label{eqn:1-to-1-norm}
\|\mathcal{T}\|_{1\rightarrow 1}:=\max_{\|\rho\|_{1}=1}\|\mathcal{T}(\rho)\|_{1}.
\end{align}
Let $\mathcal{I}_{m}$ denote the identity superoperator on $\C^{m}$. Since $2\mathcal{L}_{s}$ is a permutation of $\mathcal{L}\otimes\mathcal{I}_{2}$ where $\mathcal{L}$ is the Lindbladian in \lem{Jordan_exp},
\begin{align}
2\|\mathcal{L}_{s}\|_{\diamond}=\|\mathcal{L}\otimes\mathcal{I}_{2}\|_{\diamond}=\|\mathcal{L}\otimes\mathcal{I}_{2}\otimes\mathcal{I}_{2(N+1)}\|_{1\rightarrow 1}=\|\mathcal{L}\otimes\mathcal{I}_{4(N+1)}\|_{1\rightarrow 1}=\|\mathcal{L}\otimes\mathcal{I}_{N+1}\|_{1\rightarrow 1}
\end{align}
where the last equality comes from the fact that $\|\mathcal{L}\otimes\mathcal{I}_{m}\|_{1\rightarrow 1}=\|\mathcal{L}\otimes\mathcal{I}_{N+1}\|_{1\rightarrow 1}$ for any $m\geq N+1$ (see Section 3.3 of \cite{watrous2011theory}). Thus it suffices to show that
\begin{align}\label{eqn:L-1-to-1-norm}
\|\mathcal{L}\otimes\mathcal{I}_{N+1}\|_{1\rightarrow 1}=2.
\end{align}

On the one hand, $\big\|\mathcal{L}(|N\>\<N|)\big\|_{1}=\big\||N-1\>\<N-1|-|N\>\<N|\big\|_{1}=2$, so
\begin{align}\label{eqn:L-1-to-1-norm-geq}
\|\mathcal{L}\otimes\mathcal{I}_{N+1}\|_{1\rightarrow 1}\geq \|\mathcal{L}\|_{1\rightarrow 1}\geq 2.
\end{align}
On the other hand, for any pure state $|\Phi\>=\sum_{m,m'=0}^{N}b_{mm'}|m\>|m'\>$,
\begin{align}
\|(L\otimes I_{N+1})|\Phi\>\<\Phi|(L^{\dagger}\otimes I_{N+1})\|_{1}&=\big\|\sum_{m,n=1}^{N}\sum_{m',n'=0}^{N} b_{mm'}b_{nn'}^{*}|m-1\>\<n-1|\otimes|m'\>\<n'|\big\|_{1} \\
&=\big\|\sum_{m=1}^{N}\sum_{m'=0}^{N} b_{mm'}|m-1\>|m'\>\big\|^{2} \\
&=\sum_{m=1}^{N}\sum_{m'=0}^{N}|b_{mm'}|^{2} \\
&\leq\sum_{m=0}^{N}\sum_{m'=0}^{N}|b_{mm'}|^{2} \\
&=1.
\end{align}
Similarly,
\begin{align}
\|(L^{\dagger}\otimes I_{N+1})(L\otimes I_{N+1})|\Phi\>\<\Phi|\|_{1},\||\Phi\>\<\Phi|(L^{\dagger}\otimes I_{N+1})(L\otimes I_{N+1})\|_{1}\leq 1.
\end{align}
Therefore, for any density matrix $\rho\in\linop{L}(\C^{N+1}\otimes \C^{N+1})$, by subadditivity of the trace norm
\begin{align}\label{eqn:L-1-to-1-norm-leq}
\|\mathcal{L}\otimes\mathcal{I}_{N+1}(\rho)\|_{1}\leq 1+\frac{1}{2}(1+1)=2 \quad\Rightarrow\quad\|\mathcal{L}\otimes\mathcal{I}_{N+1}\|_{1\rightarrow 1}\leq 2.
\end{align}

Clearly \eqn{L-1-to-1-norm} follows from \eqn{L-1-to-1-norm-geq} and \eqn{L-1-to-1-norm-leq}. Therefore, we have shown $\|\mathcal{L}_{s}\|_{\diamond}=1$.

To complete the proof, we show how a simulation can be used to compute the parity of $s$.
Let $\textsf{PARITY}(s):=s_{1}\oplus s_{2}\oplus\cdots\oplus s_{N}$. We simulate $\mathcal{L}_{s}$ with initial state $|N,0\>\<N,0|$, and after time $t=4N$ we measure the system. By \lem{exp_tail_analysis}, \lem{Jordan_exp}, and \eqn{L_s-definition}, the probability that the measurement result is not $(0,\textsf{PARITY}(s))$ is
\begin{align}\label{eqn:not-parity-s-bound}
p\leq\frac{1}{64}.
\end{align}

In the simulation of $\mathcal{L}_{s}$, denote the final state of the ancilla by $\rho_{\anc}$. Let
\begin{align}
\rho_{s}&:=|\textsf{PARITY}(s)\>\<\textsf{PARITY}(s)|, \\
\rho_{s,\textrm{sim}}&:=\Tr_{\mathcal{H}}[e^{\mathcal{L}_{s}t}(|N,0\>\<N,0|)].
\end{align}
By the definition of $p$, we have
\begin{align}\label{eqn:parity-p-q}
\<\textsf{PARITY}(s)|\rho_{s,\textrm{sim}}|\textsf{PARITY}(s)\>\geq 1-p.
\end{align}
By \cite{beals2001quantum} and \cite{farhi1998limit}, the parity of $N$ bits requires $\frac{N}{2}$ bits of $s$ to compute within error $\frac{1}{4}$. Therefore, if we only learn fewer than $\frac{N}{2}$ bits of $s$,
\begin{align}\label{eqn:trace-distance-1}
D(\rho_{\anc},\rho_{s})\geq\frac{1}{4},
\end{align}
where $D(\rho,\sigma) := \frac{1}{2}\|\rho-\sigma\|_1$ is the trace distance.

On the other hand, we claim that the simulation accurately produces $\rho_{s}$. To see this, let $q:=\<1-\textsf{PARITY}(s)|\rho_{s,\textrm{sim}}|1-\textsf{PARITY}(s)\>$ and $r:=\<\textsf{PARITY}(s)|\rho_{s,\textrm{sim}}|1-\textsf{PARITY}(s)\>$. By \eqn{parity-p-q}, $q\leq p$. Write $\rho_{s,\textrm{sim}}$ in the basis $\{|\textsf{PARITY}(s)\>,|1-\textsf{PARITY}(s)\>\}$:
\begin{align}
\rho_{s,\textrm{sim}}=\left(
      \begin{array}{cc}
        1-q & r \\
        r^{*} & q
      \end{array} \right).
\end{align}
Since $\rho_{s,\textrm{sim}}\geq 0$, we have $|r|^{2}\leq q(1-q)$. Consequently,
\begin{align}
D(\rho_{s},\rho_{s,\textrm{sim}})&=\left\|\left(
      \begin{array}{cc}
        -q & r \\
        r^{*} & q
      \end{array} \right)\right\|_{1}=\sqrt{q^{2}+|r|^{2}}\leq \sqrt{q^{2}+q(1-q)}\leq \sqrt{p}\leq\frac{1}{8}, \label{eqn:trace-distance-2}
\end{align}
where the last inequality comes from \eqn{not-parity-s-bound}.

Therefore, by \eqn{trace-distance-1}, \eqn{trace-distance-2}, and the triangle inequality,
\begin{align}\label{eqn:no-fast-forward-separation}
D(\rho_{\anc},\rho_{s,\textrm{sim}})\geq D(\rho_{\anc},\rho_{s})-D(\rho_{s},\rho_{s,\textrm{sim}})\geq\frac{1}{4}-\frac{1}{8}=\frac{1}{8}.
\end{align}
A nonzero element of the matrix $A_{s}$ of $\mathcal{L}_{s}$ defined in \eqn{GKS'} has row coordinate $\big((n_{1}-1)+(N+1)(j_{1}\oplus s_{n_{1}}),n_{1}+(N+1)j_{1}\big)$ and column coordinate $\big((n_{2}-1)+(N+1)(j_{2}\oplus s_{n_{2}}),n_{2}+(N+1)j_{2}\big)$ for some $n_{1},n_{2}\in\range{N+1}$ and $j_{1},j_{2}\in\{0,1\}$. As a result, using the first query model considered in \thm{no-fast-forward}, one query to the Lindbladian can be simulated using at most two queries to the bits of $s$ ($s_{n_{1}}$ and $s_{n_{2}}$). With the second query model, one query to the Lindbladian can be simulated using a single query to a bit of $s$ because a nonzero element of the Lindblad operator $L_{s}$ has row coordinate $(n-1)+(N+1)(j\oplus s_{n})$ and column coordinate $n+(N+1)j$ for some $n\in\range{N+1}$ and $j\in\{0,1\}$. Therefore, \eqn{no-fast-forward-separation} shows that at least $\frac{N}{4}=\frac{t}{16}$ queries to the first black box, or at least $\frac{N}{2}=\frac{t}{8}$ queries to the second black box, are needed to simulate $\mathcal{L}_{s}$ for time $t=4N$ within precision $\frac{1}{8}$.
\end{proof}

\section{Conclusion and future work}
\label{sec:conclusion}

In this paper, we have developed quantum algorithms for efficiently simulating Lindbladians that are not necessarily local.  We introduced two approaches to this problem, one based on implementing sparse Stinespring isometries and another based on implementing short-time evolutions generated by sparse Lindblad operators.

We hope that future work will expand the scope of our algorithms to simulate broader and more unified classes of sparse Lindbladians.  We also hope these tools will prove useful for simulating realistic quantum systems and for developing novel quantum algorithms.

A concrete challenge for future work is to efficiently simulate the Davies master equation \cite{davies1974markovian}. For any given Hamiltonian and temperature, the Davies master equation describes a Markovian open system that converges to the Gibbs state of the Hamiltonian at that temperature.  While the quantum Metropolis algorithm \cite{temme2011quantum,yung2012quantum} can be used to prepare this Gibbs state, simulating the Davies equation would provide a method of simulating the approach to equilibrium. Note that the overcomplete GKS matrix for the Davies master equation is diagonal in the eigenbasis of its system Hamiltonian, but it is not obvious how to apply our methods in that basis.

Another natural question is to improve the complexity of simulating open quantum systems as a function of the allowed error, $\epsilon$.  The simulations presented in this paper have complexity $\poly(1/\epsilon)$, and it is natural to ask whether there is an algorithm with complexity $\poly(\log 1/\epsilon)$, as in the case of Hamiltonian simulation \cite{berry2014exponential}. After an initial version of our work was made public, reference \cite{cleve2016efficient} gave an algorithm with complexity $\poly(\log 1/\epsilon)$ for Lindbladians with sparse Lindblad operators. It remains open to find algorithms with complexity $\poly(\log 1/\epsilon)$ for Lindbladians implemented by sparse Stinespring isometries (such as the Lindbladians discussed in \sec{main}).

\section*{Acknowledgments}

We thank Richard Cleve for discussions that helped to motivate the simulation method presented in \sec{short-time}. We thank an anonymous reviewer for pointing out that third- and higher-order product formulas cannot be used for Lindbladian simulation. We also thank Marina Radulaski and Mark Wilde for identifying a typo in an earlier version of \prop{trotter-formula-second}. We acknowledge support from the Canadian Institute for Advanced Research and the National Science Foundation (grant 1526380).


\begin{appendix}

\section{Gram matrix decomposition for low-rank cases}
\label{append:low-rank-Gram}

In \sec{local}, we described how to simulate local Lindbladians using the sparse Stinespring isometry framework.  That simulation relied on an algorithm for computing the Gram vector corresponding to a given basis state, which we describe in detail here.

\begin{lemma}
Let $M$ be an $N\times N$ Gram matrix with rank $r$.  Suppose we are given a black box for the entries of $M$ and are told the coordinates of a full-rank principal submatrix.  Then \algo{low-rank-Gram} takes $x\in\range{N}$ as input and prepares the Gram vector $v_{x}$ of $M$ as a quantum state using $O(r^{3})$ time and $O(r^{2})$ queries to the entries of $M$ (in particular, both bounds are independent of $N$).
\end{lemma}

\begin{algorithm}
\DontPrintSemicolon
\SetKwInput{Input}{Input}
\SetKwInput{Output}{Output}
\Input{The rank $r$ of $M$; an index $\emph{x}\in[N]$; $\emph{S}\subseteq[N]$, where $|S|=r$ and the rows and columns specified by $S$ constitute a full-rank principal submatrix of $M$.}
\Output{Quantum state $|v_{x}\>$ proporitional to the $x$th Gram vector of $M$.}
	Let $M_{S}$ denote the principal submatrix of the rows and columns in $S$. Diagonalize it as $M_{S}=UDU^{\dagger}$. Denote $(v_{1},v_{2},\ldots,v_{r})^{\dagger}:=U\sqrt{D}$.\;
	\eIf{$x\in S$}{
        Let $v := v_{x}$.\;}
     {Let $M_{S\cup\{x\}}$ be the principal submatrix of the rows and columns in $S\cup\{x\}$. Diagonalize it as $M_{S\cup\{x\}}=U' \diag(d_1,\ldots,d_r,0)U'^{\dagger}$. Denote $$(v_{1}',v_{2}',\ldots,v_{r}',v_{x}')^{\dagger}:=U'\left(
      \begin{array}{ccc}
        \sqrt{d_{1}} & & \\
         & \ddots & \\
         & & \sqrt{d_{r}} \\
         & & 0
      \end{array} \right).$$\;
      \vspace{-12pt}
      Find the unitary $U_{x}$ such that $U_{x}v_{i}'=v_{i}\ \forall\, i\in\{1,\ldots,r\}$. Let $v:=U_{x}v_{x}'$.\; }
      Prepare a quantum state $|v\>$ proportional to $v$ using the method of reference \cite{shende2005synthesis}.\;
\caption{Preparing a Gram vector of a low-rank Gram matrix}
\label{algo:low-rank-Gram}
\end{algorithm}

\begin{proof}
\algo{low-rank-Gram} makes at most $(r+1)^{2}=O(r^{2})$ queries to the entries of $M$ because we only ask for $M_{S\cup\{x\}}$. The matrix diagonalization in line 1 and line 5 takes $O(r^{3})$ time, and line 6 can be done in $O(r^{3})$ time by computing the inverse matrix of $(v_{1}',\ldots,v_{r}')$ to get $(c_{1},\ldots,c_{r})^{\dagger}=(v_{1}',\ldots,v_{r}')^{-1}v_{x}'$, which is equivalent to $\sum_{i=1}^{r}c_{i}v_{i}'=v_{x}'$; therefore,
\begin{align}
v_{x}=U_{x}v_{x}'=\sum_{i=1}^{r}c_{i}U_{x}v_{i}'=\sum_{i=1}^{r}c_{i}v_{i}.
\end{align}

Now we show correctness. Without loss of generality, assume $S=\{1,2,\ldots,r\}$. Since $M$ has rank $r$ and $M_{S}$ is a submatrix of $M$ with full-rank $r$, all rows in $M$ can be written as a linear combination of all rows in $S$. Therefore, if we can show that $v_{i}^{\dagger}v_{x}=M_{ix}$ for arbitrary $i\in\{1,\ldots,r\}$ and $x\in\{1,\ldots,N\}$, then we have $v_{y}^{\dagger}v_{x}=M_{yx}$ for arbitrary $x,y\in\{1,\ldots,N\}$, so we have found the correct Gram vectors.

Since
\begin{align}
v_{i}'^{\dagger}v_{x}'=v_{i}^{\dagger}U_{x}U_{x}^{\dagger}v_{x}=v_{i}^{\dagger}v_{x},
\end{align}
it suffices to show that $v_{i}'^{\dagger}v_{x}'=M_{ix}$ for arbitrary $i\in\{1,\ldots,r\}$ and $x\in\{1,\ldots,N\}$. But this is trivial because line 5 of the algorithm promises
\begin{align}
&(v_{1}',v_{2}',\ldots,v_{r}',v_{x}')^{\dagger}(v_{1}',v_{2}',\ldots,v_{r}',v_{x}') \nonumber \\
&\quad=U'\left(
      \begin{array}{ccc}
        \sqrt{d_{1}} & & \\
         & \ddots & \\
         & & \sqrt{d_{r}} \\
         & & 0
      \end{array} \right)\left(
      \begin{array}{cccc}
        \sqrt{d_{1}} & & & \\
         & \ddots & & \\
         & & \sqrt{d_{r}} & 0 \\
      \end{array} \right)U'^{\dagger} \\
&\quad=U' \diag(d_1,\ldots,d_r,0) U'^{\dagger} \\
&\quad=M_{S\cup\{x\}},
\end{align}
which gives $v_{i}'^{\dagger}v_{x}'=M_{ix}$.

Finally, by reference \cite{shende2005synthesis} we can prepare $v_{x}$ as a quantum state in $O(r)$ additional time, which is dominated by the $O(r^{3})$ time of computing $v_{x}$.
\end{proof}

\section{Decomposition of sparse diagonal Lindbladians}\label{append:sparse-diagonal}

In this appendix we prove \lem{a-decomposition}, which states that every $d$-sparse-diagonal Lindbladian can be written as the sum of at most $3d^2$ strongly 1-sparse-diagonal Lindbladians, and that this decomposition can be found with constant overhead of queries using the black box in \thm{d-sparse}.

We may assume that the diagonal elements of the matrix $a$ characterizing the diagonal Lindbladian are all 0 because, as we have shown in \lem{1-sparse-a}, we can efficiently simulate a diagonal Lindbladian with diagonal matrix $a$, and efficient Lindbladian simulation is closed under positive linear combination by \prop{trotter-formula-second}.

First we show a decomposition without demanding strong sparsity, as follows.

\begin{lemma}\label{lem:d-to-1}
Every $d$-sparse-diagonal Lindbladian can be written as the sum of at most $d^{2}$ 1-sparse-diagonal Lindbladians with constant overhead in queries using the black box in \thm{d-sparse}.
\end{lemma}

\begin{proof}
We cast the problem in terms of a directed graph $G=(V,E)$ where $V=\{v_{1},v_{2},\dots,v_{N}\}$. If $a_{ij}\neq 0$ in the $a$ matrix, then there is a directed edge in $E$ that goes from $v_{i}$ to $v_{j}$. Since $a$ is $d$-sparse, each vertex in $V$ has in-degree and out-degree at most $d$. Our aim is to color the edges of the graph with at most $d^{2}$ colors such that in the subgraph including only the edges of any given color, each vertex has in-degree and out-degree at most $1$ in that subgraph (that is, the subgraph of any given color corresponds to a 1-sparse $a$).

Let $(u,v)$ denote the directed edge from vertex $u$ to vertex $v$. For each $i\in\{1,\dots,N\}$, let $\num(v_{i}):=i$. Define
\begin{align}
\Next(u)&:=\{v:(u,v)\in E\}, \\
\Prev(v)&:=\{u:(u,v)\in E\}.
\end{align}
If $v\in \Next(u)$, define
\begin{align}
\NextIdx(u,v):=|\{i:i<\num(v),v_{i}\in \Next(u)\}|;
\end{align}
if $u\in \Next(v)$, define
\begin{align}
\PrevIdx(u,v):=|\{i:i<\num(u),u_{i}\in \Prev(v)\}|.
\end{align}
For each pair of $(u,v)\in E$, since $\num(v)$ is excluded in $\NextIdx(u,v)$ and $\num(u)$ is excluded in $\PrevIdx(u,v)$, we have
\begin{align}\label{eqn:lem-bound}
0\leq \NextIdx(u,v),\PrevIdx(u,v)\leq d-1.
\end{align}
Then we assign the edge $(u,v)$ the color $d\cdot \NextIdx(u,v)+\PrevIdx(u,v)$.

We claim that for the edges in the subgraph of any given color, each vertex has in-degree at most 1. If not, say $(u,v)$ and $(u',v)$ are edges in the same subgraph where $u\neq u'$. This gives $d\cdot \NextIdx(u,v)+\PrevIdx(u,v)=d\cdot \NextIdx(u',v)+\PrevIdx(u',v)$, which by \eqn{lem-bound} leads to $\NextIdx(u,v)=\NextIdx(u',v)$ and $\PrevIdx(u,v)=\PrevIdx(u',v)$. From the latter, we get $u=u'$, a contradiction. By an analogous argument, the out-degree of each vertex in a given subgraph is at most 1. Therefore, each subgraph corresponds to a 1-sparse-diagonal Lindbladian. By \eqn{lem-bound}, we get at most $d^{2}$ subgraphs.

For each edge $(u,v)\in E$, the coloring only needs to ask the black box stated in \thm{d-sparse} about the neighbors of $u$ and $v$, so the decomposition only costs constant overhead in queries to the black box.  The resulting black box for each 1-sparse-diagonal Lindbladian takes a color and a vertex as input and outputs the edge with the color that is adjacent to the vertex.
\end{proof}

It remains to decompose each 1-sparse-diagonal term into strongly 1-sparse-diagonal pieces.

\begin{lemma}\label{lem:1-to-pairwise}
Every $1$-sparse-diagonal Lindbladian can be written as the sum of at most 3 strongly 1-sparse-diagonal Lindbladians with constant overhead in queries using the black box of the $1$-sparse-diagonal Lindbladian stated in \lem{d-to-1}.
\end{lemma}

\begin{proof}
Again consider a directed graph $G=(V,E)$ where $V=\{v_{1},v_{2},\dots,v_{N}\}$ and $E$ contains a directed edge from $v_i$ to $v_j$ if $a_{ij}\neq 0$. Since $a$ is 1-sparse, each vertex in $V$ has in-degree and out-degree at most 1. Our aim is to 3-color the edges such that for the subgraph of any given color, each component is either an isolated vertex, a directed edge, or a directed cycle of length 2 (that is, each subgraph corresponds to a strongly 1-sparse $a$).

Observe that any component of $G$ must be a directed path, a directed cycle, or an isolated vertex. For directed paths, we can color the edges with two alternating colors. For directed cycles with length 2, we can assign both edges the same color.  For directed cycles with even length at least 4, we can color the edges with alternating colors.  For directed cycles with odd length at least 3, we can color the edges with alternating colors, except that the final edge must be assigned a third color.

These assignments can be made using only three colors, and each resulting subgraph corresponds to a strongly 1-sparse-diagonal Lindbladian.

For each edge $(u,v)\in E$, the coloring only needs to ask the black box stated in \lem{d-to-1} about the neighbors of $u$ and $v$, so the decomposition only incurs constant overhead in queries to the black box, and it produces a black box for each resulting strongly 1-sparse-diagonal Lindbladian as stated in \lem{1-sparse-a}.
\end{proof}

\lem{a-decomposition} follows by directly combining \lem{d-to-1} and \lem{1-to-pairwise}.

\section{Proofs of lemmas from \sec{k-sparse-generator}}

\subsection{Proof of \lem{E-first-order-approximation}}
\label{sec:k-isometry-part-1}

\begin{proof}
Since $L|j\>=\sum_{i=1}^{k}c_{j,i}|\nu_{i}(j)\>$ for all $j\in\range{N}$,
\begin{align}
L=\sum_{j=0}^{N-1}\sum_{i=1}^{k}c_{j,i}|\nu_{i}(j)\>\<j|;\qquad L^{\dagger}=\sum_{j=0}^{N-1}\sum_{i=1}^{k}c_{j,i}^{*}|j\>\<\nu_{i}(j)|.
\end{align}
Thus for any $m,n\in\range{N}$,
\begin{align}
(1+\epsilon\mathcal{L})(|m\>\<n|)&=-\frac{\epsilon}{2}\sum_{i,j=1}^{k}\big(c_{m,i}c_{\nu_{j}^{-1}(\nu_{i}(m)),j}^{*}|\nu_{j}^{-1}(\nu_{i}(m))\>\<n| + c_{n,i}^{*}c_{\nu_{j}^{-1}(\nu_{i}(n)),j}|m\>\<\nu_{j}^{-1}(\nu_{i}(n))|\big) \nonumber \\
&\quad +\epsilon\sum_{i,j=1}^{k}c_{m,i}c_{n,j}^{*}|\nu_{i}(m)\>\<\nu_{j}(n)|+|m\>\<n|. \label{eqn:k-short}
\end{align}

On the other hand,
\begin{align}
\Tr_{\anc}[V_{\epsilon}|m\>\<n|V_{\epsilon}^{\dagger}]&=\frac{\epsilon^{2}}{4}\sum_{1\leq i\neq j\leq k,1\leq i'\neq j'\leq k}c_{m,i}c_{\nu_{j}^{-1}(\nu_{i}(m)),j}^{*}c_{n,i'}^{*}c_{\nu_{j'}^{-1}(\nu_{i'}(n)),j'}|\nu_{j}^{-1}(\nu_{i}(m))\>\<\nu_{j'}^{-1}(\nu_{i'}(n))| \nonumber \\
&\quad-\epsilon\sum_{1\leq i\neq j\leq k}\frac{1}{2}c_{m,i}c_{\nu_{j}^{-1}(\nu_{i}(m)),j}^{*}|\nu_{j}^{-1}(\nu_{i}(m))\>\<n| \nonumber  \\
&\quad-\epsilon\sum_{1\leq i\neq j\leq k}\frac{1}{2}c_{n,i}^{*}c_{\nu_{j}^{-1}(\nu_{i}(n)),j}|m\>\<\nu_{j}^{-1}(\nu_{i}(n))| \nonumber  \\
&\quad+\epsilon\sum_{i,j=1}^{k} c_{m,i}c_{n,j}^{*}|\nu_{i}(m)\>\<\nu_{j}(n)|+|m\>\<n|.
\end{align}
Taking the difference between the above equality and \eqn{k-short}, we have
\begin{align}
&(1+\epsilon\mathcal{L})(|m\>\<n|)-\Tr_{\anc}[V_{\epsilon}|m\>\<n|V_{\epsilon}^{\dagger}] \nonumber \\
&\quad=-\frac{\epsilon^{2}}{4}\sum_{1\leq i\neq j\leq k,1\leq i'\neq j'\leq k}c_{m,i}c_{\nu_{j}^{-1}(\nu_{i}(m)),j}^{*}c_{n,i'}^{*}c_{\nu_{j'}^{-1}(\nu_{i'}(n)),j'}|\nu_{j}^{-1}(\nu_{i}(m))\>\<\nu_{j'}^{-1}(\nu_{i'}(n))|. \label{eqn:isometry-single-state}
\end{align}
By subadditivity of the trace norm, it suffices to show that for any pure state
\begin{align}
|\Phi\>=\sum_{m,m'=0}^{N-1}b_{mm'}|m\>|m'\>
\end{align}
with $\sum_{m,m'=0}^{N-1}|b_{mm'}|^{2}=1$, we have
\begin{align}\label{eqn:isometry-single-state-equivalence}
\big\|((1+\epsilon\mathcal{L})\otimes\mathcal{I}_{N})(|\Phi\>\<\Phi|)-\Tr_{\anc}[(V_{\epsilon}\otimes I_{N})|\Phi\>\<\Phi|(V_{\epsilon}^{\dagger}\otimes I_{N})]\big\|_{1}=O(k^{4}\epsilon^{2}).
\end{align}

Indeed,
\begin{align}
&\Big\|\sum_{m,m',n,n'}b_{mm'}b_{nn'}^{*}\sum_{i\neq j,i'\neq j'}c_{m,i}c_{\nu_{j}^{-1}(\nu_{i}(m)),j}^{*}c_{n,i'}^{*}c_{\nu_{j'}^{-1}(\nu_{i'}(n)),j'}|\nu_{j}^{-1}(\nu_{i}(m))\>\<\nu_{j'}^{-1}(\nu_{i'}(n))|\otimes |m'\>\<n'|\Big\|_{1} \nonumber \\
&\quad\leq\sum_{i\neq j,i'\neq j'}\Big(\Big\|\sum_{m,m'}b_{mm'}c_{m,i}c_{\nu_{j}^{-1}(\nu_{i}(m)),j}^{*}|\nu_{j}^{-1}(\nu_{i}(m))\>|m'\>\Big\| \nonumber \\
&\qquad\qquad\qquad\cdot\Big\|\sum_{n,n'}b_{nn'}^{*}c_{n,i'}^{*}c_{\nu_{j'}^{-1}(\nu_{i'}(n)),j'}\<\nu_{j'}^{-1}(\nu_{i'}(n))|\<n'|\Big\|\Big) \\
&\quad=\sum_{i\neq j,i'\neq j'}\sqrt{\sum_{m,m'}|b_{mm'}|^{2}|c_{m,i}|^{2}|c_{\nu_{j}^{-1}\nu_{i}(m),j}|^{2}}\sqrt{\sum_{n,n'}|b_{nn'}|^{2}|c_{n,i'}|^{2}|c_{\nu_{j'}^{-1}\nu_{i'}(n),j'}|^{2}} \\
&\quad\leq\sum_{i\neq j,i'\neq j'}\sqrt{\sum_{m,m'}|b_{mm'}|^{2}\|L\|_{\max}^{4}}\sqrt{\sum_{n,n'}|b_{nn'}|^{2}\|L\|_{\max}^{4}} \\
&\quad=k^{2}(k-1)^{2},
\end{align}
and by \eqn{isometry-single-state} this implies \eqn{isometry-single-state-equivalence}.
\end{proof}

\subsection{Proof of \lem{square-error-spectral-norm}}
\label{sec:k-isometry-part-2}

\begin{proof}
We have
\begin{align}
V_{\epsilon}^{\dagger}V_{\epsilon}&=\sum_{m,n}|n\>\<m|\Big(\sum_{i=1}^{k}\sqrt{\epsilon}c_{n,i}^{*}\<\nu_{i}(n)|\<1|+\<n|\<0|-\epsilon\sum_{i_{1}\neq i_{2}}\frac{1}{2}c_{n,i_{1}}^{*}\cdot c_{\nu_{i_{2}}^{-1}(\nu_{i_{1}}(n)),i_{2}}\<\nu_{i_{2}}^{-1}(\nu_{i_{1}}(n))|\<0|\Big) \nonumber \\
&\quad\cdot\Big(\sum_{j=1}^{k}\sqrt{\epsilon}c_{m,j}|\nu_{j}(m)\>|1\>+|m\>|0\>-\epsilon\sum_{j_{1}\neq j_{2}}\frac{1}{2}c_{m,j_{1}}\cdot c_{\nu_{j_{2}}^{-1}(\nu_{j_{1}}(m)),i_{2}}^{*}|\nu_{j_{2}}^{-1}(\nu_{i_{1}}(m))\>|0\>\Big) \\
&=\sum_{m,n}|n\>\<m|\Big(\frac{\epsilon^{2}}{4}\sum_{\nu_{i_{2}}^{-1}(\nu_{i_{1}}(n))=\nu_{j_{2}}^{-1}(\nu_{j_{1}}(m))}c_{n,i_{1}}^{*}c_{\nu_{i_{2}}^{-1}(\nu_{i_{1}}(n)),i_{2}}c_{m,j_{1}}c_{\nu_{j_{2}}^{-1}(\nu_{j_{1}}(m)),j_{2}}^{*} \nonumber \\
&\quad-\frac{\epsilon}{2}\sum_{\nu_{j_{2}}^{-1}(\nu_{j_{1}}(m))=n}c_{m,j_{1}}c_{\nu_{j_{2}}^{-1}(\nu_{j_{1}}(m)),j_{2}}^{*}-\frac{\epsilon}{2}\sum_{\nu_{i_{2}}^{-1}(\nu_{i_{1}}(n))=m}c_{n,i_{1}}^{*}c_{\nu_{i_{2}}^{-1}(\nu_{i_{1}}(n)),i_{2}} \nonumber \\
&\quad+\epsilon\sum_{\nu_{i}(n)=\nu_{j}(m)}c_{n,i}^{*}c_{m,j}+\delta_{nm}\Big).
\end{align}
Since
\begin{align}
&\epsilon\sum_{\nu_{i}(n)=\nu_{j}(m)}c_{n,i}^{*}c_{m,j}-\frac{\epsilon}{2}\sum_{\nu_{j_{2}}^{-1}(\nu_{j_{1}}(m))=n}c_{m,j_{1}}c_{\nu_{j_{2}}^{-1}(\nu_{j_{1}}(m)),j_{2}}^{*}-\frac{\epsilon}{2}\sum_{\nu_{i_{2}}^{-1}(\nu_{i_{1}}(n))=m}c_{n,i_{1}}^{*}c_{\nu_{i_{2}}^{-1}(\nu_{i_{1}}(n)),i_{2}} \nonumber \\
&\quad=\epsilon\sum_{\nu_{i}(n)=\nu_{j}(m)}c_{n,i}^{*}c_{m,j}-\frac{\epsilon}{2}\sum_{\nu_{j_{1}}(m)=\nu_{j_{2}}(n)}c_{m,j_{1}}c_{n,j_{2}}^{*}-\frac{\epsilon}{2}\sum_{\nu_{i_{1}}(n)=\nu_{i_{2}}(m)}c_{n,i_{1}}^{*}c_{m,i_{2}} \\
&\quad=\epsilon\sum_{\nu_{i}(n)=\nu_{j}(m)}c_{n,i}^{*}c_{m,j}-\frac{\epsilon}{2}\sum_{\nu_{i}(n)=\nu_{j}(m)}c_{n,i}^{*}c_{m,j}-\frac{\epsilon}{2}\sum_{\nu_{i}(n)=\nu_{j}(m)}c_{n,i}^{*}c_{m,j} \\
&\quad=0,
\end{align}
we know that
\begin{align}\label{eqn:difference-with-identity}
V_{\epsilon}^{\dagger}V_{\epsilon}-I=\frac{\epsilon^{2}}{4}\sum_{m,n}|n\>\<m|\sum_{\nu_{i_{2}}^{-1}(\nu_{i_{1}}(n))=\nu_{j_{2}}^{-1}(\nu_{j_{1}}(m))}c_{n,i_{1}}^{*}c_{\nu_{i_{2}}^{-1}(\nu_{i_{1}}(n)),i_{2}}c_{m,j_{1}}c_{\nu_{j_{2}}^{-1}(\nu_{j_{1}}(m)),j_{2}}^{*}.
\end{align}

To prove that $\|V_{\epsilon}^{\dagger}V_{\epsilon}-I\|=O(k^{4}\epsilon^{2})$, it suffices to show that for any pure state $|\psi\>=\sum_{m=0}^{N-1}b_{m}|m\>$ where $\sum_{m=0}^{N-1}|b_{m}|^{2}=1$,
\begin{align}
\<\psi|V_{\epsilon}^{\dagger}V_{\epsilon}-I|\psi\>=O(k^{4}\epsilon^{2}).
\end{align}
Denote $l:=\nu_{i_{2}}^{-1}(\nu_{i_{1}}(n))=\nu_{j_{2}}^{-1}(\nu_{j_{1}}(m))$. Then $m=\nu_{j_{1}}^{-1}(\nu_{j_{2}}(l))$ and $n=\nu_{i_{1}}^{-1}(\nu_{i_{2}}(l))$. Plugging these into \eqn{difference-with-identity}, we get
\begin{align}
\<\psi|V_{\epsilon}^{\dagger}V_{\epsilon}-I|\psi\>&=\frac{\epsilon^{2}}{4}\sum_{l=0}^{N-1}\big|\sum_{i_{1},i_{2}=1}^{k}b_{\nu_{i_{1}}^{-1}(\nu_{i_{2}}(l))}c_{\nu_{i_{1}}^{-1}(\nu_{i_{2}}(l)),i_{1}}^{*}c_{l,i_{2}}\big|^{2} \\
&\leq\frac{\epsilon^{2}}{4}\sum_{l=0}^{N-1}\big(\sum_{i_{1},i_{2}=1}^{k}|b_{\nu_{i_{1}}^{-1}(\nu_{i_{2}}(l))}|^{2}\big)\big(\sum_{i_{1},i_{2}=1}^{k}|c_{\nu_{i_{1}}^{-1}(\nu_{i_{2}}(l)),i_{1}}^{*}c_{l,i_{2}}|^{2}\big) \\
&\leq\frac{\epsilon^{2}}{4} k^{2}\|L\|_{\max}^{4}\sum_{l=0}^{N-1}\sum_{i_{1},i_{2}=1}^{k}|b_{\nu_{i_{1}}^{-1}(\nu_{i_{2}}(l))}|^{2} \\
&=\frac{\epsilon^{2}}{4} k^{2}\sum_{i_{1},i_{2}=1}^{k}\sum_{l=0}^{N-1}|b_{\nu_{i_{1}}^{-1}(\nu_{i_{2}}(l))}|^{2} \\
&=\frac{k^{4}\epsilon^{2}}{4} \\
&= O(k^{4}\epsilon^{2}),
\end{align}
where the first inequality comes from the Cauchy-Schwartz inequality.
\end{proof}
\end{appendix}


\providecommand{\bysame}{\leavevmode\hbox to3em{\hrulefill}\thinspace}


\begin{thebibliography}{10}

\bibitem{aharonov2003adiabatic}
Dorit Aharonov and Amnon Ta-Shma, \emph{Adiabatic quantum state generation and
  statistical zero knowledge}, Proceedings of the 35th Annual ACM Symposium on
  Theory of Computing, pp.~20--29, ACM, 2003,
  \mbox{\href{http://arxiv.org/abs/quant-ph/0301023}{arXiv:quant-ph/0301023}}.

\bibitem{alagic2005decoherence}
Gorjan Alagic and Alexander Russell, \emph{Decoherence in quantum walks on the
  hypercube}, Physical Review A \textbf{72} (2005), no.~6, 062304,
  \mbox{\href{http://arxiv.org/abs/quant-ph/0501169}{arXiv:quant-ph/0501169}}.

\bibitem{babbush2015exponentially}
Ryan Babbush, Dominic~W. Berry, Ian~D. Kivlichan, Annie~Y. Wei, Peter~J. Love,
  and Al{\'a}n Aspuru-Guzik, \emph{Exponentially more precise quantum
  simulation of fermions {I}{I}: {Q}uantum chemistry in the {C}{I} matrix
  representation}, 2015,
  \mbox{\href{http://arxiv.org/abs/1506.01029}{arXiv:1506.01029}}.

\bibitem{barenco1995elementary}
Adriano Barenco, Charles~H. Bennett, Richard Cleve, David~P. DiVincenzo, Norman
  Margolus, Peter Shor, Tycho Sleator, John~A. Smolin, and Harald Weinfurter,
  \emph{Elementary gates for quantum computation}, Physical Review A
  \textbf{52} (1995), no.~5, 3457,
  \mbox{\href{http://arxiv.org/abs/quant-ph/9503016}{arXiv:quant-ph/9503016}}.

\bibitem{beals2001quantum}
Robert Beals, Harry Buhrman, Richard Cleve, Michele Mosca, and Ronald de~Wolf,
  \emph{Quantum lower bounds by polynomials}, Journal of the ACM \textbf{48}
  (2001), no.~4, 778--797,
  \mbox{\href{http://arxiv.org/abs/quant-ph/9802049}{arXiv:quant-ph/9802049}}.

\bibitem{berry2007efficient}
Dominic~W. Berry, Graeme Ahokas, Richard Cleve, and Barry~C. Sanders,
  \emph{Efficient quantum algorithms for simulating sparse {Hamiltonians}},
  Communications in Mathematical Physics \textbf{270} (2007), no.~2, 359--371,
  \mbox{\href{http://arxiv.org/abs/quant-ph/0508139}{arXiv:quant-ph/0508139}}.

\bibitem{BC12}
Dominic~W. Berry and Andrew~M. Childs, \emph{Black-box {H}amiltonian simulation
  and unitary implementation}, Quantum Information and Computation \textbf{12}
  (2012), no.~1-2, 29--62,
  \mbox{\href{http://arxiv.org/abs/0910.4157}{arXiv:0910.4157}}.

\bibitem{berry2014exponential}
Dominic~W. Berry, Andrew~M. Childs, Richard Cleve, Robin Kothari, and
  Rolando~D. Somma, \emph{Exponential improvement in precision for simulating
  sparse {Hamiltonians}}, Proceedings of the 46th Annual ACM Symposium on
  Theory of Computing, pp.~283--292, 2014,
  \mbox{\href{http://arxiv.org/abs/1312.1414}{arXiv:1312.1414}}.

\bibitem{berry2015simulating}
Dominic~W. Berry, Andrew~M. Childs, Richard Cleve, Robin Kothari, and
  Rolando~D. Somma, \emph{Simulating {Hamiltonian} dynamics with a truncated {Taylor}
  series}, Physical Review Letters \textbf{114} (2015), no.~9, 090502,
  \mbox{\href{http://arxiv.org/abs/1412.4687}{arXiv:1412.4687}}.

\bibitem{berry2015hamiltonian}
Dominic~W. Berry, Andrew~M. Childs, and Robin Kothari, \emph{Hamiltonian
  simulation with nearly optimal dependence on all parameters}, Proceedings of
  the 56th IEEE Symposium on Foundations of Computer Science, pp.~792--809,
  2015, \mbox{\href{http://arxiv.org/abs/1501.01715}{arXiv:1501.01715}}.

\bibitem{BN16}
Dominic~W. Berry and Leonardo Novo, \emph{Corrected quantum walk for optimal
  {H}amiltonian simulation}, 2016,
  \mbox{\href{http://arxiv.org/abs/1606.03443}{arXiv:1606.03443}}.

\bibitem{Chi04}
Andrew~M. Childs, \emph{Quantum information processing in continuous time},
  Ph.D. thesis, Massachusetts Institute of Technology, 2004.

\bibitem{Chi10}
Andrew~M. Childs, \emph{On the relationship between continuous- and discrete-time
  quantum walk}, Communications in Mathematical Physics \textbf{294} (2010),
  no.~2, 581--603,
  \mbox{\href{http://arxiv.org/abs/0810.0312}{arXiv:0810.0312}}.

\bibitem{CCDFGS03}
Andrew~M. Childs, Richard Cleve, Enrico Deotto, Edward Farhi, Sam Gutmann, and
  Daniel~A. Spielman, \emph{Exponential algorithmic speedup by quantum walk},
  Proceedings of the 35th ACM Symposium on Theory of Computing, pp.~59--68,
  2003,
  \mbox{\href{http://arxiv.org/abs/quant-ph/0209131}{arXiv:quant-ph/0209131}}.

\bibitem{childs2010simulating}
Andrew~M. Childs and Robin Kothari, \emph{Simulating sparse {Hamiltonians} with
  star decompositions}, Theory of Quantum Computation, Communication, and
  Cryptography, Springer, pp.~94--103, 2010,
  \mbox{\href{http://arxiv.org/abs/1003.3683}{arXiv:1003.3683}}.

\bibitem{childs2012hamiltonian}
Andrew~M. Childs and Nathan Wiebe, \emph{Hamiltonian simulation using linear
  combinations of unitary operations}, Quantum Information \& Computation
  \textbf{12} (2012), no.~11-12, 901--924,
  \mbox{\href{http://arxiv.org/abs/1202.5822}{arXiv:1202.5822}}.

\bibitem{cleve2016efficient}
Richard Cleve and Chunhao Wang, \emph{Efficient quantum algorithms for
  simulating {L}indblad evolution}, 2016, \mbox{\href{http://arxiv.org/abs/1612.09512}{arXiv:1612.09512}}.

\bibitem{davies1974markovian}
Edward~Brian Davies, \emph{Markovian master equations}, Communications in
  Mathematical Physics \textbf{39} (1974), no.~2, 91--110.

\bibitem{FGG07}
Edward Farhi, Jeffrey Goldstone, and Sam Gutmann, \emph{A quantum algorithm for
  the {H}amiltonian {NAND} tree}, Theory of Computing \textbf{4} (2008), no.~1,
  169--190,
  \mbox{\href{http://arxiv.org/abs/quant-ph/0702144}{arXiv:quant-ph/0702144}}.

\bibitem{farhi1998limit}
Edward Farhi, Jeffrey Goldstone, Sam Gutmann, and Michael Sipser, \emph{Limit
  on the speed of quantum computation in determining parity}, Physical Review
  Letters \textbf{81} (1998), no.~24, 5442,
  \mbox{\href{http://arxiv.org/abs/quant-ph/9802045}{arXiv:quant-ph/9802045}}.

\bibitem{feynman1982simulating}
Richard~P. Feynman, \emph{Simulating physics with computers}, International
  Journal of Theoretical Physics \textbf{21} (1982), no.~6, 467--488.

\bibitem{gorini1976completely}
Vittorio Gorini, Andrzej Kossakowski, and Ennackal C.~G. Sudarshan,
  \emph{Completely positive dynamical semigroups of {N}-level systems}, Journal
  of Mathematical Physics \textbf{17} (1976), no.~5, 821--825.

\bibitem{grover2002creating}
Lov Grover and Terry Rudolph, \emph{Creating superpositions that correspond to
  efficiently integrable probability distributions}, 2002,
  \mbox{\href{http://arxiv.org/abs/quant-ph/0208112}{arXiv:quant-ph/0208112}}.

\bibitem{hardy1927collected}
Godfrey~H. Hardy, P.~Seshu~Aiyar, and Bertram~M. Wilson, \emph{Collected papers
  of {S}rinivasa {R}amanujan}, AMS, 1927.

\bibitem{HHL09}
Aram~W. Harrow, Avinatan Hassidim, and Seth Lloyd, \emph{Quantum algorithm for
  linear systems of equations}, Physical Review Letters \textbf{103} (2009),
  no.~15, 150502,
  \mbox{\href{http://arxiv.org/abs/0811.3171}{arXiv:0811.3171}}.

\bibitem{jordan2009efficient}
Stephen~P. Jordan and Pawel Wocjan, \emph{Efficient quantum circuits for
  arbitrary sparse unitaries}, Physical Review A \textbf{80} (2009), no.~6,
  062301, \mbox{\href{http://arxiv.org/abs/0904.2211}{arXiv:0904.2211}}.

\bibitem{kitaev1997quantum}
Alexei~Y. Kitaev, \emph{Quantum computations: Algorithms and error correction},
  Russian Mathematical Surveys \textbf{52} (1997), no.~6, 1191--1249.

\bibitem{kliesch2011dissipative}
Martin Kliesch, Thomas Barthel, Christian Gogolin, Michael Kastoryano, and Jens
  Eisert, \emph{Dissipative quantum {Church}-{Turing} theorem}, Physical Review
  Letters \textbf{107} (2011), no.~12, 120501,
  \mbox{\href{http://arxiv.org/abs/1105.3986}{arXiv:1105.3986}}.

\bibitem{kraus2008preparation}
Barbara Kraus, Hans~P. B{\"u}chler, Sebastian Diehl, Adrian Kantian, Andrea
  Micheli, and Peter Zoller, \emph{Preparation of entangled states by quantum
  {Markov} processes}, Physical Review A \textbf{78} (2008), no.~4, 042307,
  \mbox{\href{http://arxiv.org/abs/0803.1463}{arXiv:0803.1463}}.

\bibitem{lindblad1976generators}
Goran Lindblad, \emph{On the generators of quantum dynamical semigroups},
  Communications in Mathematical Physics \textbf{48} (1976), no.~2, 119--130.

\bibitem{liu2016continuous}
Chaobin Liu and Radhakrishnan Balu, \emph{Continuous-time open quantum walks}, 2016,
  \mbox{\href{http://arxiv.org/abs/1604.05652}{arXiv:1604.05652}}.

\bibitem{lloyd1996universal}
Seth Lloyd, \emph{Universal quantum simulators}, Science \textbf{273} (1996),
  no.~5278, 1073.

\bibitem{low2016optimal}
Guang~Hao Low and Isaac~L. Chuang, \emph{Optimal {Hamiltonian} simulation by
  quantum signal processing}, 2016,
  \mbox{\href{http://arxiv.org/abs/1606.02685}{arXiv:1606.02685}}.

\bibitem{shende2005synthesis}
Vivek~V. Shende, Stephen~S. Bullock, and Igor~L. Markov, \emph{Synthesis of
  quantum logic circuits}, Proceedings of the 2005 Asia and South Pacific
  Design Automation Conference, pp.~272--275, 2005,
  \mbox{\href{http://arxiv.org/abs/quant-ph/0406176}{arXiv:quant-ph/0406176}}.

\bibitem{strauch2009reexamination}
Frederick~W. Strauch, \emph{Reexamination of decoherence in quantum walks on
  the hypercube}, Physical Review A \textbf{79} (2009), no.~3, 032319,
  \mbox{\href{http://arxiv.org/abs/0808.3403}{arXiv:0808.3403}}.

\bibitem{suzuki1990fractal}
Masuo Suzuki, \emph{Fractal decomposition of exponential operators with
  applications to many-body theories and {Monte} {Carlo} simulations}, Physics
  Letters A \textbf{146} (1990), no.~6, 319--323.

\bibitem{suzuki1991general}
Masuo Suzuki, \emph{General theory of fractal path integrals with applications to
  many-body theories and statistical physics}, Journal of Mathematical Physics
  \textbf{32} (1991), no.~2, 400--407.

\bibitem{sweke2016digital}
Ryan Sweke, Mikel Sanz, Ilya Sinayskiy, Francesco Petruccione, and Enrique
  Solano, \emph{Digital quantum simulation of many-body non-{Markovian}
  dynamics}, 2016,
  \mbox{\href{http://arxiv.org/abs/1604.00203}{arXiv:1604.00203}}.

\bibitem{sweke2014simulation}
Ryan Sweke, Ilya Sinayskiy, and Francesco Petruccione, \emph{Simulation of
  single-qubit open quantum systems}, Physical Review A \textbf{90} (2014),
  no.~2, 022331,
  \mbox{\href{http://arxiv.org/abs/1405.6049}{arXiv:1405.6049}}.

\bibitem{temme2011quantum}
Kristan Temme, Tobias~J. Osborne, Karl~G. Vollbrecht, David Poulin, and Frank
  Verstraete, \emph{Quantum {Metropolis} sampling}, Nature \textbf{471} (2011),
  no.~7336, 87--90,
  \mbox{\href{http://arxiv.org/abs/0911.3635}{arXiv:0911.3635}}.

\bibitem{terhal2000problem}
Barbara~M. Terhal and David~P. DiVincenzo, \emph{Problem of equilibration and
  the computation of correlation functions on a quantum computer}, Physical
  Review A \textbf{61} (2000), no.~2, 022301,
  \mbox{\href{http://arxiv.org/abs/quant-ph/9810063}{arXiv:quant-ph/9810063}}.

\bibitem{verstraete2009quantum}
Frank Verstraete, Michael~M. Wolf, and Juan~Ignacio Cirac, \emph{Quantum
  computation and quantum-state engineering driven by dissipation}, Nature
  Physics \textbf{5} (2009), no.~9, 633--636.

\bibitem{wang2011quantum}
Hefeng Wang, Sahel Ashhab, and Franco Nori, \emph{Quantum algorithm for
  simulating the dynamics of an open quantum system}, Physical Review A
  \textbf{83} (2011), no.~6, 062317,
  \mbox{\href{http://arxiv.org/abs/1103.3377}{arXiv:1103.3377}}.

\bibitem{watrous2011theory}
John Watrous, \emph{Theory of quantum information}, 2011.

\bibitem{werner2016positive}
A.~H. Werner, D.~Jaschke, P.~Silvi, M.~Kliesch, T.~Calarco, J.~Eisert, and
  S.~Montangero, \emph{A positive tensor network approach for simulating open
  quantum many-body systems}, Physical Review Letters \textbf{116} (2016),
  no.~23, 237201,   \mbox{\href{http://arxiv.org/abs/1412.5746}{arXiv:1412.5746}}.

\bibitem{whitfield2010quantum}
James~D. Whitfield, C{\'e}sar~A. Rodr{\'\i}guez-Rosario, and Al{\'a}n
  Aspuru-Guzik, \emph{Quantum stochastic walks: A generalization of classical
  random walks and quantum walks}, Physical Review A \textbf{81} (2010), no.~2,
  022323, \mbox{\href{http://arxiv.org/abs/0905.2942}{arXiv:0905.2942}}.

\bibitem{yung2012quantum}
Man-Hong Yung and Al{\'a}n Aspuru-Guzik, \emph{A quantum--quantum {Metropolis}
  algorithm}, Proceedings of the National Academy of Sciences \textbf{109}
  (2012), no.~3, 754--759,
  \mbox{\href{http://arxiv.org/abs/1011.1468}{arXiv:1011.1468}}.

\end{thebibliography}
\end{document}